\documentclass[sigconf,table,dvipsnames]{acmart}

\usepackage{graphicx}
\usepackage{booktabs}
\usepackage{listings}
\usepackage{enumitem}
\usepackage{tabularx}
\usepackage{bm}
\usepackage{clrscode3e}
\usepackage[noend]{algpseudocode}
\usepackage{colortbl}
\usepackage{forest}
\usepackage{algorithm}
\usepackage{xspace}
\usepackage{url}
\usepackage{textcomp}
\usepackage{multirow}
\usepackage{booktabs}
\usepackage{hyperref}
\usepackage{color}
\usepackage{tcolorbox}
\usepackage{xpatch}
\usepackage[skip=0.2pt]{subcaption}
\usepackage[nounderscore]{syntax}
\usepackage{tikz-qtree,tikz-qtree-compat}
\usepackage{tikz}
\usetikzlibrary{arrows.meta,bending,quotes}
\usepackage{ulem}
\usepackage{standalone}
\usepackage{enumitem}
\usepackage{stfloats }

\tikzset{every tree node/.style={align=center, anchor=north}}
\AtBeginDocument{%
  \providecommand\BibTeX{{%
    \normalfont B\kern-0.5em{\scshape i\kern-0.25em b}\kern-0.8em\TeX}}}

\makeatletter
\let\@authorsaddresses\@empty
\makeatother

\newcommand{\paratitle}[1]{\noindent{\bf #1}}

\setcopyright{acmcopyright}
\copyrightyear{2018}
\acmYear{2018}
\acmDOI{10.1145/1122445.1122456}

\acmJournal{JACM}
\acmVolume{37}
\acmNumber{4}
\acmArticle{111}
\acmMonth{8}

\newcommand{\amir}[1]{} 
\newcommand{\zhengjie}[1]{} 

\newcommand{\reva}[1]{{\leavevmode\color{black}{#1}}}
\newcommand{\revb}[1]{{\leavevmode\color{black}{#1}}}
\newcommand{\revc}[1]{{\leavevmode\color{black}{#1}}}
\newcommand{\common}[1]{{\leavevmode\color{black}{#1}}}

\newcommand{\narrow}[1]{\protect\scalebox{0.7}[1.0]{\ensuremath{\textsf{#1}}}}

\newcommand{\sql}[1]{\narrow{#1}}

\newcommand{\Serves}{\sql{Serves}}
\newcommand{\Drinker}{\sql{Drinker}}
\renewcommand{\Bar}{\sql{Bar}}
\newcommand{\Beer}{\sql{Beer}}
\newcommand{\Likes}{\sql{Likes}}
\newcommand{\Frequents}{\sql{Frequents}}

\definecolor{black}{rgb}{0,0,0}
\definecolor{grey}{rgb}{0.8,0.8,0.8}
\definecolor{red}{rgb}{1,0,0}
\definecolor{green}{rgb}{0,1,0}
\definecolor{applegreen}{rgb}{0.55, 0.71, 0.0}
\definecolor{darkgreen}{rgb}{0,0.5,0}
\definecolor{darkpurple}{rgb}{0.5,0,0.5}
\definecolor{darkdarkpurple}{rgb}{0.3,0,0.3}
\definecolor{blue}{rgb}{0,0,1}
\definecolor{shadegreen}{rgb}{0.95,1,0.95}
\definecolor{shadeblue}{rgb}{0.95,0.95,1}
\definecolor{shadered}{rgb}{1,0.85,0.85}
\definecolor{shadegrey}{rgb}{0.85,0.85,0.85}
\definecolor{oddRowGrey}{rgb}{0.80,0.80,0.80}
\definecolor{evenRowGrey}{rgb}{0.85,0.85,0.85}

\newcommand{\red}[1]{{\color{red} #1}}

\newcommand{\cut}[1]{{}}

\newcommand{\SR}[1]{{\color{magenta} {\tt SR: #1}}}
\newcommand{\sr}[1]{\SR{#1}}

\newcommand{\mypar}[1]{\smallskip\noindent\textbf{{#1}.}}

\DeclareMathAlphabet{\mathbbold}{U}{bbold}{m}{n}

\newtheorem{Theorem}{Theorem}
\newtheorem{Definition}{Definition}

\newtheorem{Proposition}{Proposition}
\newtheorem{Example}{Example}

\newcommand{\conjNaive}{\textsc{Conj-Naive}}
\newcommand{\disjEO}{\textsc{Disj-EO}}
\newcommand{\disjAdd}{\textsc{Disj-Add}}
\newcommand{\disjNaive}{\textsc{Disj-Naive}}

\newcommand{\conjEO}{\textsc{Conj-EO}}
\newcommand{\conjAdd}{\textsc{Conj-Add}}

\newcommand{\tup}{t}
\newcommand{\schemaOf}[1]{\ensuremath{\mathbf{#1}}}
\newcommand{\schema}{S}
\newcommand{\ains}{K}
\newcommand{\cinstance}{\ensuremath{\mathcal{I}}}

\newcommand{\db}{D}

\newcommand{\rel}{R}
\newcommand{\relarity}{k}
\newcommand{\relnum}{r}

\newcommand{\query}{Q}

\newcommand{\dom}{\textsc{Dom}}

\newcommand{\varset}{\mathcal{V}}
\newcommand{\lnset}{\mathcal{L}}
\newcommand{\constset}{\mathcal{C}}
\newcommand{\att}{{\tt A}}
\newcommand{\attrs}{{\bf  {\tt Attr}}}
\newcommand{\qvariable}{query variable}
\newcommand{\resolved}[1]{{}}
\newcommand{\aconst}{c}
\newcommand{\atable}{T}
\newcommand{\aCtable}{\textsf{T}}

\newcommand{\universalcond}{\phi}
\newcommand{\globalcond}{\phi}
\newcommand{\assignment}{\alpha}
\newcommand{\map}{\mu}
\newcommand{\Rep}{PWD}
\newcommand{\coverage}{\textsf{cov}}
\newcommand{\cov}{\textsf{C}}
\newcommand{\consistent}{\textsf{IsConsistent}}
\newcommand{\qcorrect}{\ensuremath{Q_{A}}}
\newcommand{\qincorrect}{\ensuremath{Q_{B}}}
\newcommand{\mrel}{\ensuremath{h}}
\newcommand{\aFOL}{\textsc{P}}
\newcommand{\folarity}{p}

\definecolor{lstpurple}{rgb}{0.5,0,0.5}
\definecolor{lstred}{rgb}{1,0,0}
\definecolor{lstreddark}{rgb}{0.7,0,0}
\definecolor{lstredl}{rgb}{0.64,0.08,0.08}
\definecolor{lstmildblue}{rgb}{0.66,0.72,0.78}
\definecolor{lstblue}{rgb}{0,0,1}
\definecolor{lstmildgreen}{rgb}{0.42,0.53,0.39}
\definecolor{lstgreen}{rgb}{0,0.5,0}
\definecolor{lstorangedark}{rgb}{0.6,0.3,0}	
\definecolor{lstorange}{rgb}{0.75,0.52,0.005}
\definecolor{lstorangelight}{rgb}{0.89,0.81,0.67}
\definecolor{lstbeige}{rgb}{0.90,0.86,0.45}

\DeclareFontShape{OT1}{cmtt}{bx}{n}{<5><6><7><8><9><10><10.95><12><14.4><17.28><20.74><24.88>cmttb10}{}

\lstdefinelanguage{smtlib2}{
  alsoletter=-,
  morekeywords={declare-const,define-fun,assert,minimize,maximize,check-sat,get-objectives,and,or,not,distinct},
  extendedchars=false,
  keywordstyle=\bfseries\color{lstpurple},
  deletekeywords={Int,Bool},
  keywords=[2]{Int,Bool},
  keywordstyle=[2]\color{lstblue},
}

\lstdefinestyle{psql}
{
tabsize=2,
basicstyle=\scriptsize\upshape\ttfamily,
language=SQL,
morekeywords={PROVENANCE,BASERELATION,INFLUENCE,COPY,ON,TRANSPROV,TRANSSQL,TRANSXML,CONTRIBUTION,COMPLETE,TRANSITIVE,NONTRANSITIVE,EXPLAIN,SQLTEXT,GRAPH,IS,ANNOT,THIS,XSLT,MAPPROV,cxpath,OF,TRANSACTION,SERIALIZABLE,COMMITTED,INSERT,INTO,WITH,SCN,UPDATED},
extendedchars=false,
keywordstyle=\bfseries,
mathescape=true,
escapechar=@,
sensitive=true
}

\lstdefinestyle{psqlcolor}
{
tabsize=2,
basicstyle=\scriptsize\upshape\ttfamily,
language=SQL,
morekeywords={PROVENANCE,BASERELATION,INFLUENCE,COPY,ON,TRANSPROV,TRANSSQL,TRANSXML,CONTRIBUTION,COMPLETE,TRANSITIVE,NONTRANSITIVE,EXPLAIN,SQLTEXT,GRAPH,IS,ANNOT,THIS,XSLT,MAPPROV,cxpath,OF,TRANSACTION,SERIALIZABLE,COMMITTED,INSERT,INTO,WITH,SCN,UPDATED},
extendedchars=false,
keywordstyle=\bfseries\color{lstpurple},
deletekeywords={count,min,max,avg,sum},
keywords=[2]{count,min,max,avg,sum},
keywordstyle=[2]\color{lstblue},
stringstyle=\color{lstreddark},
commentstyle=\color{lstgreen},
mathescape=true,
escapechar=@,
sensitive=true
}

\lstdefinestyle{datalog}
{
basicstyle=\footnotesize\upshape\ttfamily,
language=prolog
}

\lstdefinestyle{pseudocode}
{
  tabsize=3,
  basicstyle=\small,
  language=c,
  morekeywords={if,else,foreach,case,return,in,or},
  extendedchars=true,
  mathescape=true,
  literate={:=}{{$\gets$}}1 {<=}{{$\leq$}}1 {!=}{{$\neq$}}1 {append}{{$\listconcat$}}1 {calP}{{$\cal P$}}{2},
  keywordstyle=\color{lstpurple},
  escapechar=&,
  numbers=left,
  numberstyle=\color{lstgreen}\small\bfseries, 
  stepnumber=1, 
  numbersep=5pt,
}

\lstdefinestyle{xmlstyle}
{
  tabsize=3,
  basicstyle=\small,
  language=xml,
  extendedchars=true,
  mathescape=true,
  escapechar=£,
  tagstyle=\color{keywordpurple},
  usekeywordsintag=true,
  morekeywords={alias,name,id},
  keywordstyle=\color{lstred}
}

\lstdefinestyle{smtlib2}
{
tabsize=2,
basicstyle=\scriptsize\upshape\ttfamily,
numbers=left,
stepnumber=1,
breaklines=true,
stringstyle=\color{lstreddark},
commentstyle=\color{lstgreen},
mathescape=true,
escapechar=@,
sensitive=true
}

\renewcommand\footnotetextcopyrightpermission[1]{} 

\makeatletter
\let\@copyrightspace\relax

\begin{document}


\fancyhead{}

\title{Understanding Queries by Conditional Instances}
\settopmatter{authorsperrow=4}
\author{Amir Gilad}
\authornote{Both authors contributed equally to this research.}
\affiliation{%
\institution{Duke University}
\country{}
}
\email{agilad@cs.duke.edu}

\author{Zhengjie Miao}
\authornotemark[1]
\affiliation{%
  \institution{Duke University}
  \country{}
}
\email{zjmiao@cs.duke.edu}

\author{Sudeepa Roy}
\affiliation{%
  \institution{Duke University}
  \country{}
}
\email{sudeepa@cs.duke.edu }

\author{Jun Yang}
\affiliation{%
  \institution{Duke University}
  \country{}
}
\email{junyang@cs.duke.edu}

\renewcommand{\shortauthors}{}

\begin{abstract}
A powerful way to understand a complex query is by observing how it operates on data instances.
However, \emph{specific} database instances are not ideal for such observations: they often include large amounts of superfluous details that are not only irrelevant to understanding the query but also cause cognitive overload; and one specific database may not be enough.
Given a relational query, is it possible to provide a simple and \emph{generic} ``representative'' instance that (1) illustrates how the query can be satisfied, (2) summarizes all specific instances that would satisfy the query in the same way by abstracting away unnecessary details?
Furthermore, is it possible to find a collection of such representative instances that together completely characterize all possible ways in which the query can be satisfied?
This paper takes initial steps towards answering these questions.
We design what these representative instances look like, define what they stand for, and formalize what it means for them to satisfy a query in ``all possible ways.'' 
We argue that this problem is undecidable for general domain relational calculus queries, and  
develop practical algorithms for computing a minimum collection of such instances subject to other constraints.
We evaluate the efficiency of our approach experimentally,
and \common{show its effectiveness in helping users debug relational queries through a user study.}
\end{abstract}

\renewcommand{\algorithmicrequire}{\textbf{Input: }}
\renewcommand{\algorithmicensure}{\textbf{Output: }}


\maketitle
\setcounter{section}{0}
\section{introduction}\label{sec:intro}
A powerful way to understand a complex query is by observing how it operates on data instances.
A further in-depth approach may also consider the provenance of the query \cite{buneman2001and,cui2003lineage,green2007provenance} as an indicator how different combinations of tuples in the database  satisfy the query and generate each result. 
Another approach \cite{MiaoRY19} finds a minimal satisfying instance of the query. 
However, these 
characterizations 
are highly dependent on the given database instance, even when 
provenance is employed. 
In addition, such database instances can contain many details that divert attention from the query features themselves. 
For example, some queries are satisfied by an empty instance, but there may be other satisfying instances that are not trivial. 
Thus, some parts of the query may be ignored since they are not satisfied by the instance.
Furthermore, the evaluation on an instance leads to a satisfaction of specific combination of query atoms, but a different combination of atoms that is not satisfied by the specific database instance may reveal new insights.
While there are approaches to find a query solution without a given database instance \cite{chu2017cosette}, they provide only one way to satisfy a query, thereby again possibly missing different paths toward satisfying the query. 
%
%

\par
On the other hand, a single query can have infinitely many satisfying instances so showing all of them will not just be confusing, it may be impossible. 
Therefore, to understand all possible solutions to a query, 
we study the question of whether it is possible to provide a simple and \emph{generic} ``representative'' instance that (1) illustrates how the query can be satisfied, and (2) summarizes all specific instances that would satisfy the query in the same way by abstracting away unnecessary details. Further, we ask how we can find a collection of such representative instances that together completely characterize all possible ways to satisfy the query. 
\cut{
To answer these questions, we propose a novel approach of understanding queries, that is based on \emph{conditional instances} or \emph{c-instances} by adapting the notion of 
\emph{c-tables} \cite{ImielinskiL84} for incomplete databases from the literature, which are abstract database instances comprising variables (labeled nulls) along with a condition on those variables. Specific database instances satisfying a query can be of arbitrary size, whereas our solution using a collection of c-instances provides an abstract representation  of all satisfying database (ground) instances along with the particular way they satisfy a query using a notion of \emph{} with a small collection of succinct c-instances that provide the same \emph{coverage} by satisfying different parts of a query.  
}


\par
To answer these questions, we propose a novel approach of understanding queries based on \emph{conditional instances} or \emph{c-instances}, by adapting the notion of 
\emph{c-tables} \cite{ImielinskiL84} from the literature on incomplete databases, which are abstract database instances comprising variables (labeled nulls) along with a condition on those variables.
Thus, each c-instance can be considered a representative of all grounded instances that replace its variables with constants that satisfy the conditions they are involved in. 
However, it may be impossible to capture all satisfying instances with a single c-instance. 
Therefore, we 
use the idea of {\em coverage},
borrowed from the field of software validation \cite{MillerM63,myers2004art,ammann2016introduction}, where it has been well-studied in the context of software testing. For example, a test suite is said to \emph{cover} a function if the function is invoked during the test. This idea can be abstracted to program flows, where an edge/branch in the control-flow graph (see \cite{allen1970control} for details) is said to be covered if the edge/branch has been executed. 
For our use, when given a query $Q$ and a satisfying c-instance $\cinstance$, the atoms and conditions of $Q$ that are satisfied by all ground instances that $\cinstance$ represents are said to be covered by $\cinstance$. 
We intend to find a set of c-instances such that for every grounded instance that satisfies $Q$ with some coverage $\cov$, we have a c-instance 
that satisfies $Q$ with the same coverage $\cov$. 

\par

The idea of providing a compact representation of all instances that satisfy a query is appealing not just from a theoretical perspective, but also for multiple practical reasons. 
\common{
\begin{itemize}[leftmargin=*]
\item First, this approach can be used for explaining why a given (wrong) query is different than another (correct) one, as studied in works on counterexamples \cite{chu2017cosette,MiaoRY19}. In this scenario, we are given a query $Q_1$ and another query $Q_2$, the output is an instance $K$ such that $Q_1(K) \neq Q_2(K)$ (i.e., $K$ is a satisfying instance for $Q_1-Q_2$ or $Q_2-Q_1$). In an educational setting, such instances would help instructors and students understand why a query is wrong and debug it, without revealing the correct query to students.
\item Second, developers and data scientists who work with complex queries can this approach to explore how various parts of the queries can be triggered by different data, or to help them debug or refine these queries. For example, if there are no instances that can trigger some part of a query, it may be possible to simplify the query to remove ``dead code'' that logically contradicts other necessary conditions in the query.
\item Third, this approach offers a method for generating a suite of test instances for a complex query such that together they ``exercise'' all parts of the query.  In the field of synthetic data generation, previous works have proposed different approaches to generating data for testing workload queries \cite{BinnigKLO07,LoCH10,SanghiSHT18}. Using our approach, given a workload query $Q$, we can generate a set of instances to provide coverage testing for all parts of $Q$; furthermore, given a set of workload queries, we can generate test instances where a given subset of queries are satisfied but others are not.  These generated instances can be used for automated, comprehensive testing of queries.
\end{itemize}
} 
We illustrate the first use case above with an example below.

\begin{figure}[t]\scriptsize\setlength{\tabcolsep}{3pt}
  \subfloat[\small \Drinker\ relation]{
    \begin{minipage}[b]{0.4\linewidth}\centering
    {\scriptsize
      \begin{tabular}[b]{|c|c|c|}\hline
        {\tt name} & {\tt addr}  \\ \hline
        Eve Edwards & 32767 Magic Way \\\hline
      \end{tabular}
      }
    \end{minipage}
  }
    \hfill
  \subfloat[\small \Beer\ relation]{
    \begin{minipage}[b]{0.4\linewidth}\centering
    {\scriptsize
      \begin{tabular}[b]{|c|c|c|}\hline
        {\tt name} & {\tt brewer}  \\ \hline
        American Pale Ale &  Sierra Nevada\\\hline
      \end{tabular}
      }
    \end{minipage}
  }
  \\
    \subfloat[\small \Bar\ relation]{
    \begin{minipage}[b]{0.45\linewidth}\centering
    {\scriptsize
      \begin{tabular}[b]{|c|c|c|}\hline
        {\tt name} & {\tt addr}  \\ \hline
        Restaurant Memory & 1276 Evans Estate\\
        Tadim & 082 Julia Underpass\\
        Restaurante Raffaele & 7357 Dalton Walks \\\hline
      \end{tabular}
      }
    \end{minipage}
  }
    \hfill
   \subfloat[\small \Likes\ relation]{
    \begin{minipage}[b]{0.4\linewidth}\centering
     {\scriptsize
      \begin{tabular}[b]{|c|c|c|}\hline
        {\tt drinker} & {\tt beer}  \\ \hline
       Eve Edwards & American Pale Ale \\\hline
      \end{tabular}
      }
    \end{minipage}
  }\\
   \hfill
   \subfloat[\small \Serves\ relation]{
    \begin{minipage}[b]{1\linewidth}\centering
     {\scriptsize
      \begin{tabular}[b]{|c|c|c|c|}\hline
         {\tt bar} & {\tt beer} & {\tt price} \\ \hline
         Restaurant Memory &American Pale Ale &2.25\\
Restaurante Raffaele & American Pale Ale & 2.75 \\
Tadim & American Pale Ale & 3.5\\\hline
      \end{tabular}
      }
    \end{minipage}
  }\\
  \vspace{-2mm}
  \caption{\label{fig:running}
    Database instance $\ains_0$ of the Beers dataset. 
    We assume natural foreign key constraints from \Serves\ and \Likes\ to \Drinker, \Bar, \Beer.
  }
  \end{figure}

\begin{figure}[t]
    \centering\hrule\vspace{-2mm}
    \begin{subfigure}{1\linewidth}
    \begin{scriptsize}
        \begin{align*}
        \qcorrect=&\{(x_1, b_1)  \ \mid \ 
        \exists d_1, p_1 \big(\Serves(x_1, b_1, p_1) \land \Likes(d_1, b_1) \land \\
        & d_1 \ \sql{LIKE} \ \text{'Eve\textvisiblespace\%'} \land 
        \forall x_2, p_2 (\neg \Serves(x_2, b_1, p_2) \lor p_1 \geq p_2 \big) \}
        \end{align*}
        \caption{\mdseries\itshape Query \qcorrect: for each beer liked by any drinker whose first name is Eve, find the bars that serve this beer at the \emph{highest} price}\label{fig:q_correct}
        
    \end{scriptsize}
    \end{subfigure}
    
    \begin{subfigure}{1\linewidth}
        \begin{scriptsize}
    \begin{align*}
    \qincorrect=\{(x_1, b_1) & \ \mid \ 
    \exists d_1, p_1 \big( 
    \exists x_2, p_2 ( \Serves(x_1, b_1, p_1) \land \Likes(d_1, b_1) \\
    &\land d_1 \ \sql{LIKE} \ 'Eve\%' \land \Serves(x_2, b_1, p_2) \land p_1 > p_2\big) \}
    \end{align*}
        \caption{\mdseries\itshape Query \qincorrect\ which is similar to \qcorrect\ but does not use the difference operator and instead, find beers served at a \emph{non-lowest} price 
        }
        
        \label{fig:q_incorrect}
            \end{scriptsize}
    \end{subfigure}
    \hrule\vspace{-3mm}
    \caption{Correct query \qcorrect\ and incorrect query \qincorrect. 
    Note that the formula in $\qcorrect$ has a space after `Eve' whereas $\qincorrect$ does not. Here and later, \textvisiblespace  denotes the space symbol. 
    }
    \label{fig:queries}
\end{figure}

\begin{figure}[t]\hrule\vspace{-2mm}
\begin{scriptsize}
    \begin{align*}
    \qincorrect-\qcorrect  =  \{(x_1, b_1)  \ \mid \ 
    \exists d_1, p_1 \big( 
    \exists x_2, p_2 ( \Serves(x_1, b_1, p_1) \land \Likes(d_1, b_1) \land d_1 \ \sql{LIKE} \ 'Eve \%' \\
      \land \Serves(x_2, b_1, p_2) \land p_1 > p_2\big)
    \land  
    \forall d_2, p_3 \big(  \neg \Likes(d_2, b_1) \lor \neg (d_2 \ \sql{LIKE} \ \text{`Eve\textvisiblespace \%'}) \lor \\  \neg \Serves(x_1, b_1, p_3) \lor (\exists x_3, p_4 ( \Serves(x_3, b_1, p_4) \land  p_3 < p_4 ) )  \big)
    \}
    \end{align*}    
\end{scriptsize}\hrule
\vspace{-3mm}
    \caption{The difference query $\qincorrect\ - \qcorrect$ from Figure~\ref{fig:queries}.
    }
    \label{fig:diff-query}
\end{figure}

\begin{figure}[t]\scriptsize\setlength{\tabcolsep}{3pt}
\centering
  \subfloat[\small \Drinker\ relation]{
    \begin{minipage}[b]{0.25\columnwidth}\centering
  {\scriptsize
      \begin{tabular}[b]{|c|c|}\hline
        {\tt name} & {\tt addr} \\ \hline
        $d_1$ & $*$  \\\hline
      \end{tabular}
      }
      \end{minipage}
      }
  \subfloat[\small \Bar\ relation]{
    \begin{minipage}[b]{0.2\columnwidth}\centering
    {\scriptsize
      \begin{tabular}[b]{|c|c|}\hline
        {\tt name} & {\tt addr}  \\ \hline
        $x_1$ & $*$\\
        $x_2$ & $*$\\
        $x_3$ & $*$\\\hline
      \end{tabular}
      }
    \end{minipage}
  }
   \subfloat[\small Serves relation]{
    \begin{minipage}[b]{0.3\columnwidth}\centering
     {\scriptsize
      \begin{tabular}[b]{|c|c|c|}\hline
         {\tt bar} & {\tt beer} & {\tt price} \\ \hline
         $x_1$ &$b_1$ &$p_1$ \\
        $x_2$ &$b_1$  &$p_2$  \\
        $x_3$ & $b_1$ & $p_3$\\\hline
      \end{tabular}
      }
    \end{minipage}
  }
  \subfloat[\small \Beer\ relation]{
    \begin{minipage}[b]{0.25\columnwidth}\centering
    {\scriptsize
      \begin{tabular}[b]{|c|c|}\hline
        {\tt name} & {\tt brewer} \\ \hline
        $b_1$ &  $*$ \\\hline
      \end{tabular}
      }
    \end{minipage}
  }\\
   \subfloat[\small \Likes\ relation]{
    \begin{minipage}[b]{.4\columnwidth}\centering
     {\scriptsize
      \begin{tabular}[b]{|c|c|}\hline
        {\tt drinker} & {\tt beer}  \\ \hline
       $d_1$ & $b_1$ \\\hline
      \end{tabular}
      }
    \end{minipage}
  }
  \subfloat[\small Global condition]{
    \begin{minipage}[b]{0.6\columnwidth}\centering
    {\scriptsize
      \begin{tabular}[b]{|c|}\hline
      $d_1$ LIKE `Eve\%'  
    $\land p_1 > p_2 \land p_2 > p_3$\\\hline
      \end{tabular}
      }
    \end{minipage}
    }
  \vspace{-2mm}
  \caption{\label{fig:c-instance}
  C-instance $\cinstance_0$ that satisfies $\qincorrect-\qcorrect$ and generalizes the counterexample $K_0$ in Figure \ref{fig:running}.}
  \vspace{-3mm}
\end{figure}

\begin{Example}
\label{eg:running}
Consider the database $K_0$ shown in Figure \ref{fig:running} containing information about drinkers (\Drinker\ table), beers (\Beer\ table), bars (\Bar\ table), which beer does a drinker like (\Likes\ table), 
and which bar serves which beer (\Serves\ table). 
Also consider the queries \qcorrect\ and \qincorrect\ 
written in Domain Relational Calculus (DRC) 
in Figures \ref{fig:q_correct} and \ref{fig:q_incorrect}, respectively.  
The correct query \qcorrect\ returns a list of bars that serve the most expensive beer liked by any drinker whose first name is `Eve', whereas \qincorrect\ is a very similar query that chooses bars serving beers not at the lowest price and only requires first names to have a prefix of `Eve'.   
Figure \ref{fig:diff-query} shows the formula for $\qincorrect-\qcorrect$ but is not easily understandable and does not clearly show the difference between the queries.
In this case, Figure \ref{fig:running} gives the minimum counterexample $K_0$ for the difference between \qcorrect\ and \qincorrect\ \cite{MiaoRY19}. In particular, \qincorrect\ returns the tuples (Restaurante Raffaele, American Pale Ale) and (Tadim, American Pale Ale) while \qcorrect\ only returns the latter tuple. 
\par
Now consider the more general counterexample as a \emph{c-instance} (defined in the next section) showing the differences between the queries  $\qincorrect-\qcorrect$ in Figure \ref{fig:c-instance}. This c-instance, $\cinstance_0$, shows abstract tuples with variables instead of constants ($*$ are `don't care' variables) and a condition that the variables must satisfy (there should be a drinker whose name is `Eve' with a space after and the order of the prices in Serves table should be $p_1 > p_2 > p_3$). Thus, $\cinstance_0$ not only generalizes the counterexample in Figure \ref{fig:running} (i.e., there exists an assignment to the variables that results in the instance in Figure \ref{fig:running} and satisfies the global condition), but, it also specifies the `minimal' condition for which \qincorrect\ differs from \qcorrect\ (the global condition). The ground instance in Figure \ref{fig:running} contains specific values that may confuse the user and divert attention from the core differences. 
This is one of the c-instances in our \emph{universal solution} that includes three c-instances. Each of the c-instances captures a facet of the difference between \qincorrect\ and \qcorrect. 



\end{Example}

\mypar{Our contributions}
Our contributions are summarized below. 
\begin{itemize}[leftmargin=2mm]
    \item 
    We propose a framework for characterizing all query answers using c-instances using the notion of \emph{coverage} and a \emph{universal solution} that captures different ways a given DRC query can be satisfied. 
    \cut{
    
    We then define a model for query characterization based on the notion of coverage that defines a universal solution as a set of c-instances that satisfy the query, such that for every ground instance that satisfies it with some coverage $\cov$, there is a counterpart c-instance in the set with coverage $\cov$. Each such c-instance has to be minimal, where the minimality definition for c-instances is given in terms of size, but any other reasonable minimality definition will also work with our model and algorithms. We further give a less restrictive definition of universal solution called a c-solution that is a set of minimal c-istances with different coverages.
    }
    \item We argue that deciding whether a universal solution or even any satisfying c-instance exists is undecidable for general DRC queries, by giving a reduction from the finite satisfiability problem for First Order Logic formulas \cite{T-undecidable50}. However, for the 
    class of conjunctive queries with negation $CQ^{\neg}$, the universal solution can be found in poly-time in the query size. 

    \item Since the problem in general is undecidable, we give two practical algorithms for finding a minimal set of satisfying c-instances. The first algorithm runs an exhaustive search subject to a size limit on the c-instances 
    inspired by the chase procedure from Data Exchange 
    \cite{FaginKMP03,FaginKP03}.
    Then, we provide a more efficient chase algorithm that may return a smaller set of satisfying c-instances. 
    \cut{We start with an empty instance and perform BFS, creating a chase tree \cite{BaranyCKOV17} where each c-instance in the queue is expanded using a recursive procedure that iterates over the nodes of the query syntax tree and handles each node ($\land$, $\lor$, $\exists$, $\forall$) appropriately. The algorithm gradually expands each c-instance and the set of c-instances by considering all expansion options for each node. 
    The minimality of all c-instances is checked in post-processing. 
    As this algorithm performs a comprehensive search for c-instances, its performance is lacking in practice. 
    In particular, handling the disjunctive connectives is very time consuming.
    Thus, we propose an optimization that changes the original syntax tree into a set of conjunctive trees such that each subtree rooted at a $\lor$ node $Q_1 \lor Q_2$, is converted into three conjunctive trees ($Q_1 \land Q_2$, $\neg Q_1 \land Q_2$, and $Q_1 \land \neg Q_2$). This reduction allows us to run the chase algorithm on each conjunctive tree. While this procedure is more efficient than the first one, it may miss c-instances that will be found with the first algorithm (demonstrated in the paper).
    }
    
    \item We 
    experimentally show scalability and quality of the c-instances returned by our algorithms varying different parameters.
    Our experiments use a real collection of wrong queries submitted by students of an undergraduate database course as well as \common{queries 
    over the TPC-H schema.
    \cut{
    We examine the effect of the query structure and number connectives and quantifiers on the runtime. We also measure the number of c-instances found in the c-solution outputted by our approach and its variations. 

    We further examine the properties of the obtained c-instances in terms of their size and coverage.} 
    Finally, we provide a comprehensive user study and a case study that show the usefulness of our approach. }
\end{itemize}

\section{Related Work}
\label{sec:related}

\mypar{Test data generation} 
QAGen \cite{BinnigKLO07} was among the first systems that focused on data generation in a query-aware fashion. The system aimed at testing the performance of a database management system given a database schema, one parametric Conjunctive Query, and a collection of constraints on each operator. 
MyBenchmark \cite{LoCH10} extends \cite{BinnigKLO07} by generating a {\em set} of database instances that approximately satisfies the cardinality constraints from a set of query results. 
HYDRA \cite{SanghiSHT18} uses a \textit{declarative approach} that allows for 
the generation of a database summary that can be used for dynamically generating data for query execution. 
Cosette \cite{chu2017cosette}, which targets checking SQL equivalence without any test instances, encodes SQL queries to constraints using symbolic execution, and uses a constraint solver to find one counterexample that differentiates two input queries. 
RATest \cite{MiaoRY19} proposes an instance-based counterexample for distinguishing two queries, where the emphasis is on the cardinality of the generated counterexample. 
\revb{
The main differences between our work and \cite{MiaoRY19} are that (1) we provide abstract instances with variables and conditions to pinpoint the source of error while comparing a wrong query against a correct query,  (2) we only use the query and the schema to generate the counterexample 
and {\em do not need a database instance} to provides a counterexample, whereas \cite{MiaoRY19} requires a database instance to output one sub-instance as the counterexample,  (3) \cite{MiaoRY19} provides one grounded instance where the correct and wrong queries differ, while we generate {\em a set of c-instances}
aims to 
show all possible ways two queries can differ.
}
\cut{
Conversely, our solution aims at characterizing queries by multiple minimal {\em abstract} instances that satisfy the query and cover it in different ways. Thus, finding the difference between two queries is one use-case of our solution and for it, we consider all c-instances that differentiate the queries, where one of them has the same coverage as the instance found by Cosette or RATest. 
}
XData \cite{chandra2015data} generates test data by covering different types of query errors that can commonly occur. Qex \cite{veanes2010qex} is a tool for generating input relations and parameter values for a parameterized SQL query and aims at unit testing of SQL queries. It also generates one instance that satisfies the query and does not support nested queries and set operations, and thus does not support the full class of DRC queries. 
Olston et. al. \cite{olston2009generating}
studied the problem of generating small example data for dataflow programs to help users understand the behavior of their programs. 
\par
\textbf{Explanations for query results using provenance} 
Data provenance has been studied from many aspects \cite{buneman2001and,green2007provenance,Userssemiring1,GS13,CheneyProvenance,Olteanu12,sarma2008exploiting,amsterdamer2011provenance}. 
A multitude of approaches have used provenance for {\em query answers explanations}. 
One approach \cite{RS14} is based on tuple interventions in the database in the goal finding the influencing tuples over the query answers. 
Another \cite{MeliouRS14} quantifies the responsibility of each input tuples to the query answer, an idea inspired by causality \cite{pearl2009causal}. Shapley values \cite{shapley1953value} have also been used in this measure the contribution of each input tuple to the query result \cite{LivshitsBKS20}. 
Natural Language (NL) explanations have also been proposed \cite{DeutchFG20} when users employ an NLIDB such as \cite{nalir}.
The problem of explaining why a certain tuple is not in the query answer, also referred to as `why-not provenance', has been studied using two approaches;
instance-based \cite{DBLP:journals/pvldb/HuangCDN08,DBLP:journals/pvldb/HerschelHT09,
DBLP:journals/pvldb/HerschelH10,LS17} where explanations are (missing) input tuples,
and query-based \cite{DBLP:conf/sigmod/ChapmanJ09,DBLP:conf/sigmod/TranC10,DeutchFGH20} where explanations are based on query predicates or operators. 
Our approach suggests a query characterization that is independent of a specific database instance and thus also its provenance. 
\par
\textbf{Coverage in software testing}
Test coverage is used to measure the percentage of a given software that is executed during the tests. 
Intuitively, the higher the test coverage, the lower the likelihood of the software containing bugs and unforeseen errors \cite{MillerM63,zhu1997software,malaiya2002software}. 
Different criteria for coverage have been proposed \cite{myers2004art,zhu1997software}, e.g., function coverage (checking the percentage of functions in the program that are executed during testing) and branch coverage (checking the percentage of branches, i.e., decisions that have true and false outcomes, in the program that are executed during testing). 
These ideas have certainly influenced our model. 
\par
\textbf{Chase in schema mappings}
\revb{
The chase procedure was originally suggested in the context of database dependencies \cite{AhoBU79,MaierMS79} and later used for generating schema mappings \cite{FaginKMP03,ChiticariuT06}. 
The latter application aims to map one database schema to another, by using tuple-generating and equality-generating dependencies. These dependencies are then used in a chase procedure to generate the mapping between the schemas. 
Previous work has explored the complexity of the chase procedure and the types of solutions it is able to generate \cite{DeutschNR08,FaginKP05,GottlobN06}. 
Our notion of a universal solution is therefore inspired by the notion of universal solution in the schema mapping problem \cite{FaginKMP03}. 
Recent work has proposed an abstraction of schema mappings \cite{AtzeniBPT19}, which allows reusing meta schema mappings. In this paper, the c-instances can be thought of as ``meta-instances'' that can be mapped to concrete ones (e.g., \cite{MiaoRY19}) as needed. 
}
\section{Model for Query Characterization}\label{sec:model}

\cut{
\begin{table}[t]
\begin{small}
\centering
\caption{Notations}\label{tab:table-Of-notation}
\vspace{-3mm}
\begin{tabularx}{\linewidth}{| c | X | c |}
\hline $\schemaOf{\rel}$ & Schema\\
\hline $\varset$ & An infinite set of variables\\
\hline $\dom(A)$ & Domain of attr. $A$\\
\hline $Q$ & DRC query\\
\hline $\cinstance$ & C-instance\\
\hline $\Rep(\cinstance)$ & The set of ground instances generated by the c-instance \cinstance\\
\hline $\mathcal{S}_\cinstance$ & A set of c-instances\\
\hline $\assignment$ & An assignment to a DRC query from a ground instance\\
\hline $\map$ & 
A mapping from a c-instance $\cinstance$ to a ground
instance. 
\\
\hline
\end{tabularx}
\end{small}
\end{table}
}


\subsection{Databases and Relational Calculus}

First we review and define domains and Domain Relational Calculus which will be used for express queries in this paper.
A database schema $\schemaOf{\rel}$ is a collection $(\rel_1, ... , \rel_\relnum)$ of relation schemas. Each $\rel_i$ is defined over a set of attributes denoted by $\attrs(\rel_i)$. For each attribute $\att \in \attrs(\rel_i)$, its {\bf domain} is a set of (possibly infinite) {\bf constants} and is denoted as $\dom(\att)$, and $\dom = \cup_{\att} \dom(\att)$; for simplicity, we will frequently use $\dom$ instead of $\dom(\att)$ with the implicit assumption that the constants are from the right domain $\dom(\att)$. 
Two relations can share the same attribute $\att$; we use $\rel_i.\att$ to explicitly denote an attribute $\att \in \attrs(\rel_i)$. Further, two attributes may share the same domain (e.g., when they share the same name or are related by foreign key constraints). A {\bf ground instance} (or simply an {\bf instance} when it is clear from the context) is a (possibly empty) finite set of tuples with constant attribute values that conform to the schema and corresponding domains. In addition, 
we allow standard constraints like key constraints, foreign key constraints, and functional dependencies in our framework.

\paratitle{DRC queries and tree representation.}
We next review the definition of Domain Relational Calculus (DRC) \cite{LacroixP77} and use it to define queries and syntax trees. 
It has been shown that DRC is equivalent to Relational Algebra \cite{codd1972relational}, which provides the theoretical foundation to query languages such as SQL. 

\begin{Definition}[DRC Queries]\label{def:drc}
Given a schema $\schemaOf{R}$, a {\bf DRC query} $Q$ has the form $Q = \{(x_1, x_2, ..., x_{\folarity}) \mid \aFOL_Q(x_1, ... x_{\folarity})\}$ where 
$\aFOL_Q$ is a standard first order logic (FOL) formula \cite{AHV} involving relation names $R_1, \cdots, R_r$, constants from $\dom$, a set of {\bf \qvariable s}  $\varset_Q$ for attribute values, quantifiers $\exists, \forall$, operators $\neg, =, >, \geq, <, \leq, \neq, LIKE$ etc., and connectives $\wedge, \vee$. Here, $\varset_Q^{out} = \{x_1, \cdots, x_p\} \subseteq \varset_Q$ denote {\bf output variables} of the query $Q$, which can be an empty set for \emph{a Boolean query}. The output variables are \emph{free variables} in $P_Q$; the remaining variables in $P$ are {\em quantified} under $\forall$ or $\exists$. 
\par
The formula $\aFOL_Q$ is built up from {\bf DRC atoms} of the following forms: (1) 
$\rel(y_1 ..., y_{\relarity})$ or $\neg \rel(y_1 ..., y_{\relarity})$, where $\rel \in \schemaOf{R}$ is a relation,  
and each $y_i \in \varset_Q \cup \dom$ is a  query variable or a constant, 
and (2) conditions $x_1 \ op \ x_2$ or $x_1 \ op \ \aconst$, where $x_1, x_2 \in \varset_Q$, $\aconst \in \dom$, and $op$ is a binary operator.
\par
A ground instance $D$ is said to {\bf satisfy} a DRC query $Q$ (denoted by $D \models Q$) if $Q(D) \neq \emptyset$ (for a Boolean query, $Q(D) = \{\{\}\}$ or \emph{true}), i.e., there is a {\bf satisfying assignment} $\assignment: \varset_Q^{out} \rightarrow \dom$  of the output variables of $Q$ 
to the constants in $D$ 
such that $P_Q$ evaluates to true. 
\end{Definition}
We make a few assumptions without loss of generality: (1) in the FOL formula $\aFOL_Q$, all negations appear in DRC atoms (which can be achieved by repeated applications of standard equivalences like $\neg (\forall x P(x)) = \exists x(\neg P(x))$ and De Morgan's laws like $\neg (a \vee b) = \neg a \wedge \neg b$, etc.); (2) the DRC queries are {\bf safe} or domain independent \cite{AHV}, i.e., any variable $y_i$ that appears in a negated relation $\neg \rel(y_1 ..., y_{\relarity})$ also appear under a positive relation, e.g., queries like $\{x : \neg R(x)\}$ are not allowed; and (3) each quantified variable is unique in $P_Q$ (which can be achieved by renaming). 





\begin{Example}\label{ex:query}
The queries \qcorrect\ and \qincorrect\ are shown in Figure \ref{fig:queries}, whereas Figure~\ref{fig:diff-query} gives the difference query $\qincorrect-\qcorrect$ that identifies beers at a non-lowest price but also not at the highest price. 
In Figure \ref{fig:diff-query}, the output variables are $x_1, b_1$ and the FOL formula is specified on the right hand side by renaming the variables in $\qcorrect, \qincorrect$ in  Figure \ref{fig:queries}.
The ground instance $\ains$ in Figure~\ref{fig:running} satisfies  $\qincorrect-\qcorrect$ in Figure \ref{fig:diff-query}, since there is a satisfying assignment $\assignment$ from the output variables $x_1, b_1$ in the query $\qincorrect-\qcorrect$, i.e., $\assignment(x_1) = \text{``Restaurant Raffaele''}$, $\assignment(b_1) = \text{``American Pale Ale''}$, such that the formula in the query is satisfied when $d_1$ = ``Eve\textvisiblespace Edwards'', $p_1$ = 2.75, $x_2$ = ``Restaurant Memory'', $p_2 = 2.25$, therefore  the first part of the FOL formula from $\qincorrect$ is true. The second part of the FOL formula with $\wedge$ is also true for all $d_2, p_3$: while the first two disjuncts under $\neg$ evaluates to false, $\neg \Serves(x_1, b_1, p_3) = true$ for $p_3 = 2.25, 3.5$, and for $p_3 = 2.75$ the fourth disjunct is satisfied with $x_3 = \text{``Tadim''}$ and $p_4 = 3.5$.
\end{Example}
\begin{Definition}[Syntax Tree of Query]\label{def:syntax-tree}
A syntax tree of a query $Q$ is tree for the FOL formula $\aFOL_Q$ satisfying the following rules:
\begin{enumerate}
    \item Each leaf node is a DRC atom.
    \item Each internal node is either a quantifier with a single variable (e.g., $\forall x$ and $\exists x$) with a single child, or a connective ($\land$ and $\lor$) with two children.
\end{enumerate}
Further, since in the formula $\aFOL_Q$ all negations appear in the DRC atoms, all negations in the syntax tree appear in the leaves; we do not use separate nodes for negation.
\end{Definition}

Given a DRC query $Q$, we can have a unique syntax tree following the order of quantifiers in $Q$ (e.g., for $\exists x y$, $\exists y$ appears as the child of $\exists x$, and a fixed order of associative connectives with appropriate parentheses, e.g., $(p_1 \vee p_2 \vee p_3)$ is always assumed to be $((p_1 \vee p_2) \vee p_3)$. However, two equivalent DRC queries may have different syntax trees, e.g., $\{x \mid R(x)\}$ and $\{x \mid (R(x) \land T(x)) \lor (R(x) \land \neg T(x))\}$.
The syntax tree for the query $\qincorrect - \qcorrect$ from Figure~\ref{fig:diff-query} is shown in Figure \ref{fig:syntax-tree-minus} (to save space, we put multiple quantifiers with variables in the same node). The special treatment of the negation operator $\neg$  is for the sake of convenience in our algorithms. 

\begin{figure}[t]
    \begin{minipage}{1.0\linewidth}
     \centering\hspace*{-2mm}
    \begin{tikzpicture}[
        scale=.6,
    level 1/.style={level distance=0.6cm, sibling distance=-9mm},
    level 2/.style={sibling distance=-10mm, level distance=0.6cm}, 
    level 3/.style={level distance=0.8cm,sibling distance=-15mm},
    level 4/.style={level distance=1.0cm,sibling distance=-3mm},
    level 5/.style={level distance=1.0cm, sibling distance=0mm},
    snode/.style = {shape=rectangle, rounded corners, draw, align=center, top color=white, bottom color=blue!20},
    coverednode/.style = {shape=rectangle, rounded corners, draw, align=center, top color=applegreen, bottom color=applegreen!20,dashed},
    logic/.style = {shape=rectangle, rounded corners, draw, align=center, top color=white, bottom color=red!70}]
    \Tree
    [.\node[logic] {$\boldsymbol{\land}$}; 
        [.\node[logic] {$\boldsymbol{\exists {d_1, p_1}}$}; 
            [.\node[logic] {$\boldsymbol{\land}$}; 
                [.\node[logic] {$\boldsymbol{\land}$}; 
                    [.\node[logic] {$\boldsymbol{\land}$}; \node[coverednode]{$\Likes(d_1, b_1)$}; \node[coverednode]{$d_1$ \sql{LIKE} \text{`Eve\%'}};]
                    \node[coverednode,xshift=0mm] {$\Serves(x_1, b_1, p_1)$};]
                [.\node[logic] {$\boldsymbol{\exists {x_2, p_2}}$};
                    [.\node[logic] {$\boldsymbol{\land}$}; \node[coverednode] {$\Serves(x_2, b_1, p_2)$}; \node[coverednode] {$p_1 > p_2$};
                    ]
                ]
            ]
        ]
        [.\node[logic,xshift=0mm, yshift=-6mm] {$\boldsymbol{\forall {d_2, p_3}}$}; 
            [.\node[logic,xshift=0mm, yshift=-6mm] {$\boldsymbol{\lor}$}; 
                [.\node[logic, xshift=0mm, yshift=-6mm] {$\boldsymbol{\lor}$}; 
                    [.\node[logic, yshift=-6mm] {$\boldsymbol{\lor}$}; \node[snode, yshift=-6mm]{$\neg \Likes(d_2, b_1)$}; \node[snode, yshift=-6mm]{$\neg$($d_2$ \sql{LIKE} `Eve\textvisiblespace\%')};]
                    \node[coverednode,xshift=0mm, yshift=-6mm] {$\neg \Serves(x_1, b_1, p_3)$};]
                [.\node[logic,xshift=0mm, yshift=-6mm] {$\boldsymbol{\exists {x_3, p_4}}$};
                    [.\node[logic, xshift=0mm, yshift=-6mm] {$\boldsymbol{\land}$};
                        \node[coverednode,xshift=0mm, yshift=-6mm] {$\Serves(x_3, b_1, p_4)$}; \node[coverednode,xshift=0mm, yshift=-6mm] {$p_3 < p_4$};
                    ]
                ]
            ]
        ]
    ]
    \end{tikzpicture}
    \vspace{-3.5mm}
    \caption{Syntax tree of $\qincorrect-\qcorrect$ in Example~\ref{ex:query}. Atoms covered by the c-instance in Figure~\ref{fig:c-instance} are  \revb{in green dashed boxes}.} \label{fig:syntax-tree-minus}
    \vspace{-2mm}
\end{minipage}
\end{figure}


\cut{
\begin{Definition}[Domain Sets]
Let $\schemaOf{\rel} = (\rel_1, ..., \rel_{\relnum})$ be a relational schema, a set $\schema$ is a domain set of $\schemaOf{\rel}$ if 
$$\schema \subseteq \bigtimes_{1 \leq i \leq \relnum} \times_{\att \in \attrs(\rel_i)} \dom(\att)$$
where $\dom(\att)$ is the domain of $\att$. 
We denote $\schema.\att$ as the domain of attribute $\att$ in $\schema$ \amir{isn't this denoted by $\dom(\att)$ as in the expression above?}. 
A database instance $\db$ is said to be valid under $\schema$ iff for any value $v$ in $\db$ of an attribute $\att$, $v \in \schema.\att$.
\end{Definition}
}

\cut{
\sout{
\emph{Valid database instances.}
A database instance $\db$ is valid under the schema $\schema$ and a set of Foreign Key constraints $\mathcal{C}$ iff for any value $v$ in $\db$ of an attribute $\att$, $v \in \dom(\att)$ and $\db$ satisfies $\mathcal{C}$. 
}

\begin{Example}\label{eg:domain-set}
\sout{
Consider the simplified beers schema containing just the Likes relation with its attributes as detailed in Example \ref{ex:domain}. 
A domain set of this schema is: $\dom(Likes.drinker) = \{$`Amy', `Eve'$\}$, $\dom(Likes.beer) = \{$`Corona', `Budweiser'$\}$. 

The instance $\big\{Likes('Eve', 'Corona')\big\}$ is a valid instance, but the instance $\big\{Likes(`Bob', `Espresso') \big\}$ is invalid under the domain set, since $`Bob' \notin \dom(Likes.drinker)$ and $`Espresso' \notin \dom(Likes.beer)$.
}
\end{Example}
}


\cut{
\begin{Example}\label{eg:domain-set}
The instance $\big\{Likes('Eve', 'Corona')\big\}$ is a valid instance, but the instance $\big\{Likes(`Bob', `Espresso') \big\}$ is invalid under the domain set, since $`Bob' \notin \dom(Likes.drinker)$ and $`Espresso' \notin \dom(Likes.beer)$.
\end{Example}
}

\subsection{Conditional Instances or C-Instances}

We next give the definition of a \emph{c-instance} adapting the concepts of \emph{v-tables} and \emph{c-tables} from the literature \cite{ImielinskiL84}.
We distinguish between query variables whose domain is denoted by $\varset$ (Definition~\ref{def:drc}), and variables in the c-instances whose domain is denoted by $\lnset$; we refer to the latter as {\bf labeled nulls} (called \emph{marked nulls} in \cite{ImielinskiL84}) for clarity. The c-instances involve conditions using {\bf atomic conditions}, which are either (1) an atom of the form $[x~op~c]$ ($\neg [x~op~c]$) or $[x~op~y]$($\neg[x~op~y]$) where $x$ and $y$ are labeled nulls in $\lnset$, $c$ is a constant in $\dom(x)$, and $op \in \{<, >, \leq, \geq, =, \neq, LIKE, ...\}$ is a binary operator, or (2) a condition of the form $\neg \rel(x_1, \ldots, x_{\relarity})$ where $\rel$ is a relation on $\relarity$ attributes.

\cut{
\begin{Definition}[Atomic Condition]
\label{def:atom-cond}
An atomic condition is either (1) an atom of the form $[x~op~c]$ ($\neg [x~op~c]$) or $[x~op~y]$($\neg[x~op~y]$) where $x$ and $y$ are labeled nulls in $\lnset$, $c$ is a constant in $\dom(x)$, and $op \in \{<, >, \leq, \geq, =, \neq, LIKE, ...\}$ is a binary operator, or (2) a condition of the form $\neg \rel(x_1, \ldots, x_{\relarity})$ where $\rel$ is a relation on $\relarity$ attributes.
\end{Definition}
}




\begin{Definition}[Conditional Instance or c-instance]
A 
\emph{v-table} with a relational schema $\rel_i \in \schemaOf{\rel}$ is a table $\aCtable_i$ in which for each tuple $\tup \in \atable_i$ and
each attribute $\att \in \attrs(\rel_i)$, $\tup[\att]$ is either a constant from $\dom(\att)$ or is a labeled null from $\lnset$. 

A {\bf c-instance} $\cinstance$ of $\schemaOf{\rel}$ is a tuple of the form $(\{\aCtable_1, \ldots , \aCtable_\relnum\}, \universalcond)$, where for
each $i \in [1, \relnum]$, $\aCtable_i$ is a v-table with schema $\rel_i$, and 
$\universalcond$ is a conjunction of atomic conditions, which is associated with the c-instance, denoted as the {\bf global condition}.
\end{Definition}



Our definition of c-instances is slightly different from those found in previous literature \cite{ImielinskiL84}, as we only associate the instance with a global condition, while there are no local conditions associated with a single tuple or even a single table in the instance. 
Table-level conditions might still appear in the global condition as a conjunct that contains labeled nulls from a single table, e.g. in Figure \ref{fig:c-instance}, the condition $p_1 > p_2$ is only relevant to the $Serves$ table. 
\resolved{\SR{this is unclear -- local tables can also have conditions involving two labeled nulls, they can easily be included in global conditions with a conjunction $\wedge$ -- we need to say this, E.G. $p_1 > p_2$ belongs to the same table. did you mean conditions specifying whether a tuple or a table is present or absent?all tables are always present --- tuples are present if they do not have a condition that evaluates to false.} \amir{Revised}}

A conditional table (c-table) is a special case of a c-instance when there is only one relation in the instance, hence we only discuss c-instances in the rest of this paper.
Note that we also allow for 
labeled nulls that do not affect the evaluation of $\universalcond$, and are not needed for joins between tables. These are called ``don't care'' labeled nulls and are denoted by $*$ for simplicity instead of having unique names. 

\begin{figure}[t]\scriptsize\setlength{\tabcolsep}{3pt}
\centering
 \vspace{-2mm}
  \subfloat[\small \Drinker\ relation]{
    \begin{minipage}[b]{0.25\columnwidth}\centering
  {\scriptsize
      \begin{tabular}[b]{|c|c|}\hline
        {\tt name} & {\tt addr} \\ \hline
        $d_1$ & $*$  \\\hline
      \end{tabular}
      }
      \end{minipage}
      }
  \subfloat[\small \Bar\ relation]{
    \begin{minipage}[b]{0.2\columnwidth}\centering
    {\scriptsize
      \begin{tabular}[b]{|c|c|}\hline
        {\tt name} & {\tt addr}  \\ \hline
        $x_1$ & $*$\\
        $x_2$ & $*$\\\hline
      \end{tabular}
      }
    \end{minipage}
  }
   \subfloat[\small \Serves\ relation]{
    \begin{minipage}[b]{0.28\columnwidth}\centering
     {\scriptsize
      \begin{tabular}[b]{|c|c|c|}\hline
         {\tt bar} & {\tt beer} & {\tt price} \\ \hline
         $x_1$ &$b_1$ &$p_1$ \\
        $x_2$ &$b_1$  &$p_2$\\\hline
      \end{tabular}
      }
    \end{minipage}
  }
  \subfloat[\small \Beer\ relation]{
    \begin{minipage}[b]{0.25\columnwidth}\centering
    {\scriptsize
      \begin{tabular}[b]{|c|c|}\hline
        {\tt name} & {\tt brewer} \\ \hline
        $b_1$ &  $*$ \\\hline
      \end{tabular}
      }
    \end{minipage}
  }\\
    \hfill
   \subfloat[\small \Likes\ relation]{
    \begin{minipage}[b]{0.30\columnwidth}\centering
     {\scriptsize
      \begin{tabular}[b]{|c|c|}\hline
        {\tt drinker} & {\tt beer}  \\ \hline
       $d_1$ & $b_1$ \\\hline
      \end{tabular}
      }
    \end{minipage}
  }
  \subfloat[\small Global condition]{
    \begin{minipage}[b]{0.7\columnwidth}\centering
    {\scriptsize
      \begin{tabular}[b]{|c|}\hline
      $d_1$ LIKE `Eve\%' $\land \neg$($d_1$ LIKE `Eve\textvisiblespace\%') $\land p_1 > p_2$\\\hline
      \end{tabular}
      }
    \end{minipage}
    }
  \vspace{-2mm}
  \caption{\label{fig:c-instance-1}
  C-instance $\cinstance_1$ that satisfies $\qincorrect-\qcorrect$.}
  \vspace{-5mm}
\end{figure}

\begin{Example}
Consider the c-instance $\cinstance_0$ shown in Figure \ref{fig:c-instance}.
The single tuple in the drinker table says that the name of the drinker is $d_1$, its address is a ``don't care'' token, and the condition says that the name $d_1$ must start with
``Eve''.
The condition also enforces an order on the prices. 
$b_1$ cannot be replaced with $*$ because all three beers in Serves must be the same.
\end{Example}

C-instances define a set of possible worlds, each defined by a \emph{mapping} (see below) to the labeled nulls in the c-instance.
\begin{Definition}[Mapping for C-instances]\label{def:map}
Given a c-instance $\cinstance = (\aCtable_1, \ldots , \aCtable_\relnum, \universalcond)$ over a schema $\schemaOf{\rel}$, a {\bf mapping} $\map: \lnset \rightarrow \dom$ 
maps the labeled nulls $\lnset$ in $\cinstance$ to their respective domains. 
\end{Definition}
Extending $\map$ to the c-instance $\cinstance$, we denote $\map(\cinstance)$ as the ground instance $\map(\cinstance) = \{\map(x): x \text{ is a labeled null in } \aCtable_i$, $\text{ for }i \in [1, \relnum]\}$. Let $\lnset' \subseteq \lnset$ be the labeled nulls appearing in $\globalcond$ of $\cinstance$. Then, 
$\phi_{\map} = \phi(\map(\lnset'))$ yields the evaluation of $\phi$ (True or False) with the mappings given by $\map$.



\begin{Definition}[Possible Worlds and Consistency for C-instances]\label{def:possible-world}
Given a c-instance $\cinstance = (\aCtable_1, \ldots , \aCtable_\relnum, \universalcond)$ of schema $\schemaOf{\rel}$, the {\bf set of possible worlds} for $\cinstance$ is
$$\Rep(\cinstance) = \{\map(\cinstance): \map \text{ is a mapping} \land \universalcond_\map=True \}.$$
A c-instance $\cinstance$ is said to be {\bf consistent}, denoted by $\consistent(\cinstance)$ \cut{or satisfiable }if $\Rep(\cinstance) \neq \emptyset$. 
\end{Definition}
Note that ground instances in $\Rep(\cinstance)$ cannot contain extra tuples that are not in  the c-instance $\cinstance$. \resolved{ \amir{why do we need this? I forgot...}. \sr{ok to say this.}}
They may, however, have fewer tuples than the c-instance mapped to them because we consider set semantics, where a mapping may map two tuples with labeled nulls to the same tuple with constant values. 

\begin{Example}
One possible world of the c-instance in Figure \ref{fig:c-instance} is the ground instance in Figure~\ref{fig:running}. In this ground instance, all labeled nulls have a mapping such that the global condition is satisfied.
\end{Example}


\subsection{Query Characterization by C-Instances}\label{sec:characterization}

We next give 
definitions for our framework of query characterization through c-instances. The intuition is that the c-instances should be \emph{sound}, i.e., they should correctly capture a query, as well as \emph{``complete''} in terms of \emph{different ways} a query can be satisfied by ground instances, which requires more careful considerations using \emph{``coverage''} of ground and c-instances as we discuss below.

\resolved{\red{in terms of both soundness and completeness.} \amir{Didn't understand this}}

\begin{Definition}
\label{def:satisfy}
[Satisfying c-instances]
Given a query $Q$, a c-instance $\cinstance$ is said to {\bf satisfy} $Q$ (denoted $\cinstance \models Q$) if $\cinstance$ is consistent (Definition~\ref{def:possible-world}) and 
for every ground instance $D$ in $\Rep(\cinstance)$, $D \models Q$. 
\end{Definition}

\begin{Example}
Consider the query $\qincorrect-\qcorrect$ in Figure \ref{fig:diff-query}, 
and the c-instance $\cinstance_0$ shown in Figure \ref{fig:c-instance}. 
$\cinstance_0 \models \qincorrect-\qcorrect$ since every mapping of constants from the domain of its labeled nulls that satisfies the global condition will generate a ground instance $D$ such that $D \models \qincorrect-\qcorrect$; $D = K_0$ in Figure~\ref{fig:running} is an example.  
\cut{For example, the ground instance $D$ in Figure~\ref{fig:running} satisfies 
$\qincorrect-\qcorrect$ and for the c-instance $\cinstance$ in Figure \ref{fig:c-instance}, it holds that $D \in \Rep(\cinstance)$ since there is a mapping that maps, among other things, $d_1$ in Figure \ref{fig:c-instance} to `Eve Edwards' in Figure~\ref{fig:running} and $p_1, p_2, p_3$ in Figure \ref{fig:c-instance} to $3.5, 2.75, 2.25$ respectively in Figure~\ref{fig:running}.
}
\end{Example}

\resolved{\sr{You can briefly mention this below the defn using Example 8 -- can be shortened.}\amir{Didn't understand this}}


\cut{
\begin{Definition}[Homomorphism and c-homomorphism]

Given two c-instances 
$\cinstance_1, \cinstance_2$, a {\em homomorphism} from $\cinstance_1$ to $\cinstance_2$ is a function $h:Vars(\cinstance_1) \cup Cons(\cinstance_1) \rightarrow$ 
$Vars(\cinstance_2) \cup Cons(\cinstance_2)$, where $Vars(\cinstance)$ (or $Cons(\cinstance)$) is the set of variables (or constants) of a c-instance $\cinstance$, such that: (1) $h$ is an identity function on constants, and (2) $\forall x \in Vars(\cinstance_1) \cup Cons(\cinstance_1)$, $h(x)$ \red{=  a constant} $\in \dom(x)$ or $h(x)$ is a variable and $\dom(h(x)) = \dom(x)$. Note that $\forall A \in Atoms(\cinstance_1)$ where $A=\rel(x_1, ..., x_{\relarity})$ is a tuple in a c-table with relation $R$, it holds that $h(A) = \rel(h(x_1), ..., h(x_{\relarity})) \in Atoms(\cinstance_2)$, i.e., $h$ can be lifted to tuples (and in particular, preserves relation names) and further lifted to c-instances. 

\SR{do not follow the `i.e.' part, especially for formula, for which you need to define what happens for connectives.}

\SR{ADDED -- (3) should not appear in the above sentence, can appear in the next sentence. Do not use R(x1,.., xk) as relational atom, use only R or R with attribute names, R(x1,.., xk) means a tuple. Did you mean mapping of tuples or relation names? if it is relation name why do you need the mapping? the sentence is not correctly phrased if you meant mapping for tuples.}
\amir{Revised}
$h$ can also be extended over ground instances by only looking at the tuples but not atomic conditions.


A homomorphism $h$ \red{from $\cinstance_1$ to $\cinstance_2$} is said to be a {\em c-homomorphism} if $\consistent(\cinstance_1)$ is implied by $\consistent(\cinstance_2)$, i.e., any mapping of the labeled nulls of $\cinstance_2$, $\map_2$, that satisfies the global conditions in $\cinstance_2$ can be converted into the assignment $\assignment_1 =  \assignment_2 \circ h$ that satisfies the global conditions in $\cinstance_1$. 
We say that $\cinstance_1$ c-homomorphically maps to $\cinstance_2$ if there is a c-homomorphism from $\cinstance_1$ to $\cinstance_2$. Note that for a ground instance $\ains$, $\consistent(\ains)$ is always true, 
\sr{say that it is alwaays true as $\cinstance_1 \in Rep(\cinstance_1)$} \amir{Revised}
so a homomorphism $h$ mapping a c-instance $\cinstance_1$ to a ground instance $\ains$ is a c-homomorphism if the condition in $\cinstance_1$ evaluates to True in $\ains$. 
\sr{do not follow this paragraph. if $I_2$ is a ground intance, $I_1$ is already satisfiable and therefore c-homomorphism holds?} \amir{Revised in and after the def.}
\end{Definition}

In particular, the definition says that any homomorphism from a c-instance $\cinstance$ to a ground instance $\ains \in Rep(\cinstance)$ is a c-homomorphism since $\ains$ is satisfiable. 
}

\begin{figure}[t]\scriptsize\setlength{\tabcolsep}{3pt}
\centering
 \vspace{-2mm}
  \subfloat[\small \Drinker\ relation]{
    \begin{minipage}[b]{0.25\columnwidth}\centering
  {\scriptsize
      \begin{tabular}[b]{|c|c|}\hline
        {\tt name} & {\tt addr} \\ \hline
        $d_1$ & $*$  \\
        $d_2$ & $*$  \\\hline
      \end{tabular}
      }
      \end{minipage}
      }
  \subfloat[\small \Bar\ relation]{
    \begin{minipage}[b]{0.2\columnwidth}\centering
    {\scriptsize
      \begin{tabular}[b]{|c|c|}\hline
        {\tt name} & {\tt addr}  \\ \hline
        $x_1$ & $*$\\
        $x_2$ & $*$\\
        $x_3$ & $*$\\\hline
      \end{tabular}
      }
    \end{minipage}
  }
   \subfloat[\small \Serves\ relation]{
    \begin{minipage}[b]{0.25\columnwidth}\centering
     {\scriptsize
      \begin{tabular}[b]{|c|c|c|}\hline
         {\tt bar} & {\tt beer} & {\tt price} \\ \hline
         $x_1$ &$b_1$ &$p_1$ \\
        $x_2$ &$b_1$  &$p_2$  \\
        $x_3$ & $b_1$ & $p_3$\\\hline
      \end{tabular}
      }
    \end{minipage}
  }
\subfloat[\small \Beer\ relation]{
    \begin{minipage}[b]{0.25\columnwidth}\centering
    {\scriptsize
      \begin{tabular}[b]{|c|c|}\hline
        {\tt name} & {\tt brewer} \\ \hline
        $b_1$ &  $*$ \\\hline
      \end{tabular}
      }
    \end{minipage}
  }\\
   \subfloat[\small \Likes\ relation]{
    \begin{minipage}[b]{0.25\columnwidth}\centering
     {\scriptsize
      \begin{tabular}[b]{|c|c|}\hline
        {\tt drinker} & {\tt beer}  \\ \hline
       $d_1$ & $b_1$ \\\hline
      \end{tabular}
      }
    \end{minipage}
  }
  \subfloat[\small Global condition]{
    \begin{minipage}[b]{0.6\columnwidth}\centering
    {\scriptsize
      \begin{tabular}[b]{|c|}\hline
      $d_1$ \sql{LIKE} `Eve\%' $\land$ $d_1$ \sql{LIKE} `Eve\textvisiblespace\%' $\land \neg \Likes(d_2, b_1)$ 
      $\land$ \\$\neg$ ($d_2$ \sql{LIKE} `Eve\textvisiblespace\%') 
      $\land p_1 > p_2 \land p_2 > p_3$\\\hline
      \end{tabular}
      }
    \end{minipage}
    }
  \caption{\label{fig:c-instance-2}
  C-instance $\cinstance_2$ that satisfies $\qincorrect-\qcorrect$. It is not minimal since one of the \Serves\ tuples can be removed without changing the coverage.}
\end{figure}

\cut{
\begin{Example}
Considering the c-instance $\cinstance_0$ in Figure~\ref{fig:c-instance} and the c-instance $\cinstance_1$ in Figure~\ref{fig:c-instance-1}. There is a homomorphism from $\cinstance_0$ to $\cinstance_1$ that maps all variables to themselves except mapping $(b_3, p_3)$ to either $(b_1, p_1)$ or $(b_2, p_2)$. However, there is no c-homomorphism from $\cinstance_0$ to $\cinstance_1$ because the atomic conditions on $d_1$ in $\cinstance_0$  ($d_1$ LIKE `Eve\textvisiblespace\%') and $\cinstance_1$ ($\neg( d_1$ LIKE `Eve\textvisiblespace\%'$)$) are conflicting, and $p_2 > p_3$ does not appear in the condition in $\cinstance_1$. The identity mapping is a c-homomorphism from $\cinstance_0$ to $\cinstance_2$ in Figure~\ref{fig:c-instance-2}. 

Consider the ground instance $\ains$ in Figure~\ref{fig:running}. There is a c-homomorphism, $h_1$, from $\cinstance_0$ to $\ains$, where some of the mappings are
$\{d_1 \to \text{Eve Edwards}, b_1 \to \text{Restaurant Raffaele}, b_2 \to \text{Tadim}, b_3 \to \text{Restaurant Memory},   $
$b_1 \to \text{American Pale Ale}, p_1 \to 3.5, p_2 \to 2.75, p_3 \to 2.25\}$ (don't-cares are skipped). 
\end{Example}
}

\paratitle{Coverage.} Given a DRC query, the \emph{Coverage} of a (ground or c-) instance captures different parts of a DRC query, 
that are satisfied by the instance, and helps us define the \emph{completeness} of  a set of c-instances. 
An inductive definition of coverage is given below. 



\begin{Definition}\label{def:coverage-ground-instance}[Coverage of ground instances]
Given a DRC query syntax tree $Q$, a ground instance $K$ such that $K \models Q$, and a satisfying assignment $\assignment: \varset_Q^{out} \rightarrow \dom$ of the output variables of $Q$,  the {\bf coverage $\coverage(Q, K, \assignment)$
of $Q$ by $K$ under $\assignment$} identifies a subset of the DRC atoms (leaves) of the query syntax tree recursively top-down by extending $\assignment$ to all free variables in a subtree as follows: 
\begin{enumerate}
    \item If $Q$ consists of a single DRC atom (here all variables in $Q$ are free), then $\coverage(Q, K, \assignment)$ contains the DRC atom of $Q$ if:
    \begin{enumerate}
        \item $Q=\rel(x_1, ... , x_\relarity)$ and $\rel(\assignment(x_1), ... , \assignment(x_\relarity))$ appears in $K$, or
        \item $Q=\neg\rel(x_1, ... , x_\relarity)$ and $\rel(\assignment(x_1), ... , \assignment(x_\relarity))$ is not in $K$, or 
        \item $Q=x~op~y$ and $\assignment(x)~op~\assignment(y)$ evaluates to True;
    \end{enumerate}
    otherwise $\coverage(Q, K, \assignment) \gets \emptyset$.
    \item If $Q = Q_1 \land Q_2$ or $~Q = Q_1 \lor Q_2$ (here the sets of free variables in $Q, Q_1, Q_2$ are the same), then $\coverage(Q, K, \assignment) = \coverage(Q_1, K, \assignment) \cup \coverage(Q_2, K, \assignment)$.
    \item If $Q = \exists x Q'(x)$ or $~Q = \forall x Q'(x)$ (here $x$ is a new free variable in $Q'$), then \\$\coverage(Q, K, \assignment)$ = $\cup_{c \in \dom_K} \coverage(Q', K, \assignment \cup \{x \to c\})$, where $\dom_K$ denotes the constants appearing in $K$.
\end{enumerate}
The {\bf coverage of $K$ for $Q$} is defined as $\coverage(Q, K) \gets \bigcup_{\assignment} \coverage(Q, K, \assignment)$.
\end{Definition}


Intuitively, the coverage $\coverage(Q, K)$ is the set of atoms and conditions of $Q$ that can be covered by a ground instance $K$, eventually leading to a satisfying assignment of the output variables of $Q$. 
Therefore, we use union to combine the coverages in Definition \ref{def:coverage-ground-instance} for all cases, since we are interested in all possible ways to satisfy a query. 
Since $\assignment$ is a satisfying assignment of the output variables, the coverages of $Q'$ in case (3) in Definition \ref{def:coverage-ground-instance} for both $\exists, \forall$, and for both $Q_1, Q_2$ for a $\wedge$ node and for at least one of them for a $\vee$ node in case (2) above is non-empty.
For universal quantifiers, it is worth noting that when the quantified variable takes different constants, different branches of the inner query ($Q'$) may be satisfied, and thus provide different coverages, and to take all of them into account, we again employ union.

\begin{Example}\label{ex:coverage_ground}
The only satisfying assignment to the difference query $\qincorrect-\qcorrect$ depicted in Figure \ref{fig:diff-query} w.r.t. the ground instance shown in Figure \ref{fig:running} is given by the assignment $\assignment(x_1) = \text{``Restaurant Raffaele''}$, $\assignment(b_1) = \text{''American Pale Ale''}$ described in Example \ref{ex:query}. By applying the recursive top-down process implied by Definition \ref{def:coverage-ground-instance}, 
the DRC atoms of the query covered by this assignment are the leaves colored in green in Figure \ref{fig:syntax-tree-minus}. Note that different assignments of a $\forall$ variable can cover different leaves, e.g., for $\forall p_3$ node in the right subtree,  $p_3 = 2.25, 3.5$ covers the   node $\neg \Serves(x_1, b_1, p_3)$ atom whereas $p_3 = 2.75$ covers $\Serves(x_3, b_1, p_4)$ and $p_3 < p_4$ atoms as discussed in  Example \ref{ex:query}. 
\end{Example}

\begin{Definition}\label{def:coverage-c-instance}[Coverage of satisfying c-instances]
Given a DRC query $Q$ and a c-instance $\cinstance$ such that $\cinstance \models Q$,  the {\bf coverage of $\cinstance$ for $Q$} is defined as $\coverage(Q, \cinstance) = \bigcap_{K \in \Rep(\cinstance)} \coverage(Q, K)$.
\end{Definition}

Since $\Rep(\cinstance)$ can contain ground instances with different coverages, the coverage of $\cinstance$ is defined as the common coverage of all possible worlds.
Therefore, the coverage of a c-instance $\cinstance$ is always a (not necessarily strict) subset of any ground instance in $\Rep(\cinstance)$.

\begin{Example}\label{ex:coverage-1}
Reconsider the syntax tree of the query $\qincorrect-\qcorrect$ shown in Figure \ref{fig:syntax-tree-minus}. 
Every ground instance in $\Rep(\cinstance_0)$ ($\cinstance_0$ is depicted in Figure \ref{fig:c-instance}) has to have exactly three \Serves\ tuples due to the global conditions (recall that the mapping from a c-instance to each ground instance in $\Rep$ has to be onto). 
Thus, the coverage of the c-instance $\cinstance_0$ depicted in Figure \ref{fig:c-instance} is exactly the coverage of the ground instance in Figure \ref{fig:running}/Example \ref{ex:coverage_ground} and highlighted in green \revb{with dashed frames} in the tree shown in Figure \ref{fig:syntax-tree-minus}.
The c-instance $\cinstance_2$ in Figure \ref{fig:c-instance-2} covers all DRC atoms in the query $\qincorrect-\qcorrect$. To see this, note that there are two drinker variables $d_1, d_2$ in $\cinstance_2$ such that $\neg \Likes(d_2, b_1)$ as well as $\neg (d_2$ \sql{LIKE} `Eve\textvisiblespace\%'$)$ hold in all 
ground instances in $\Rep(\cinstance_2)$. Therefore, for the above satisfying assignments of output variables, when the $\forall d_2$ in the right subtree iterates over the constants corresponding to $d_2$, the two remaining uncovered leaves are covered as well.   
\end{Example}


  


\paratitle{Minimality of a c-instance.} 
We now define minimality w.r.t. the coverage of a c-instance. 
In the next definition, we denote by $|\cinstance|$ the {\bf size of a c-instance $|\cinstance|$}, defined as the total number of tuples and atomic conditions of $\cinstance$. For instance, $|\cinstance_0| = 12$ in Figure~\ref{fig:c-instance} (9 tuples and 3 atomic conditions).

\begin{Definition}[Minimal satisfying c-instance]\label{def:minimal-cinstance}
Given a DRC query $Q$, a c-instance $\cinstance$ with coverage $\cov$ satisfying $Q$ is {\bf minimal} if for every other satisfying c-instance $\cinstance'$ that has coverage $\cov$, it holds that $|\cinstance| \leq |\cinstance'|$. 
\end{Definition}

\begin{Example}
Following Example~\ref{ex:coverage-1}, the c-instance $\cinstance_0$ in Figure~\ref{fig:c-instance} is a minimal satisfying c-instance assuming natural foreign key constraints from $\Serves, \Likes$ to $\Drinker, \Beer, \Bar$, since any other c-instance with the same coverage has a larger size. In particular, $\cinstance_1$ of smaller size  (=10) in Figure \ref{fig:c-instance-1} does not have the same coverage, since in the $\forall$ nodes of right subtree when $d_2$ (in query) = $d_1$ (in $\cinstance_1$) and $p_3$  (in query) = $p_1$ (in $\cinstance_1$), the two rightmost leaves $\Serves(x_3, b_1, p_4)$ and $p_3 < p_4$ are not covered by any ground instance in $\Rep(\cinstance_1)$. 
\end{Example}

Minimality ensures that we do not include redundant tuples or conditions in our c-instances. While we adopt the above simple notion of minimality, it is ensured by a post-processing step in our algorithms, so any other reasonable form of minimality can also be used in this framework.

\cut{
\paratitle{Minimality of a c-instance.} 
Minimality of a c-instance is defined w.r.t. other c-instances. 
To compare between different c-instances, we first define the concept of epimorphism. 

\begin{Definition}[Epimorphism]\label{def:homomorphism}
Given two c-instances $\cinstance_1, \cinstance_2$ on the same schema with labeled nulls $\lnset_1,  \lnset_2$, 
constants $\constset_1, \constset_2$, and global conditions $\universalcond_1, \universalcond_2$ respectively, an {\bf epimorphism $h$ from $\cinstance_1$ to $\cinstance_2$} is a function 
$h: \lnset_1 \cup \constset_1 \rightarrow \lnset_2 \cup \constset_2$ such that, (1) for each constant $c \in \constset_1$, $h(c) = c$, 
(2) $h$ is onto for $\lnset_2$, i.e., for each $y \in \lnset_2$ there is a $x \in \lnset_1$ such that $h(x) = y$, (3) for each tuple
$t=\rel(x_1, ..., x_{\relarity})$ in $\cinstance_1$, $h(t) = \rel(h(x_1), ..., h(x_{\relarity}))$ is in the same table $\rel$ in $\cinstance_2$,
and 
(4) $\globalcond_2 \Rightarrow h(\globalcond_1)$, where $h(\globalcond_1)$is the Boolean formula $\globalcond_1$ obtained by replacing the labeled nulls $x \in \lnset_1$ with $h(x)$.
If there exists such an $h$ we say that {\bf $\cinstance_1$ is epiomorphic to $\cinstance_2$}.
\end{Definition}


\sr{please check the prposition below.}\amir{Looks good}

\sout{For two functions $h_1, h_2$ mapping labeled nulls to other labeled nulls and constants, let us denote by $h_1 \circ h_2$ the composition of $h_1$ and $h_2$, i.e., for a labeled null $x$, $(h_1 \circ h_2)(x) = h_1(h_2(x))$.} \amir{This seems redundant}

\begin{proposition}
For c-instances $\cinstance_1, \cinstance_2$, if $h$ is an epimorphism from $\cinstance_1$ to  $\cinstance_2$, and ground instance $K \in \Rep(\cinstance_2)$ by the mapping $\map$, then $K \in \Rep(\cinstance_1)$ by the mapping $\map \circ h$. Therefore, $\Rep(\cinstance_1) \supseteq \Rep(\cinstance_2)$ and $\consistent(\cinstance_2) \Rightarrow \consistent(\cinstance_1)$ 
\end{proposition}

\begin{proof}
Let $\lnset_1, \lnset_2$ be the labeled nulls, and $\globalcond_1, \globalcond_2$ be the global conditions of $\cinstance_1$ and $\cinstance_2$ respectively. 
Since $K \in \Rep(\cinstance_2)$ by the  mapping $\map: \lnset_2 \rightarrow \dom$, then $\globalcond_2(\map(\lnset_2'))$ = true, where $\lnset_2' \subseteq \lnset_2$ is the set of labeled nulls used in $\globalcond_2$. Since $\globalcond_2 \Rightarrow h(\globalcond_1)$, any mapping $\map$ that satisfies $\globalcond_2$ also satisfies $h(\globalcond_1)$. Therefore, it follows that the formula $\globalcond_1$ where labeled nulls $x$ in $\lnset_1$ are replaced by $h(x)$, and in turn by constants in $\map(h(x))$ is true, which gives the ground instance $K$. Therefore, $K \in \Rep(\cinstance_1)$ by the mapping $\map \circ h$.

\end{proof}



Note that a mapping $\map$ is a special form of epimorphism that maps a c-instance to a satisfying ground instance by mapping all labeled nulls to constants.

\begin{Example}
Considering the c-instance $\cinstance_0$ in Figure~\ref{fig:c-instance} and the c-instance $\cinstance_1$ in Figure~\ref{fig:c-instance-1}. There is an epimorphism from $\cinstance_0$ to $\cinstance_1$ that maps all variables to themselves except $b_3, p_3$, and maps $(b_3, p_3)$ to either $(b_1, p_1)$ or $(b_2, p_2)$. However, there is no epimorphism from $\cinstance_0$ to $\cinstance_1$ because the atomic conditions on $d_1$ in $\cinstance_0$  ($d_1$ LIKE `Eve\textvisiblespace\%') and $\cinstance_1$ ($\neg( d_1$ LIKE `Eve\textvisiblespace\%'$)$) are conflicting, and $p_2 > p_3$ does not appear in the condition in $\cinstance_1$. The identity mapping is an epimorphism from $\cinstance_0$ to $\cinstance_2$ in Figure~\ref{fig:c-instance-2}. 

\end{Example}

We now define minimality using the concept of epimorphism. 

\begin{Definition}[Minimal satisfying c-instance]
\cut{
}

Given a query $Q$, a c-instance $\cinstance$ with coverage $\cov$ satisfying $Q$ is minimal if for every other c-instance $\cinstance'$ with coverage $\cov'$ such that 
$\cinstance$ is epimorphic to $\cinstance'$, it holds that $\cov' \subsetneq \cov$. 
\end{Definition}

Intuitively, an epimorphism `shrinks' the c-instance and thus, if we are able to shrink the c-instance without losing coverage, it is not minimal. 
If $\cov' = \cov$, we would rather take $\cinstance'$ as it is `smaller' with the same coverage. 

\sr{the example below has to be fixed, not correct} \amir{revised}
\begin{Example}
Following Example~\ref{ex:coverage-1}, the c-instance $\cinstance_0$ in Figure~\ref{fig:c-instance} is a minimal satisfying c-instance, since any epimorphic c-instance does not satisfy $\qincorrect-\qcorrect$ (e.g., mapping any set of tuples in \texttt{Bar} or \texttt{Serves} to a single tuple). In particular, $\cinstance_0$ is not epimorphic to the c-instance $\cinstance_1$ in Figure \ref{fig:c-instance-1} since the global condition in $\cinstance_1$ does not imply the one in $\cinstance_0$, and $\cinstance_0$ is not epimorphic to the c-instance $\cinstance_2$ in Figure \ref{fig:c-instance-2} since there are two Drinker tuples in $\cinstance_2$, and thus any function from $\cinstance_0$ to $\cinstance_2$ cannot be onto $\cinstance_2$.


\end{Example}
} 




\mypar{Universal solution} Our framework for finding solutions for a query $Q$ is given in terms of a set of minimal c-instances $\cinstance \models Q$, which immediately ensures soundness of our solutions, since any output c-instance is guaranteed to satisfy the query. Conversely, the notion of completeness is more challenging since there can be satisfying ground instances of unbounded size with redundant tuples that do not affect the query answer. The notion of coverage helps us define a notion of completeness using a \emph{universal solution}, which ensures that for all satisfying ground instances of a certain coverage, a c-instance with the same coverage is included.

\begin{Definition}[Minimal c-Solution and Universal Solution]\label{def:universal}
Let $Q$ be a DRC query over a schema $\schemaOf{\rel}$ 
and domain $\dom$. 
A {\bf minimal c-solution} of $Q$ is a set of minimal c-instances of $Q$, $\mathcal{S}_{\cinstance} = \{\cinstance_1, \ldots, \cinstance_k\}$, such that  for all $\cinstance_i$, $\cinstance_i \models Q$, and for any two $\cinstance_i, \cinstance_j$ where $i \neq j$, $\coverage(Q, \cinstance_i) \neq \coverage(Q, \cinstance_j)$. 
\par
A {\bf universal solution} of $Q$ is a minimal c-solution $\mathcal{S}_{\cinstance} = \{\cinstance_1, \ldots, \cinstance_k\}$ such that 
(1) if there exists a ground instance $K$, where  $K \models Q$, with coverage $\cov$, there is a c-instance $\cinstance_i \in \mathcal{S}_{\cinstance}$ 
with coverage 
$\cov = \cov_i$, 
(2) if we remove any c-instance from $\mathcal{S}_{\cinstance}$, condition (1) does not hold.
\end{Definition}
\cut{
A {\bf universal solution} of $Q$ is a set of minimal satisfying c-instances of $Q$, $\mathcal{S}_{\cinstance} = \{\cinstance_1, \ldots, \cinstance_k\}$, such that (1) if there exists a ground instance $K$, where  $K \models Q$, with coverage $\cov$, there is a c-instance $\cinstance_i \in \mathcal{S}_{\cinstance}$ 
with coverage 
$\cov = \cov_i$, 
(2) if we remove any c-instance from $\mathcal{S}_{\cinstance}$, condition (1) does not hold.
}

Note that for the universal solution, we do not require that $\ains \in \Rep(\cinstance_i)$ since $\ains$ may contain more tuples than $\cinstance_i$ and thus may not be part of the set $\Rep(\cinstance_i)$. 

\begin{Example}\label{ex:universal-sol}
Reconsider the $\qincorrect-\qcorrect$ shown in Figure \ref{fig:diff-query}. The set $\{\cinstance_0, \cinstance_1\}$ is a minimal c-solution for  $\qincorrect-\qcorrect$ since they have different coverages and both satisfy $\qincorrect-\qcorrect$.  Two of the three c-instances in the universal solution are $\cinstance_0$, $\cinstance_1$ shown in Figures \ref{fig:c-instance}, \ref{fig:c-instance-1}, \ref{fig:c-instance-2}, respectively (the rest are shown in Section \ref{sec:case-study} in the case study). 
\end{Example}





\subsection{Computability of the Universal Solution}

\begin{proposition}\label{prop:complexity}
The computability and complexity of finding a universal solution is as follows:
\begin{itemize}

    \item[(1)] Finding a universal solution is poly-time in the size of the query for $CQ^{\neg}$ (and therefore also for $CQ$s), where $CQ^{\neg}$ is the class of conjunctive queries with negation that includes 
    operators  $\exists$, $\land$, $\neg$ for individual atoms and conditions.
    \item[(2)] 
    Checking whether a universal solution exists (or even any minimal c-solution exist) is undecidable for general DRC queries that may include the operators $\forall$, $\exists$, $\lor$, $\land$, and $\neg$.
\end{itemize}
\end{proposition}

\begin{proof} 
(1) The universal solution of $q \in CQ^{\neg}$ is a single c-instance comprising all relational atoms $R(x_1, \cdots, x_k)$ of the query,  and a global condition that is the conjunction of all comparisons ($x~ op~ y$, $x~ op~ c$ with or without negation) and negated relational atoms $\neg R(x_1, \cdots, x_k)$ in the query. This implies a poly-time complexity in the size of the query. 
\par
\cut{
(2) 
The universal solution of $Q = \{q_1, \ldots, q_u\} \in UCQ_{\neg}$ is composed of a set of c-instances where each of them satisfies and covers a conjunction of a subset of CQs from $\{q_1, \ldots, q_u\}$ that can be satisfied together (which is a CQ itself). This is equivalent to enumerating all satisfying assignments to the UCQ (every CQ is a conjunct in the formula). We can also reduce the problem of counting the number of satisfying assignments to a DNF formula to the problem of funding a universal solution for a UCQ. Thus, the problem is \#P-complete. 
\par
(3)
}
(2) Finding whether a universal solution exists is an undecidable problem for general DRC queries due to a reduction from the 
\emph{finite satisfiability problem in first order logic} (FOL) that is known to be undecidable by the Trakhtenbrot’s Theorem \cite{T-undecidable50}. 
An FOL sentence $\phi$ is finitely satisfiable if there exists
a finite 
ground instance
$K$ such that $\phi$ is true over $K$, which is true if and only if there is a satisfying c-instance $\cinstance$. Hence the universal solution for $Q$ is non-empty if and only if the global condition $\phi_Q$ is finitely satisfiable, which is undecidable when $Q$ is a general DRC query. 
\end{proof}
\cut{
If we take two queries, $Q_1$ and $Q_2$, the universal solution for $Q_1 - Q_2$ is empty if and only if $Q_1 \subsetneq Q_2$. 
To show this, consider the case where the universal solution for $Q_1 - Q_2$ is empty. 
Thus, there is no satisfying c-instance for this query, and in particular, there is no ground instance that satisfies it. 
Suppose now that $Q_1 \subsetneq Q_2$, so we know that there is no ground instance satisfying $Q_1 - Q_2$, and therefore, there is no c-instance that satisfies $Q_1 - Q_2$.
}

\revc{Note that the inclusion of $\forall$ operators and arbitrary positioning of $\neg$
    make general DRCs harder than $CQ^{\neg}$. In $CQ^{\neg}$, negations are only allowed in front of relational atoms and conditions subject to standard `safety' constraints \cite{AbiteboulHV95}. The query that returns all beers that are not liked by some drinker: $ \{(b) ~|~ \exists x, d, a~ (\Beer(b, x) \wedge \Drinker(d, a)\wedge \neg \Likes(d, b)) \}$ in the DRC form, and in the equivalent (safe) Datalog with negation form: $Q(b) ~{:-}~ \Beer(b, x)$, $\neg \Likes(d, b),$ $\Drinker(d, a)$, is an example of 
$CQ^{\neg}$.}

Since the problem of finding a universal solution for general DRC queries is undecidable, in Section~\ref{sec:disjunct} we give an algorithm that 
builds an exhaustive minimal c-solution up to a certain limit on the size of the c-instances to ensure halting by checking all possible assignments of variables. Further, we give a more efficient algorithm in Section~\ref{sec:conjunct} that relaxes the requirement of generating all possible c-instances by providing a subset of satisfying c-instances. 


\section{Algorithm for Minimal C-Solution}\label{sec:disjunct}
In this section we show how to compute an exhaustive set of satisfying c-instances up to a size limit for a DRC query $Q$ by adapting ideas from the \emph{chase} procedure \cite{FaginKMP03,FaginKP03}.

\subsection{Basic Notions and Overview}
Given a query syntax tree $Q$ (Definition \ref{def:syntax-tree}), our algorithm constructs an exhaustive set of c-instances in the minimal c-solution by recursively extending each c-instance in multiple ways. 
To explore the different options of extending each c-instance, our algorithm takes a similar approach to that of the chase procedure, 
that was originally proposed for database dependencies \cite{AhoBU79,MaierMS79}. 
It constructs the different options for c-instances using a breadth first search (BFS) procedure, thereby implicitly generating a chase tree \cite{BaranyCKOV17}. 
However, unlike the classic chase algorithm that directly adds or modifies tuples, our procedure converts the syntax tree into a conjunction of atoms and then maps the atoms in the conjunction to tuples and conditions which are added to the c-instance. 
Each quantifier and connector in the tree triggers a tailored recursive call. 
While creating the c-instances, the algorithm keeps track of the mappings between the query variables  and the labeled nulls in the c-instances. 
We next formally define this homomorphism between query and c-instance; since we build the homomorphism in steps, we define it as a partial function.

\begin{Definition}[Homomorphism between query and c-instance]
Given a query $Q$ and a c-instance $\cinstance$ on the same schema and domain $\dom$, with query variables  $\varset_Q$ and labeled nulls $\lnset_{\cinstance}$, and 
constants $\constset_Q, \constset_\cinstance \subseteq \dom$ respectively, a {\bf homomorphism $\mrel$ from $Q$ to $\cinstance$} is a partial function
$h: \varset_Q \cup \constset_Q \rightarrow \lnset_\cinstance \cup \constset_\cinstance$ such that, (1) for each constant $c \in \constset_Q$, $h(c) = c$, and (2) for an atom
$a = \rel(x_1, ..., x_{\relarity})$ in $Q$, if 
all of $h(x_1), \cdots, h(x_k)$ are defined, $h(a) = \rel(h(x_1), ..., h(x_{\relarity}))$ is in relation $\rel$ in $\cinstance$. 
\end{Definition}

As opposed to assignments from queries to ground instances, homomorphisms to c-instances are not restricted to output variables but can also map quantified variables to labeled nulls. This allows the algorithms to consider multiple different homomorphisms of variables to different labeled nulls. For universally quantified variables $\forall x$, the algorithm keeps track of multiple homomorphisms even within the same c-instance. 

\cut{
\begin{Example}
A homomorphism $\mrel$ from the difference query shown in Figure \ref{fig:diff-query} to the c-instance in Figure \ref{fig:c-instance} maps the variables to labeled nulls of the same name as follows $h(x_1) = x_1$, $h(b_1) = b_1$, $h(d_1) = d_1$, $h(p_1) = p_1$ etc.
\end{Example}
}

We also abuse terminology as follows. Given a query $Q$ and a c-instance $\cinstance$ such that there is a homomorphism $\mrel$ from $Q$ to $\cinstance$, we refer to the {\bf domain of a query variable in $Q$} as the set of labeled nulls with the same domain in $\cinstance$ according to $\mrel$, i.e., the labeled nulls in the same attribute or in corresponding attributes in different tables that share the same domain (e.g., by foreign key dependencies). 

\cut{
\begin{figure}
    \centering
    \includegraphics[width=.8\linewidth]{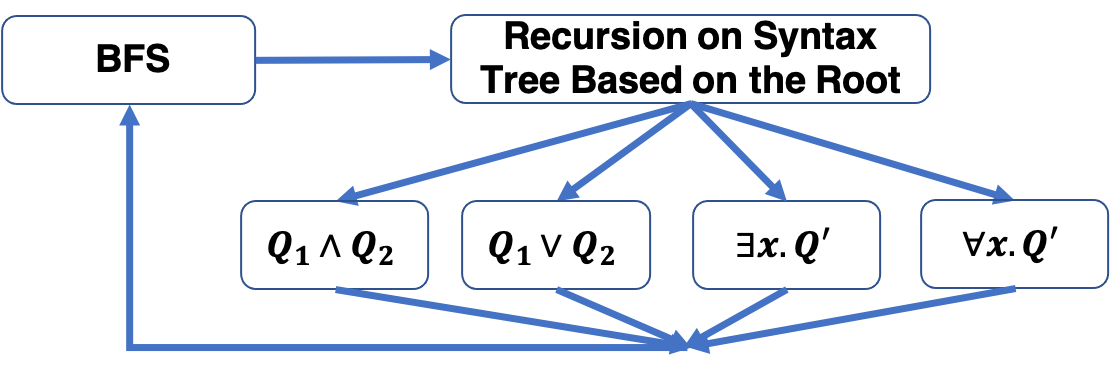}
    \caption{Illustration of our algorithm for generating an exhaustive set of satisfying c-instances. We start with an empty c-instance and generate more c-instances using a BFS approach. The BFS calls a procedure that handles each connector and quantifier separately and recurses over the syntax tree}
    \label{fig:chase_illustration}
\end{figure}
}

The procedure
starts by mapping the free variables of the query to fresh labeled nulls in the c-instance. 
Then, it performs a BFS where for each c-instance in the queue, it expands the homomorphism and the c-instance using a recursive procedure. 
For a syntax tree with no quantifiers, the recursive procedure adds its atoms as tuples to the c-instance, ensuring that the variables in the atoms are converted to their labeled null counterparts according to the homomorphism. 
Otherwise, it handles the syntax tree based on its root: each quantifier ($\forall$, $\exists$) or connective ($\land$, $\lor$) is handled separately. 


To ensure the minimality of each c-instance in the obtained set and the minimality of the set itself (Definitions~\ref{def:minimal-cinstance} and  \ref{def:universal}), we use a post-processing procedure that checks the coverage of each c-instance, and 
for any coverage, it keeps a single c-instance with minimum size (breaking ties arbitrarily). 

\subsection{Exhaustive Chase for C-Instances}\label{sec:complete-algo}
Algorithm~\ref{alg:rd-tree-chase-main} and Algorithm~\ref{alg:rd-tree-chase} form the main body of our `chase' procedure. The procedure starts by calling Algorithm~\ref{alg:rd-tree-chase-main} ($\proc{Tree-Chase-BFS}$) on the schema $\schemaOf{\rel}$, the entire query $Q$, an empty instance $\cinstance_0$, and an empty mapping $h_0$. In addition, the size bound \emph{limit} sets the maximum number of tuples and atomic conditions in the global condition allowed in the c-instance, and is meant to ensure halting of the algorithm since finding a satisfying c-instance for a general DRC query is undecidable (Proposition~\ref{prop:complexity}).  

\begin{algorithm}[t]\caption{Tree-Chase-BFS}\label{alg:rd-tree-chase-main}
{\footnotesize
\begin{codebox}
\algorithmicrequire $\schemaOf{\rel}$: the database schema; 
$Q$: a DRC query;\\ $\cinstance_0$: a c-instance of schema of $\schemaOf{\rel}$;
$h_0$: a mapping; \\
$limit$: the maximum number of tuples and conditions in the c-instance;\\
\algorithmicensure A list of satisfying c-instances for $Q$. \\
\Procname{$\proc{Tree-Chase-BFS}(\schemaOf{\rel}, Q, h_0, \cinstance_0, limit)$}
\li res $\gets []$, queue $\gets $ an empty queue
\li \For $x \in FreeVar(Q)$\label{l:iter-free}
    \Do
          \li Create a fresh variable $x'$ in the domain of $x$
          \li $\cinstance_0.domain(x) \gets \cinstance_0.domain(x) \cup \{x'\}$
          \li $h_0 \gets h_0 \cup \{x \to  x'\}$\label{l:add-fresh}
    \End
\li queue.push($\cinstance_0$)
\li visited $\gets \emptyset$
\li \While $\neg$ queue.isEmpty()
    \Do
        \li $I \gets $ queue.pop()
        \li \If I $\in$ visited or $|I| > limit$
        \li \Then continue
        \End
        \li visited $\gets $ visited $\cup \{I\}$
        \li \If $\proc{Tree-SAT}(Q, I, \emptyset)$ and $\consistent(I)$ 
            \Then  
                \li $res.append(I)$
                \li continue
            \End
        \li $Ilist \gets \proc{Tree-Chase}(\schemaOf{\rel}, Q, I, h_0, limit)$
        \li \For $J \in Ilist$
        \Do
           \li \If 
           $\consistent(J)$ and $|J| \leq limit$ and $J \notin visited$
            \Then
                \li queue.push($J$)
            \End
        \End
    \End
\li \Return $res$
\end{codebox}
}
\end{algorithm}

\mypar{Breadth-first search}
First, to initialize the instance $\cinstance_0$ and the mapping $h_0$ from free variables in the query to labeled nulls and constants in $\cinstance_0$, for each free variable $x$ in $Q$, Algorithm~\ref{alg:rd-tree-chase-main} will create a new labeled null and add it to the domain of $x$ in $\cinstance_0$ and update $h_0$ (Line 2-5). Then, the algorithm runs in a Breadth-first search manner: $\cinstance_0$ is initially added to the empty queue; every time the algorithm takes the c-instance from the head of the queue, checks whether the instance has already been generated and its size does not exceeds $limit$ (Line 10). 
The procedure for checking $I \in visited$ takes into account renaming of variables; it first compares certain properties of the c-instances (e.g., number of tuples, size of conditions etc.) and filters out candidates that cannot be equivalent to $I$, and then it checks all possible mappings to previously generated c-instances. 
It also checks (Line 13) (1) whether $I \models Q$ by the \proc{Tree-SAT} procedure, and 
(2) whether it is consistent, i.e., $\Rep(I) \neq \emptyset$, (we use an SMT solver in our implementation). 
It then runs the recursive procedure on the current c-instance (Line 16) and adds each one of the resulting c-instances to the queue (Lines 17--19).

\newcounter{lastcodelinenumber}

\newif\ifcodeboxcontinued

\xapptocmd{\endcodebox}{%
  \setcounter{lastcodelinenumber}{\value{codelinenumber}}%
  \global\codeboxcontinuedfalse
}{}{}
\xapptocmd{\codebox}{%
  \ifcodeboxcontinued
  \setcounter{codelinenumber}{\value{lastcodelinenumber}}%
  \fi
}{}{}
\newcommand\continuecodebox{\global\codeboxcontinuedtrue}

\begin{algorithm}[t]\caption{Tree-Chase }\label{alg:rd-tree-chase}
  \centering
{\footnotesize
\begin{codebox}
\algorithmicrequire $\schemaOf{\rel}$: the database schema; 
$Q$: the query as its syntax tree; \\
$I$: current c-instance; 
$\mrel$: current homomorphism from $Q$ to $I$;\\
$limit$.\\
\algorithmicensure a list of c-instances\\
  \Procname{$\proc{Tree-Chase}(\schemaOf{\rel}, Q, I, h, limit)$}
  \li res $\gets []$
  \li \If there are no quantifiers  in $Q$
  \Then
    \li $L \gets \proc{tree-to-conj}(Q)$
    \li \For $\psi \in L$
    \Do
        \li $J \gets \proc{Add-to-Ins}(\schemaOf{\rel}, I, h(\psi))$
        \li \If $\consistent(J)$ 
            \Then \li res.append($J$)
        \End
    \End
  \li \ElseIf Q.root.operator $\in \{\land\}$
  \Then
        \li res $\gets \revc{\proc{Handle-Conjunction}}(\schemaOf{\rel}, Q, I, h, limit)$
\li \ElseIf Q.root.operator $\in \{\lor\}$
  \Then
        \li res $\gets \revc{\proc{Handle-Disjunction}}(\schemaOf{\rel}, Q, I, h, limit)$
  \li \ElseIf Q.root.operator $\in \{\exists\}$
  \Then
        \li res $\gets \proc{Handle-Existential}(\schemaOf{\rel}, Q, I, h, limit)$
  \li \ElseIf Q.root.operator $\in \{\forall\}$
  \Then 
        \li res $\gets \proc{Handle-Universal}(\schemaOf{\rel}, Q, I, h, limit)$
    \End\label{alg:rd-tree-chase:line:forall-add-fresh-end}
\End
\li \Return res

\end{codebox}
}
\end{algorithm}

\mypar{Recursive generation of c-instances}
The recursive procedure Algorithm~\ref{alg:rd-tree-chase} (\proc{Tree-Chase}) handles the query according to its root operator. 
It gets as input the schema of the relational database $\schemaOf{\rel}$, the syntax tree of a DRC query $Q$, the current c-instance $I$,
and the current homomorphism from $Q$ to $I$.
For the case where the query has no quantifiers (Line 2-7), the algorithm converts the syntax tree into a list of conjunction of atoms/atomic conditions, and then an instance is created for each conjunction under the homomorphism $\mrel$ by \proc{Add-to-Ins}. 
The algorithm proceeds according to 
the root of the syntax tree ($\wedge, \vee, \exists, \forall$) by calling procedures that get the same input as Algorithm \ref{alg:rd-tree-chase}, call Algorithm~\ref{alg:rd-tree-chase-main} recursively, and output a list of c-instances that are sent back to Algorithm~\ref{alg:rd-tree-chase-main}. 

\begin{algorithm}[t]\caption{\revc{Handle-Conjunction}}\label{alg:and-op}
  \centering
{\footnotesize
\begin{codebox}
  \Procname{$\proc{Handle-Conjunction}(\schemaOf{\rel}, Q, I, h, limit)$}
    \li res $\gets []$
    \li lres $\gets \proc{Tree-Chase-BFS}(\schemaOf{\rel}, Q.root.lchild, I, h, limit)$
    \li \For $J \in lres$
    \Do
        \li \If $\consistent(J) = false$ 
        \Then
            \li Continue
        \End
        \li rres $\gets \proc{Tree-Chase-BFS}(\schemaOf{\rel}, Q.root.rchild, J, h, limit)$
        \li \For $K \in rres$ 
        \Do
            \li \If $\consistent(K)$ 
            \Then
              \li res.append($K$)
            \End
        \End
    \End
\li \Return $res$
\end{codebox}
}
\end{algorithm}

\mypar{Handling conjunction ($\wedge$)}
The \proc{Handle-Conjunction} procedure recursively calls \revc{Algorithm \ref{alg:rd-tree-chase-main}} on both children of the root, 
and every pair of solutions to each child is merged into one instance by taking a union of every solution to the left subtree with every solution to the right subtree, 
and adding the consistent instances to the list of c-instances. 
\cut{Note that although the mapping 
$\mrel$ guarantees that solutions to each child map the free variables in $Q$ to the same labeled nulls in the instance, there can be newly created labeled nulls in the recursive calls and hence we have to validate the consistency of the resulting instance. 
}

  

\mypar{Handling disjunction ($\vee$)}
\cut{For the case where the root node is $\lor$, Algorithm \ref{alg:rd-tree-chase} calls the algorithm that handled the disjunctive connector\footnote{\label{note1}The pseudo code can be found in the full version \cite{?} \amir{cite}} which uses a procedure to 
}
The \proc{Handle-Disjunction} procedure (Algorithm \ref{alg:or-op})
reduces the disjunctive tree into three conjunctive trees by replacing $Q =Q_1 \lor Q_2$ with 
$Q_1 \land Q_2$, $Q_1 \land \neg Q_2$, and $\neg Q_1 \land Q_2$, 
since one of the three formulae is True iff $Q_1 \lor Q_2$ is True. 
This conversion introduces negation to some subtrees, such a negated subtree 
is translated into a 
syntax tree with negation only on the leaves.
Then, Algorithm \ref{alg:rd-tree-chase-main} is called with each of the modified trees, and each set of c-instances obtained from 
the three cases is added to the result set.

\mypar{Handling existential ($\exists$) and universal ($\forall$) quantifiers}
If the root node 
has $\exists x$, Algorithm~\ref{alg:rd-tree-chase} calls  
\proc{Handle-Existential} (Algorithm \ref{alg:exists-op}) that iterates over all labeled nulls or constants in the domain of 
$x$, also adds a fresh labeled null, updates the homomorphism (as the quantified variable becomes free in the subtree), and recursively calls Algorithm~\ref{alg:rd-tree-chase-main}. 
\proc{Handle-Universal} (Algorithm \ref{alg:forall-op})
handles the case when the root has $\forall x$. 
The difference with $\exists$ is that, like the $\land$ case, the solutions to all 
labeled nulls and constants that 
$x$ is mapped to are merged into one instance.
The algorithm first checks whether there is no root and returns the inputted c-instance in that case (Lines 2--3) adds each mapping from $x$ to a labeled null or constant to the homomorphism, runs the recursive procedure to find all c-instances with this mapping and merges it with other c-instances generated with other homomorphism that  map $x$ to other labeled nulls or constants (Lines 5--14). 
It further generates c-instances by mapping $x$ to a fresh labeled null (Lines 19--24).

\mypar{Ensuring minimality in post-processing}
After Algorithm \ref{alg:rd-tree-chase-main} returns a set of c-instances, we remove the c-instances that are not minimal in the following manner. 
For each c-instance in the set, we compute a hash string for its coverage (we keep track of the coverage of each c-instance as it is created). 
Then, for each c-instance in the set, we get all other c-instances in the set with the same string representing its coverage and remove all but the minimal one according to their size. 
\revc{Note that the hash function is applied to the {\em coverage of the c-instance rather than the c-instances themselves}, i.e., the function hashes the covered atoms of the query. Thus, it allows us to efficiently detect c-instances with the same coverage and remove those that are not minimal in terms of their number of tuples and atomic conditions (ref.  
Section \ref{sec:characterization}).}

\paratitle{Soundness, termination, and complexity.}
Given a syntax tree $Q$ of a DRC query, Algorithm~\ref{alg:rd-tree-chase-main} when given $(\schemaOf{R}, Q, \emptyset, \emptyset)$ will output a list of c-instances that are consistent, minimal, and satisfy the query denoted by $Q$ (validated in Line 13), i.e., our procedure is sound and generates a valid minimal c-solution. 
\par
Although the problem of verifying if a satisfying c-instance exists is undecidable (Proposition~\ref{prop:complexity}),  Algorithm~\ref{alg:rd-tree-chase-main} is  guaranteed to terminate given the $limit$ parameter. There are finitely many distinct c-instances (that are not isomorphic in terms of renaming of variables)  up to size $limit$ given a query. If the size of a c-instance increases over $limit$, the algorithm will ignore this c-instance and not push it into $queue$ (Lines 10-11). 
The algorithm will also not get into an infinite loop because of the $visited$ set. Every generated c-instance is placed into this set and c-instance already found is the set are ignored (including renaming of variables) and are not pushed into $queue$ (Lines 18--19). Since the size of the schema is constant, given a limit on size of the c-instances, the domains of all labeled nulls in the c-instances are also finite since
the domain is derived from existing labeled nulls in the c-instance. 
\par
The running time of the algorithm is exponential in the number of operators and size of the c-instances (bounded by $limit\ \times$ no. of relations $\times$ max no. of attributes) since the algorithm performs an exhaustive search on c-instances subject to the size limit, resulting in a high complexity. This motivates us to design a more efficient algorithm by generating a possibly smaller minimal c-solution that we describe in the next subsection.

\cut{

\subsection{Analysis of Algorithm \ref{alg:rd-tree-chase-main}}\label{sec:analysis}

\mypar{Soundness and completeness}
Given a syntax tree $Q$ of a DRC query, Algorithm~\ref{alg:rd-tree-chase-main} when given $(\schemaOf{R}, Q, \emptyset, \emptyset)$ will output a list of c-instances that are consistent and satisfy the query denoted by $Q$ (validated in Line 13), i.e., our procedure is sound. 


\cut{
\begin{Theorem}[Soundness]\label{thm:soundness}
Let $Q$ be a syntax tree of a DRC query over a schema $\schemaOf{R}$, and let $\cinstance$ be a c-instance in the set outputted by Algorithm \ref{alg:rd-tree-chase-main} when given $(\schemaOf{R}, Q, \emptyset, \emptyset)$, then $\cinstance$ satisfies $Q$.
\end{Theorem}

\begin{proof}[proof sketch]
To prove the proposition, it is enough to show that in Line 14 of Algorithm \ref{alg:rd-tree-chase-main} the procedure for checking that the c-instance $I$ satisfies the query $Q$\footnote{The pseudo code of the procedure, along with a proof of correctness are included in the full version \cite{?} \amir{cite}}. 
\end{proof}
}


\begin{Proposition}
[Completeness]\label{theorem:completeness}
Given a syntax query tree $Q$ and a ground instance $\ains$ such that $\ains \models Q$ and 
its coverage $\coverage(Q, k) = \cov$, and there exists a c-instance $\cinstance^\star$ of size $\leq limit$ with coverage $\cov$, 
then the set of c-instances $\mathcal{S}_{\cinstance} = \{\cinstance_1, \ldots, \cinstance_k\}$ returned by Algorithm \ref{alg:rd-tree-chase-main} with the parameter $limit$ contains a c-instance $\cinstance_i \in \mathcal{S}_{\cinstance}$ with coverage $\cov$.

\cut{
\red{Given a syntax query tree $Q$ and the set of c-instances $\mathcal{S}_{\cinstance} = \{\cinstance_1, \ldots, \cinstance_k\}$ returned by Algorithm \ref{alg:rd-tree-chase-main} with  parameter $limit$, if there is a ground instance $K$ such that $K \models Q$ and 
its coverage $\coverage(Q, k) = \cov$, if there exists a c-instance of size $\leq limit$ with coverage $\cov$
then there exists a c-instance $\cinstance_i \in \mathcal{S}_{\cinstance}$ with coverage $\cov' = \cov$.} \sr{please check -- how does limit affect K?}

Given a syntax query tree $Q$ and the set of c-instances $\mathcal{S}_{\cinstance} = \{\cinstance_1, \ldots, \cinstance_k\}$ outputted by Algorithm \ref{alg:rd-tree-chase-main} with the parameter $limit$, if a ground instance $\ains$ that satisfies $Q$ with coverage $\cov$,
then there exists $\cinstance_i \in \mathcal{S}_{\cinstance}$ with coverage $\cov' = \cov$. \amir{Rephrased. Verify}
}
\end{Proposition}

\sr{do not follow the proof yet -- stopping here}
\begin{proof}[proof sketch]
The proof is by induction over the size of the syntax tree, where the base case is proven by considering a c-instance of size $1$ generated in lines 2--7 in Algorithm \ref{alg:rd-tree-chase}. 
The induction then proceeds to consider the cases where $Q = Q_1\land Q_2$, $Q = Q_1\lor Q_2$, $Q = \exists x.~Q'$, and $Q = \forall x.~Q'$. 

\amir{This can be removed if we need space:}
As an example, we give the proof for the conjunctive case: Assume $Q = Q_1 \land Q_2$ is of size $n$ and there is a ground instance $\ains$ covering $Q$ with coverage $\cov_\ains$. 
Thus, $\ains$ covers $Q_1$ with coverage $\cov^1_\ains$ and covers $Q_2$ with coverage $\cov^2_\ains$. 
Based on the induction hypothesis, there exist $\cinstance_1, \cinstance_2$ found by Algorithm 
\ref{alg:rd-tree-chase-main} with coverages $\cov^1_\ains$ and $\cov^2_\ains$ respectively (as $Q_1$ and $Q_2$ both have size smaller than $n$). 
To compute the c-instances for $Q$, Algorithm 
\ref{alg:rd-tree-chase-main} calls Algorithm 
\ref{alg:rd-tree-chase} that calls Algorithm 
\ref{alg:and-op} to handle the conjunction. 
So, the c-instance $\cinstance_1$ is part of the list generated in line 1 of the algorithm, and  $\cinstance_2$ is merged into it in line 5. 
Our goal is then to show that the algorithm generates a new c-instance $\cinstance^\star$ that is  of $\cinstance_1$ and $\cinstance_2$ is satisfiable and has coverage $\cov_\ains$. 
The merge is defined as adding the tuples in $\cinstance_1$ to $\cinstance_2$ along with its mappings, and joining their global conditions using a conjunction \amir{check with Zhengjie}. 

Since $\ains$ is a ground instance, it cannot contain a contradiction and cannot satisfy contradicting conditions of $Q$, even from different assignments. 
The conditions included in $\cinstance_1$ and in $\cinstance_2$ are subsets of the conditions satisfied by $\ains$ and their conjunction is also satisfied by $\ains$ and thus they do not contradict each other. 

For coverage, observe that $\cov_\ains = \cov^1_\ains \cup \cov^2_\ains$. This is exactly the coverage of $\cinstance^\star$, obtained from merging $\cinstance_1$ and $\cinstance_2$.
\end{proof}

\mypar{Termination}
Recall that in Section \ref{sec:model} we have determined that that the problem of finding a universal solution (Definition \ref{def:universal}) is undecidable. 
Despite this, Algorithm \ref{alg:rd-tree-chase-main} is guaranteed to terminate due to the $limit$ parameter (Line 11) and the $visited$ set.

\begin{proposition}
Given a query syntax tree $Q$, Algorithm \ref{alg:rd-tree-chase-main} is guaranteed to terminate.
\end{proposition}

\begin{proof}
Since the $limit$ parameter is finite, there are finitely many options for c-instances that it can generate. If the size of the c-instance increases over $limit$, the algorithm will ignore this c-instance and not push it into $queue$ (Lines 11--12). 
The algorithm will also not get into an infinite loop as well because of the $visited$ set. Every generated c-instance is placed into this set and c-instance already found is the set are ignored and are not pushed into $queue$ (Lines 11--12). 
Finally, note that the domain of all labeled nulls in the c-instance is derived from existing labeled nulls representing the same attribute in the c-instance, and thus cannot be infinite.
\end{proof}

\mypar{Complexity}
\amir{Verify:} 
The complexity of Algorithm \ref{alg:rd-tree-chase-main} is $\Theta(3^{d \cdot |Dom|^{u + e}} \cdot |Dom|^2)$, where $d$ is the number of disjunctions, $|Dom|$ is the number of cells in the c-instance (the maximum number is $limit \cdot k \cdot \schemaOf{\rel}$ where $k$ is the number of tables in the c-instance), and $u$ ($e$) is the number of universal (existential) quantifiers. This bound is also tight. 
To explain why it is a lower bound, consider the query $\forall x_1 \ldots \forall x_m R(x_1) \lor R(x_2) \lor \ldots \lor R(x_m)$. Given this query, denoted by $Q$, Algorithm \ref{alg:rd-tree-chase-main} first generates a c-instance, where the free variables in $Q$ are mapped to fresh variables in the domain (there are no free variables in $Q$). Then, it performs a BFS using Algorithm \ref{alg:rd-tree-chase}. 
Algorithm \ref{alg:rd-tree-chase} recurses over the tree, and for each $\forall$ quantifier, it calls Algorithm \ref{alg:forall-op} that performs $|Dom|$ iterations and for each iteration, generates a tree (due to the loop in Line 5 and the recursive call in Line 7). For each tree, the algorithm iterates over the c-instances generated in Lines 9--10, resulting in the need for $|Dom|^2$ iterations. 
Then, for each $\lor$ operator, it calls Algorithm \ref{alg:or-op} that expands the subtree into $3$ trees resulting in $3^d$ trees. For each one of these trees, Algorithm \ref{alg:and-op} that performs a recursive call on the left child resulting in $|Dom|$ instances and for each of these instances, recurses on the right child (Line 5). Thus, we have $(3^{|Dom|^u})^{d} \cdot |Dom|^2$ operations. 
Since $\land$ does not expand the current number of trees and $\exists$ expands the number by $|Dom|$, the above expression is an upper bound on the number of trees as well. 

}

\subsection{Optimization by Conjunctive Tree Chase}\label{sec:opt-conj-tree-chase} 
\label{sec:conjunct}

\cut{
In the previous section, we have 
given a solution to obtain a universal solution subject to a size constraintx1, h(), however, the solution involves enumerating all possible c-instances that can be generated from the syntax tree. This leads to slow runtimes in practice (see Section \ref{sec:experiments} \amir{Do we say that there?}). 
For a more interactive approach we propose a solution that forgoes completeness but maintains soundness. 
At its heart, this 
}

The optimized approach converts the original syntax tree into a set of syntax trees where each tree does not contain disjunctions ($\vee$), then performs the chase procedure described in Algorithm \ref{alg:rd-tree-chase-main} on each one of the trees. This speeds up the solution dramatically since there is no need to expand every disjunctive operator in the tree into a set of trees that do not contain disjunction, thus simplifying the process of generating c-instances.
\cut{
Our approach can be summarized as follows: (1) convert the syntax tree into a set of conjunctive trees, (2) perform the chase procedure described in Algorithm \ref{alg:rd-tree-chase-main} on each one of the trees. 
}

\mypar{Conversion of a tree with $\vee$ into conjunctive trees}
For sets, $Q_1\lor Q_2$ is equivalent to three sets $\{Q_1\land Q_2, \neg Q_1\land Q_2, Q_1\land \neg Q_2\}$. 
However, this is not always true for FOL formulae, since $f_1 = \forall x (P(x) \vee Q(x))$ is not equivalent to $f_2 = (\forall x P(x)\wedge Q(x)) \vee (\forall x \neg P(x)\wedge Q(x)) \vee (\forall x P(x)\wedge \neg Q(x))$ and only $f_2 \Rightarrow f_1$ holds; 
\revc{
we demonstrate this below, first with a toy example and then using our running example.

\begin{Example}[Toy example]\label{ex:loss-of-completeness-toy}
Consider $f=\forall x~ (even(x) \lor odd(x))$ where the predicate $even(x)$ ($odd(x)$) is true if $x$ is even (odd). Suppose the domain of $x$ is $\{3,4\}$. $f$ is clearly satisfied with this domain. 
Now consider the conversion of $f$ into three conjunctions: $f_1 = \forall x~ (even(x) \land odd(x))$, $f_2 = \forall x~ (\neg even(x) \land odd(x))$, and $f_3 = \forall x~ (even(x) \land \neg odd(x))$. 
$f_1$ is not satisfied since neither $3$ nor $4$ are both even and odd, $f_2$ is not satisfied since $4$ is even, and $f_2$ is not satisfied since $3$ is odd, thereby losing completeness in the conversion.
\end{Example}
}

\begin{Example}\label{ex:loss-of-completeness-running}
Consider the syntax tree of $\qincorrect - \qcorrect$ in Figure~\ref{fig:syntax-tree-minus}, and a sub-formula of its right subtree:
$$\forall p_3 \big(\neg \Serves(x_1, b_1, p_3) \lor 
\exists x_3, p_4 (\Serves(x_3, b_1, p_4) \land p_3 < p_4 ) \big)$$ 
We can convert this into the following three formulae:
$$\{\forall p_3 \big(\neg \Serves(x_1, b_1, p_3) \land \exists x_3, p_4 (\Serves(x_3, b_1, p_4) \land p_3 < p_4 ) \big),$$ $$\forall p_3 \big(\neg \Serves(x_1, b_1, p_3) \land 
\forall x_3, p_4 (\neg \Serves(x_3, b_1, p_4) \lor p_3 \geq p_4 ) \big),$$ $$\forall p_3 \big(\Serves(x_1, b_1, p_3) \land 
\exists x_3, p_4 (\Serves(x_3, b_1, p_4) \land p_3 < p_4 ) \big)\}.$$

These three formulae are equivalent to the original one if there is only one variable that $p_3$ can be mapped to in the price domain (in this case the original formula would be unsatisfiable). However, when there are more than two variables or constants in the domain, this conversion is not equivalence-preserving and will miss satisfying c-instances. 
For example, the c-instance $\cinstance_0$ in Figure \ref{fig:c-instance} satisfies the original formula but does not satisfy any of the three conjunctive formulae - assigning $p_1$ from $\cinstance_0$ to the universally quantified $p_3$ will not satisfy any of the formulae.
\end{Example}

Bearing this in mind, 
we describe an algorithm that performs this conversion.
Given a syntax tree $Q$, the algorithm converts it into a set of conjunctive syntax trees based on the above principal. 

We start with a syntax tree $Q$ that may contain the $\lor$ operator in different nodes. 
The algorithm recurses on the tree where the base case is a $Q$ that contains a single atom, then the algorithm simply creates a conjunctive tree from this atom, or its negation if it is negated. 
If the root of $Q$ is an $\land$ node, the algorithm continues to recurse over the two children and joins each pair of obtained subtrees.
If the root of $Q$ is an $\lor$ node, the algorithm considers three cases, as mentioned above: (1) converting the root into $\land$ and recursing over both of its children, 
(2) converting the root into $\land$, negating the right child, and recursing over both of its children, 
and, finally, (3) converting the root into $\land$, negating the left child, and recursing over both of its children, 
All solutions to the three cases are added to the list of c-instances $res$. 
If the root of $Q$ is a $\forall$ or $\exists$ quantifier, the algorithm recurses over the child of the root and adds the resulting trees to $res$. 


\mypar{Chase for conjunctive trees}
The main chase procedure utilizes Algorithm \ref{alg:rd-tree-chase-main} and applies it on the conjunctive tree obtained from the conversion algorithm.
It gets a schema, a syntax tree, and a limit as input. 
\revb{It first converts} the input tree into a set of conjunctive trees using 
the conversion algorithm.
It then calls Algorithm \ref{alg:rd-tree-chase-main} with each one of the conjunctive trees and adds the resulting c-instances to the list of results which is then outputted. 

\cut{
\begin{algorithm}[t]\caption{Conjunctive-Tree-Chase-Step}\label{alg:conj-tree-chase-step}
{\small

\begin{codebox}
\algorithmicrequire $R$: the schema of the relational database; \\
$Q$: a conjunctive tree of a DRC query; \\
$I$: current c-instance; \\
$f$: current mapping from $var(Q)$ to $var(I) \cup cons(I)$.\\
\algorithmicensure A c-instance\\
  \Procname{$\proc{Conj-Tree-Chase-Step}(R, Q, I, f)$}
  \li \If Q.root is an atom $r(x_1, ..., x_n)$ or negation of an atom $\neg r(x_1, ..., x_n)$
  \Then 
    \li Add $f(Q.root) to I$ if it has not been added
    \li \If $I$ is satisfiable
        \Then
            \li \Return $I$
        \li \Else
            \li \Return UNSAT
        \End
  \End
  \li \If Q.root.operator $\in \{\land\}$
  \Then
    \li $J \gets \proc{Conj-Tree-Chase-Step}(R, Q.root.lchild, I, f)$
   \li $J \gets \proc{Conj-Tree-Chase-Step}(R, Q.root.rchild, J, f)$
     \li \If $J$ is UNSAT
        \Then
            \li \Return UNSAT
    \li \Else
            \li \Return $J$
        \End
  \li \ElseIf Q.root.operator $\in \{\exists\}$
  \Then
    \li \For $x \in I.domain(Q.root.variable)$
    \Do
        \li $g \gets f \cup \{Q.root.variable \to x\}$
        \li J $\gets \proc{Conj-Tree-Chase-Step}(R, Q.root.lchild, I, g)$ 
        \li \If $J$ is not UNSAT
        \Then
            \li \Return $J$
        \End
    \End
    \li Add a fresh variable $x'$ to $I.domain(Q.root.variable)$
    \li $g \gets f \cup \{Q.root.variable \to x'\}$
    \li J $\gets \proc{Conj-Tree-Chase-Step}(R, Q.root.lchild, I, g)$ 
        \li \If $J$ is UNSAT
        \Then
            \li \Return UNSAT
        \li \Else
            \li \Return $J$
        \End
  \li \ElseIf Q.root.operator $\in \{\forall\}$
  \Then 
    \li \If $I.domain(Q.root.variable) = \emptyset$
        \Then
            \li Add $Q.root.variable$ to $I.domain(Q.root.variable)$
        \End
    \li $J \gets I$
    \li \For $x \in I.domain(Q.root.variable)$
    \Do
        \li $g \gets f \cup \{Q.root.variable \to x\}$
        \li $J \gets \proc{Conj-Tree-Chase-Step}(R, Q.root.lchild, J, g)$ 
        \li \If $J$ is UNSAT
        \Then
            \li \Return UNSAT
        \End
    \End
    \li \Return $J$
    
  \End
\end{codebox}
}
\end{algorithm}
}

\mypar{Soundness and complexity}
The soundness of the algorithm, i.e., that it returns a set of minimal satisfying c-instances, follows from the fact that the final set is returned by Algorithm~\ref{alg:rd-tree-chase-main}. Although there is an exponential dependency in the number of operators, the algorithm gives a better running time by avoiding recursive calls to split a query tree into three query trees recursively for the disjunction operator, at the cost of possibly not generating some satisfying c-instances that are generated by Algorithm~\ref{alg:rd-tree-chase-main}.
\cut{
We next prove that every c-instance outputted by the conjunctive version of the chase algorithm satisfies $Q$ and that each c-instance outputted by Algorithm 

For soundness, note that this is immediately implied by the soundness of Algorithm \ref{alg:rd-tree-chase-main}, as we call it for each one of the conjunctive trees. 

While Example \ref{ex:} shows that the conversion to conjunctive trees hinders completeness, a limited version of completeness can also be shown. Following Proposition \ref{theorem:completeness}, we can conclude that the chase procedure outputs a complete solution for each individual conjunctive tree.

\mypar{Complexity}
The complexity of Algorithm \ref{alg:conj-tree-cons} is 
$\mathcal{O}(3^{d}\cdot |Dom|^{u + e + 2})$
where $d$ is the number of disjunctions in $Q$, $|Dom|$ is the maximum domain size, and $u$ ($e$) is the number of universal (existential) quantifiers. 
Since we only consider conjunctive trees, Algorithm \ref{alg:rd-tree-chase} always skips Line 13 and never calls Algorithm \ref{alg:or-op}. Thus, there is no extension of the original tree beyond the one done by Algorithm \ref{alg:conj-tree-cons}. 
The first factor is due to Algorithm \ref{alg:conj-tree-cons} generating $3$ trees for each $\lor$ operator. The second factor stems from the iterations over the entire domain size for each of the universal quantifiers in each conjunctive tree when Algorithm \ref{alg:or-op} is called.
}




\mypar{Other optimizations}
The time complexity is also largely affected by the number of labeled nulls in each domain, especially when handling universal quantifiers. Hence, we try to optimize the algorithm by disallowing the universal quantifier to add new labeled nulls, though this might lose completeness. Moreover, notice that our $\proc{Tree-Chase-BFS}$ is always initially called with $\cinstance_0 =$an empty c-instance, we could manipulate it to achieve different coverage by calling $\proc{Tree-Chase-BFS}$ on a c-instance that is properly initialized with the tuples or atomic conditions we target to cover. We further evaluate these optimizations in Section~\ref{sec:experiments}.

\mypar{\revb{Setting the limit parameter in Algorithm \ref{alg:rd-tree-chase-main}}}
\revb{The $limit$ parameter can be set in several different manners in practice. 
One approach is setting a default $limit$ according to query complexity. $limit$ determines the maximal size of the c-instance. In our experimental results (Section \ref{sec:experiments}) we have seen that it can be set to some multiple of the number of query atoms to safely allow for a c-instance to covers all the atoms of the query (we have used a multiple of 2 in our evaluation). 
Another alternative, aimed at an interactive experience, is to set a timeout parameter instead of the $limit$ (as done in Section \ref{sec:experiments}), thus allowing Algorithm \ref{alg:rd-tree-chase-main} to explore higher limits as needed up to the allotted time. 
}
\section{Experiments}\label{sec:experiments}

We investigate the performance of our approach and compare it to different variations of our approach 
in the following aspects: (1) runtime for varying query complexity (2) varying the limit parameter in Algorithm \ref{alg:rd-tree-chase-main}, (3) properties of the output c-instances, and (4) case studies showing the actual obtained c-instances for a sample of the experimental queries.

\mypar{Setup} We implemented our methods in Python 3.7.
We ran all experiments locally on a 64-bit Ubuntu 18.04
LTS server with 3.20GHz Intel Core i7-8700 CPU and 32GB
2666MHz DDR4 RAM. We compare
the following variants of our algorithms and optimizations. 

\begin{itemize}
\item \textbf{\disjNaive{}}: this method implements the 
exhaustive chase procedure described in Section~\ref{sec:complete-algo}.
\item \textbf{\conjNaive{}}: this method implements the 
optimized conjunctive tree chase procedure described in Section~\ref{sec:conjunct} that converts the original syntax tree into a set of syntax trees without disjunction.
\item \textbf{\disjEO{}/\conjEO{}}: this method adapts \disjNaive{}/\conjNaive{} by only allowing the algorithm to add 
labeled nulls to the c-instance when handling an existential quantifier node.
\item \textbf{\disjAdd{}/\conjAdd{}}: this method first runs \disjEO{}(or \conjEO{}) on the empty c-instance, then gets the minimal c-solution. If there are still leaf atoms not covered by any of the c-instance in the c-solution, then it iterates over every remaining uncovered leaf atom, creating corresponding labeled nulls and adding the atom to the initial c-instance, and runs \disjEO{} (\conjEO{} resp.) on the 
initialized c-instance.
\end{itemize}

We evaluate both the efficiency/scalability of our algorithms in
terms of runtime and the quality of results with respect to different
measures, and compare them against systems from related work.

\mypar{Datasets} We used two datasets in our experiments, \textbf{Beers} and \textit{TPC-H}.
For Beers dataset,
queries in these experiments come from submissions by students for 
an assignment in an undergrad database course. 
We picked 5 questions (skipped those with only simple selection and join) and sampled a few students' queries, then manually rewrote them into domain relational calculus.
There were 5 (correct) standard queries and 10 students' wrong queries; we also considered the difference between the standard queries and the wrong queries (and also the opposite direction),
resulting in additional 20 queries.
Some queries are very complex as they use the difference operator multiple times,
resulting in nested universal quantification in the DRC query. 
 \common{Similarly, for TPC-H, we picked 4 queries (Q4, Q16, Q19, and Q21) and dropped their aggregate functions, then made two wrong queries each, resulting in 28 test queries in total. The statistics of the datasets are in Table~\ref{tab:datasets}.}

\begin{table}[t]
\centering
\small
\resizebox{0.47\textwidth}{!}{
\begin{tabular}{cccccc}\toprule
Dataset & \# Queries & Mean \# Atoms  & Mean \# Quantifiers & Mean \# Or & Mean Height   \\ \midrule
Beers & 35 & 6.40 & 13.94 & 2.17 & 9.54 \\
TPC-H & 28 & 11.96 & 23.07 & 4.18 & 12.07
\\ \bottomrule
\end{tabular}}
\common{\caption{\label{tab:datasets}\small Dataset statistics.}}
\vspace{-8mm}
\end{table}

\begin{figure*}
\begin{minipage}{0.96\linewidth}
\begin{minipage}{0.24\linewidth}
  \centering
  \includegraphics[width=\linewidth]{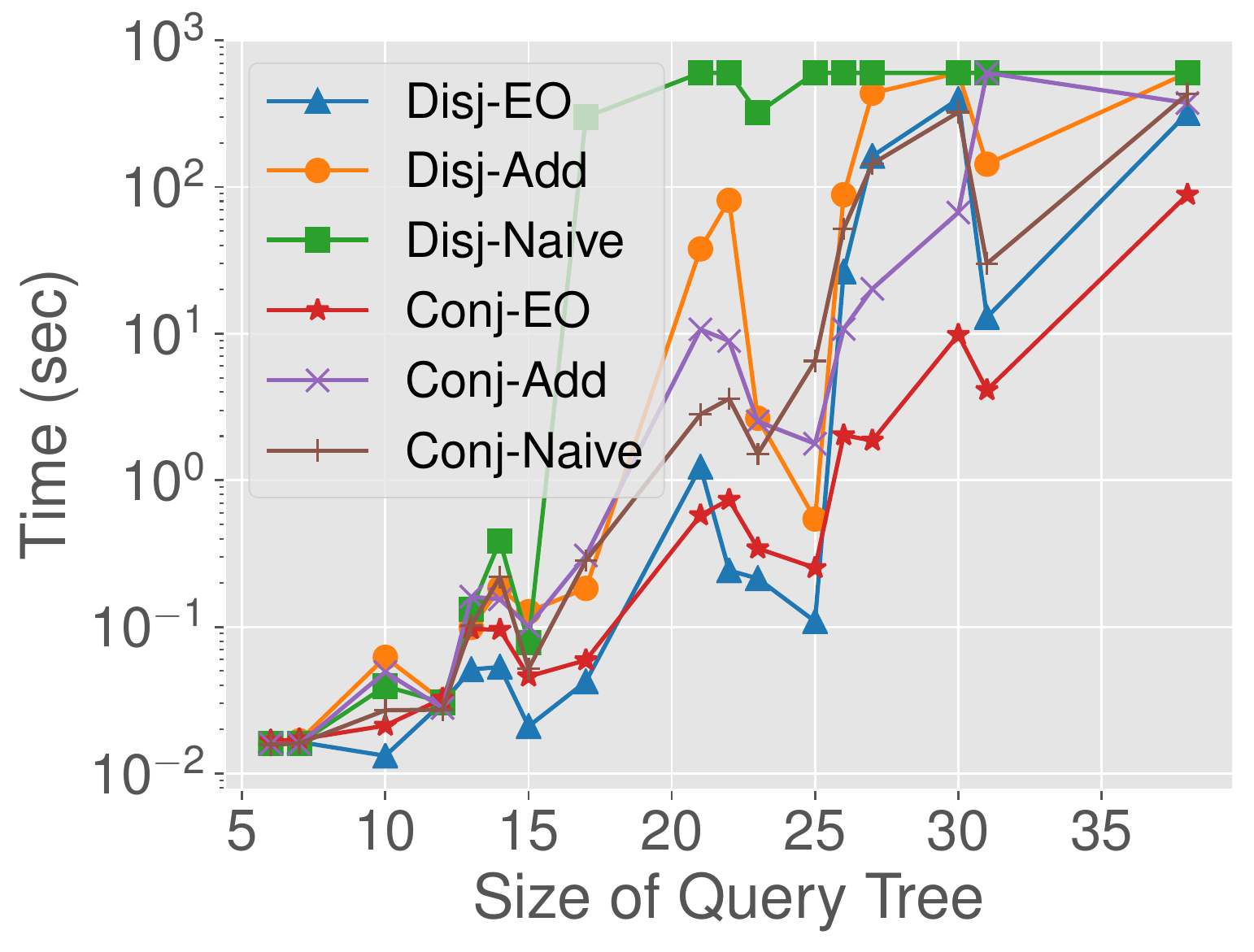}
\end{minipage}
  \hfill
\begin{minipage}{0.24\linewidth}
  \centering
  \includegraphics[width=\linewidth]{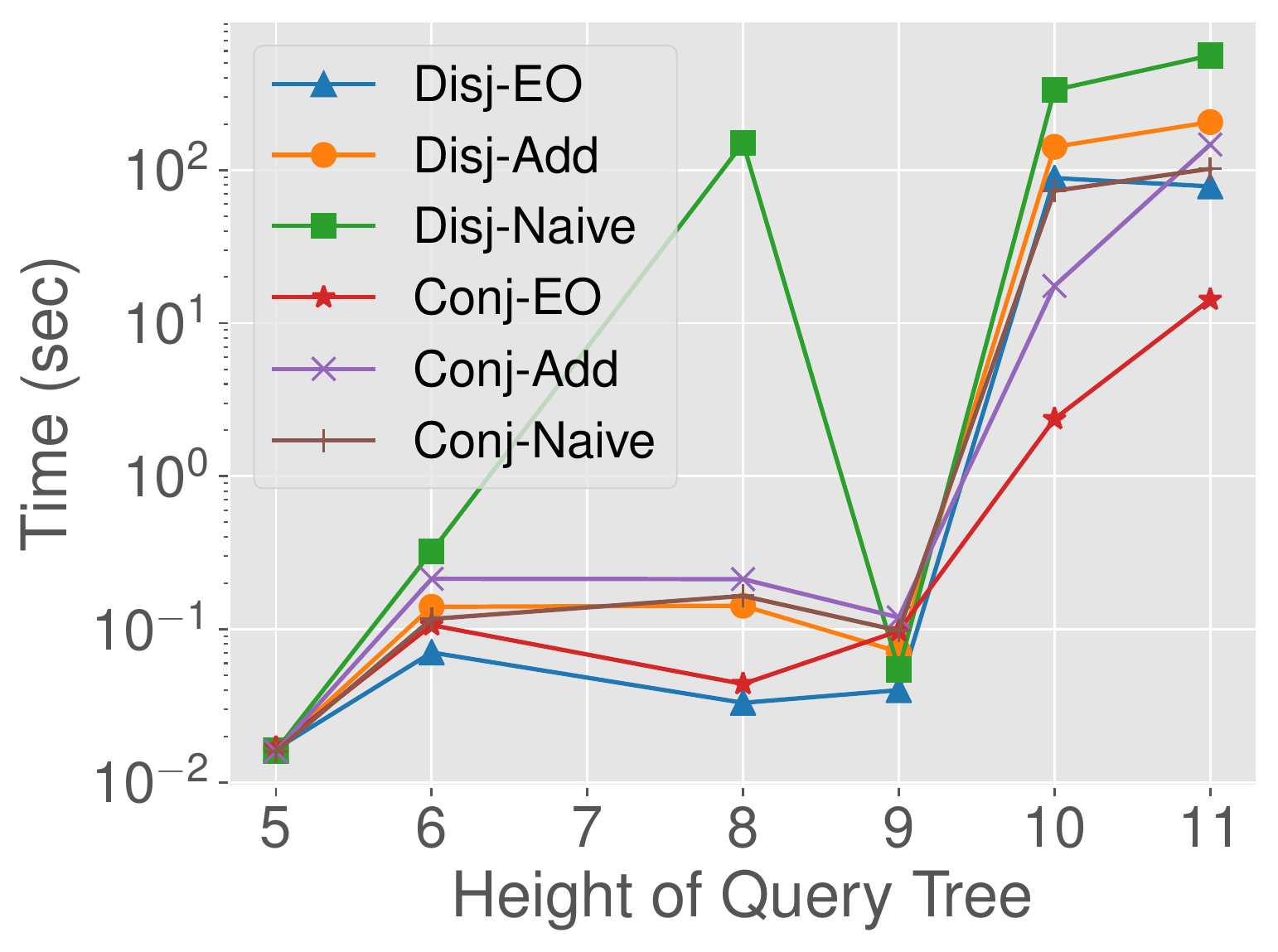}
\end{minipage}
\hfill
  \begin{minipage}{0.24\linewidth}
    \centering
    \includegraphics[width=\linewidth]{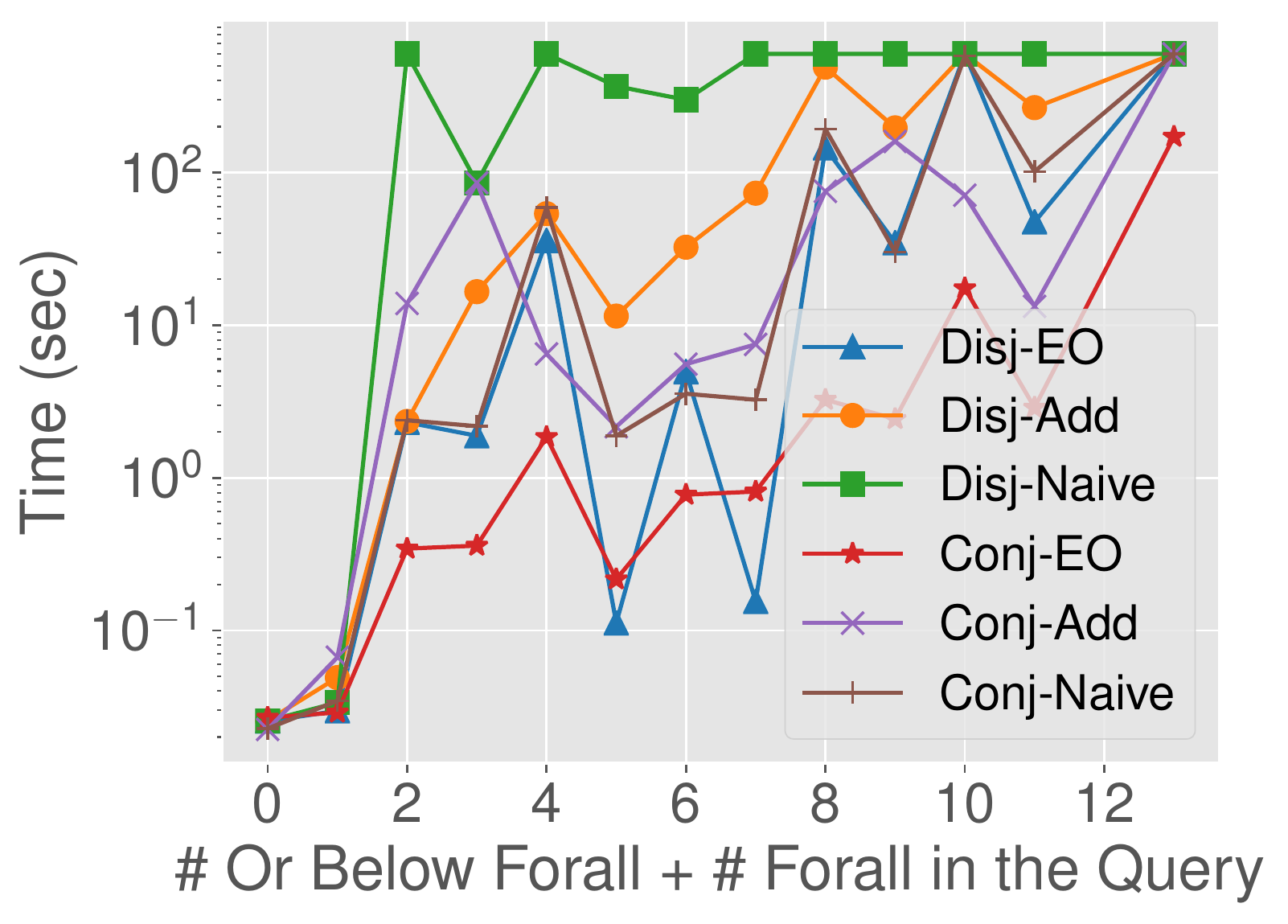}
  \end{minipage}
  \hfill
  \begin{minipage}{0.24\linewidth}
    \centering
    \includegraphics[width=\linewidth]{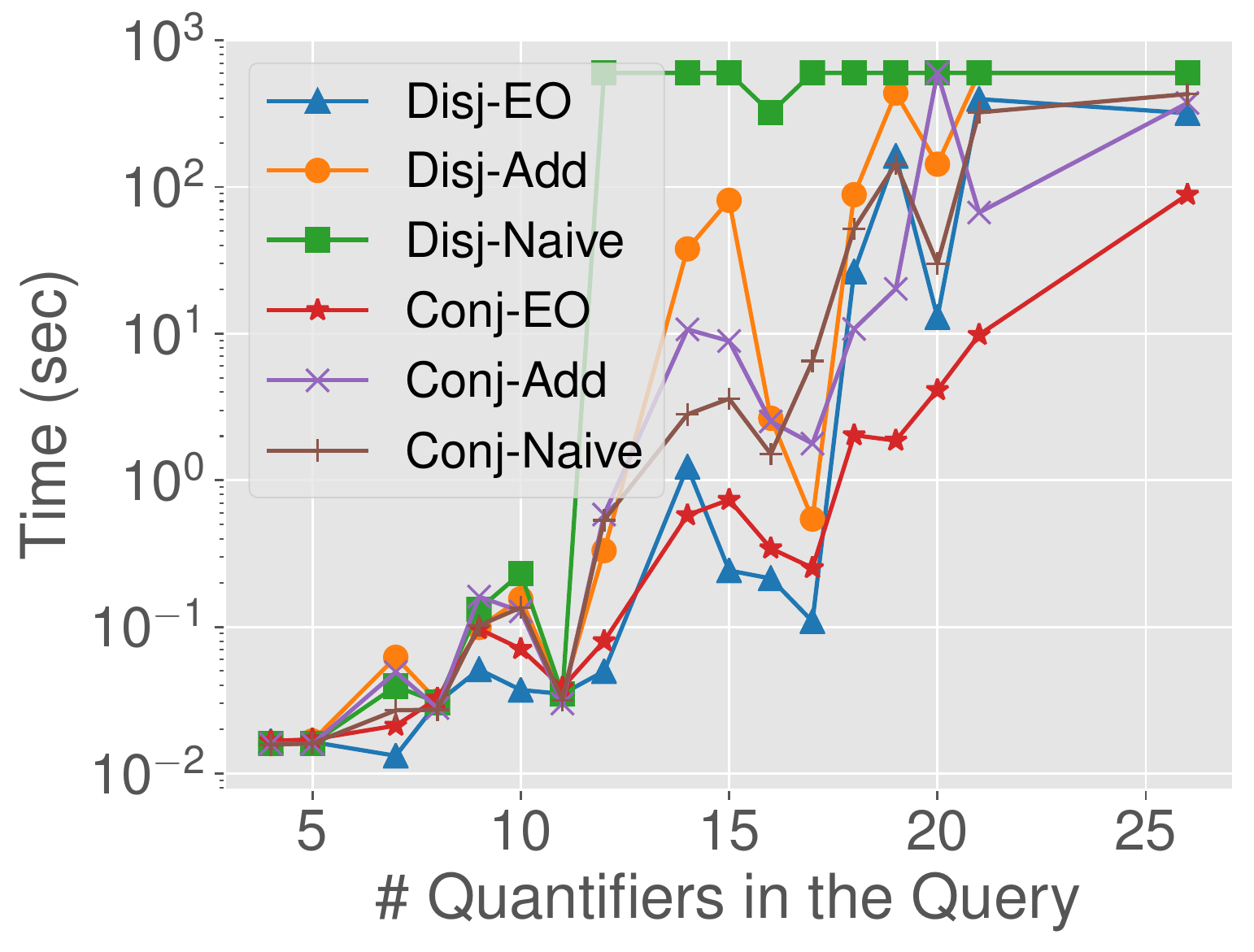}
  \end{minipage}
 \vspace{-5mm}
\caption{\label{fig:query-complexity-vs-time}\small Running time 
vs.\ various measures of query
  complexity.\mdseries\ $limit=10, timeout=600 sec$.}
       \vspace{-3.5mm}
\end{minipage}
\end{figure*}

\subsection{Performance Evaluation}\label{sec:evaluation}

\mypar{Scalability} To evaluate the scalability of our approach, we study how query complexity
affects the running time. We consider four measures of query complexity: 
\reva{
(1) number of nodes in the query tree, (2) the height of the query tree, 
(3) number of universal quantifiers plus number of disjunction that is below a
universal quantifier, and (4) the number of both universal and existential quantifiers. Although most of these parameters are specific to our algorithms that operate on DRC queries, the number of universal quantifiers also has a corresponding complexity notion in the SQL form since each universal quantifier in DRC would lead to at least one negated sub-query in SQL.

\begin{Example}\label{ex:query-complexity}
Recall the query in Figure \ref{fig:diff-query} and its syntax tree in Figure \ref{fig:syntax-tree-minus}. 
The number of nodes in the tree is $27$ (measure (1)),
the height of the tree is $8$ (measure (2)), the query contains $2$ universal quantifier, $3$ disjunctions below it, and $6$ existential quantifiers, so measure (3) is $5$ and (4) is $8$. 
For reference, the queries \qcorrect\ and \qincorrect\ from our running example are shown in SQL in Figure \ref{fig:sql-queries}.
\end{Example}
}

\begin{figure}
    \centering
    \begin{subfigure}{.5\linewidth}
     \begin{scriptsize}
    \begin{center}
    \begin{tabular}{c}
    \begin{lstlisting}[
              frame=none,
              language=SQL,
              showspaces=false,
              basicstyle=\ttfamily,
              numbers=none,
              commentstyle=\color{gray},
              mathescape=true
            ]
    SELECT l.beer, s.bar
    FROM Likes l, Serves s
    WHERE l.drinker LIKE 'Eve %' AND 
    l.beer = s.beer
    AND NOT EXISTS(
        SELECT * FROM Serves
        WHERE beer = s.beer AND price > s.price);
    \end{lstlisting}
    \end{tabular}
    \end{center}
    \end{scriptsize}
    \caption{Correct query \qcorrect}
    \end{subfigure}%
    \begin{subfigure}{.5\linewidth}
    \begin{scriptsize}
    \begin{center}
    \begin{tabular}{c}
    \begin{lstlisting}[
              frame=none,
              language=SQL,
              showspaces=false,
              basicstyle=\ttfamily,
              numbers=none,
              commentstyle=\color{gray}
            ]
    SELECT S1.beer, S1.bar 
    FROM Likes L, Serves S1, Serves S2 
    WHERE L.drinker LIKE 'Eve%' AND 
    L.beer = S1.beer AND L.beer = S2.beer
          AND S1.price > S2.price;
    \end{lstlisting}
    \end{tabular}
    \end{center}
    \end{scriptsize}
    \caption{Incorrect query \qincorrect}
    \end{subfigure}
    \vspace{-4mm}
    \caption{\reva{Queries from our running example in SQL.}}
    \label{fig:sql-queries}
    \vspace{-3mm}
\end{figure}

 \common{We set the limit threshold to be 10 for the Beers dataset and 15 for TPC-H, and stops the algorithm 
if it does not finish in 10 minutes (20 minutes for TPC-H).
The results are shown in Figure~\ref{fig:query-complexity-vs-time} and Figure~\ref{fig:query-complexity-vs-time-and-quality-tpch}, respectively.} We report the average \reva{running time for different queries} with the same value 
of the complexity measure.

As shown in in Figure~\ref{fig:query-complexity-vs-time}, the running time increases 
with query complexity. 
\disjNaive{} has the worst time complexity (more than exponential), which did not finish for most of the complex queries with more than 10 quantifiers, followed by \disjEO{} and \disjAdd{}, whereas \conjNaive{}, \conjEO{}, \conjAdd{} perform better.
The running time of \conjNaive{}, \conjEO{}, \conjAdd{} increases exponentially as expected since their complexity largely depends on the number of conjunctive syntax trees generated from the original query.
Compared to the total number of nodes and the height, 
the number of universal quantifiers and the number of disjunction nodes 
are more crucial to the growth in the running time.
\common{Similar trends are illustrated in Figure~\ref{fig:query-complexity-vs-time-and-quality-tpch}, whereas \conjEO{} still performs better than \disjEO{}, while the running time of \conjAdd{} is very close to \disjAdd{}. We conjecture that this is because the overall complexity of queries in the TPC-H dataset is much higher than the Beers dataset, as shown in Table~\ref{tab:datasets}, leading to more generated conjunctive syntax trees. Our results indicate that our solution scales well for complex schemas and queries (except for very complex and long queries: for only 4 extremely complex cases out of the 28 cases in the TPC-H dataset, our algorithm failed to return any results) .
}

\begin{figure}
\begin{minipage}{0.96\linewidth}
\begin{minipage}{0.46\linewidth}
  \centering
  \includegraphics[width=\linewidth]{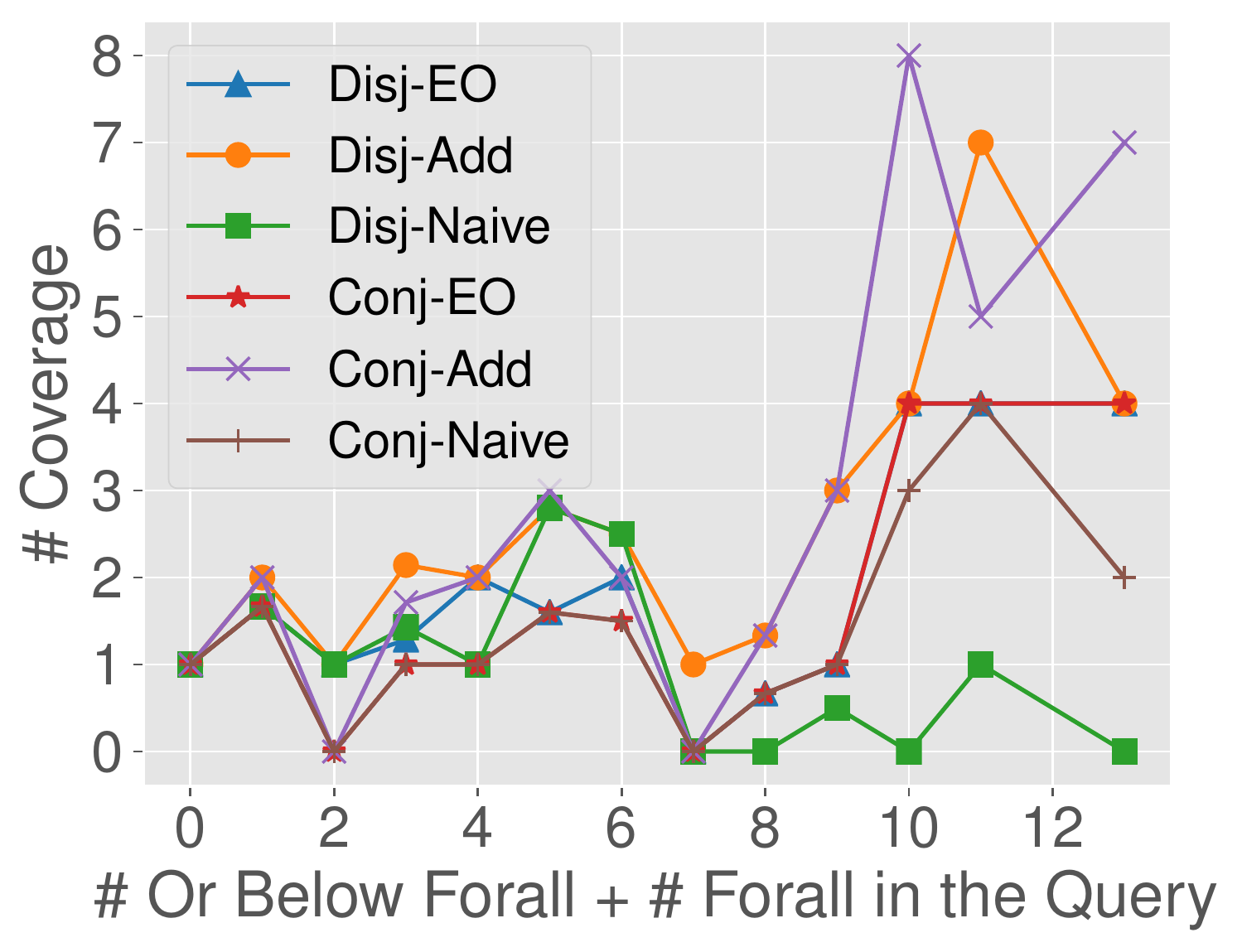}
\end{minipage}
  \hfill
\begin{minipage}{0.46\linewidth}
  \centering
\includegraphics[width=\linewidth]{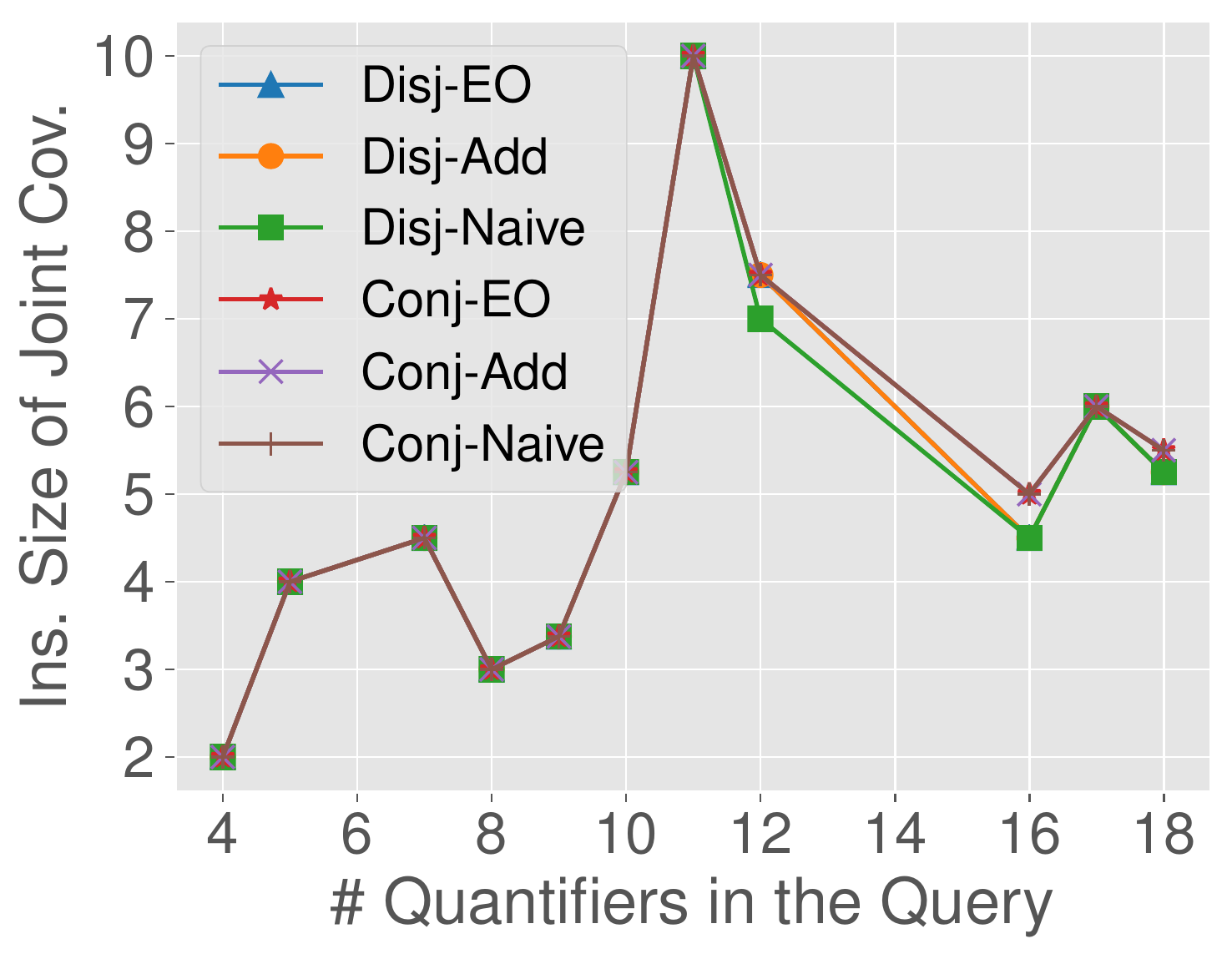}
\end{minipage}
\vspace{-5mm}
\caption{\label{fig:query-result-quality}\small Result quality by query complexity.
\mdseries\ $limit=10, timeout=600 sec$.}
       \vspace{-4.5mm}
\end{minipage}
\end{figure}

\begin{figure}
\begin{minipage}{0.96\linewidth}
\begin{minipage}{0.46\linewidth}
  \centering
  \includegraphics[width=\linewidth]{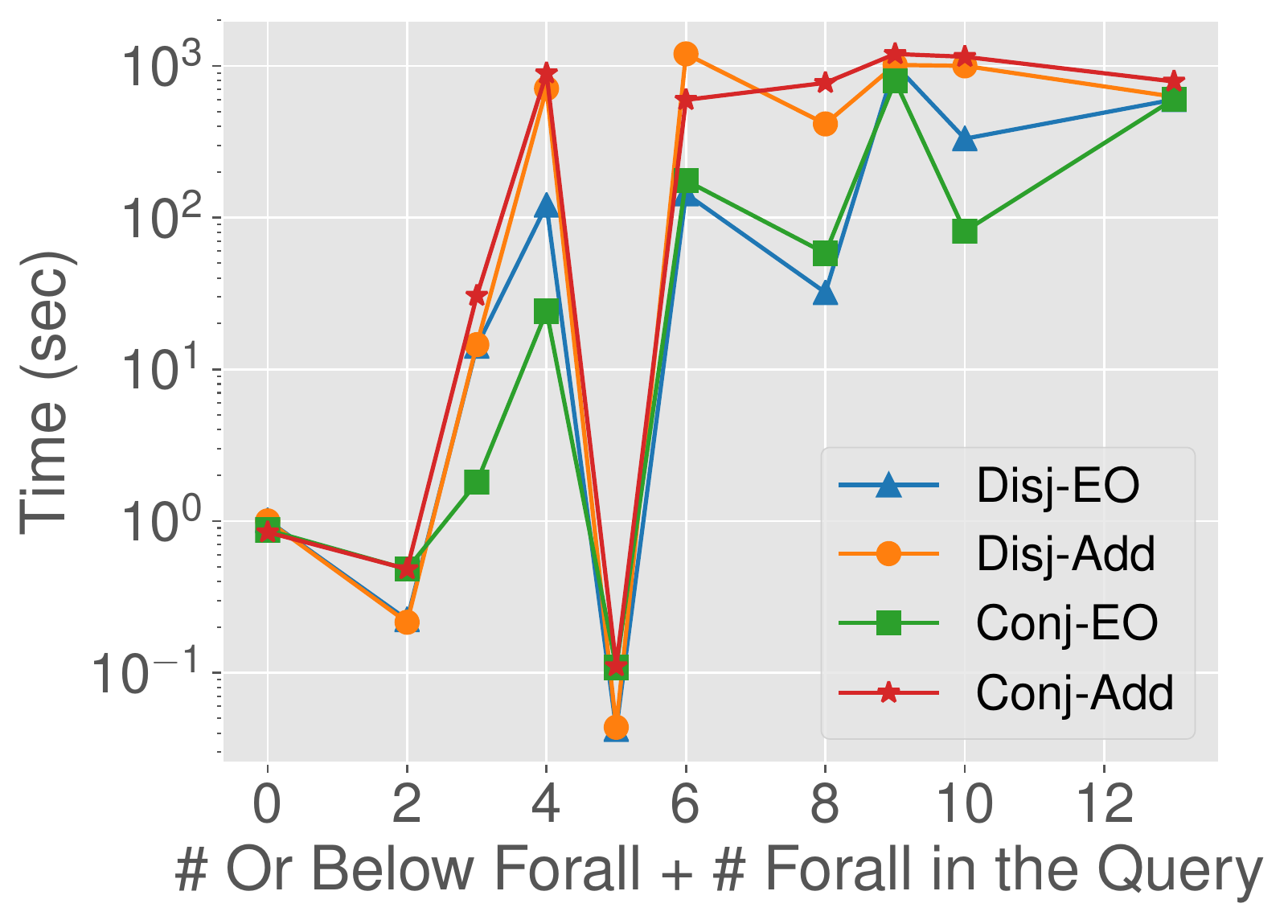}
\end{minipage}
  \hfill
\begin{minipage}{0.46\linewidth}
  \centering
\includegraphics[width=\linewidth]{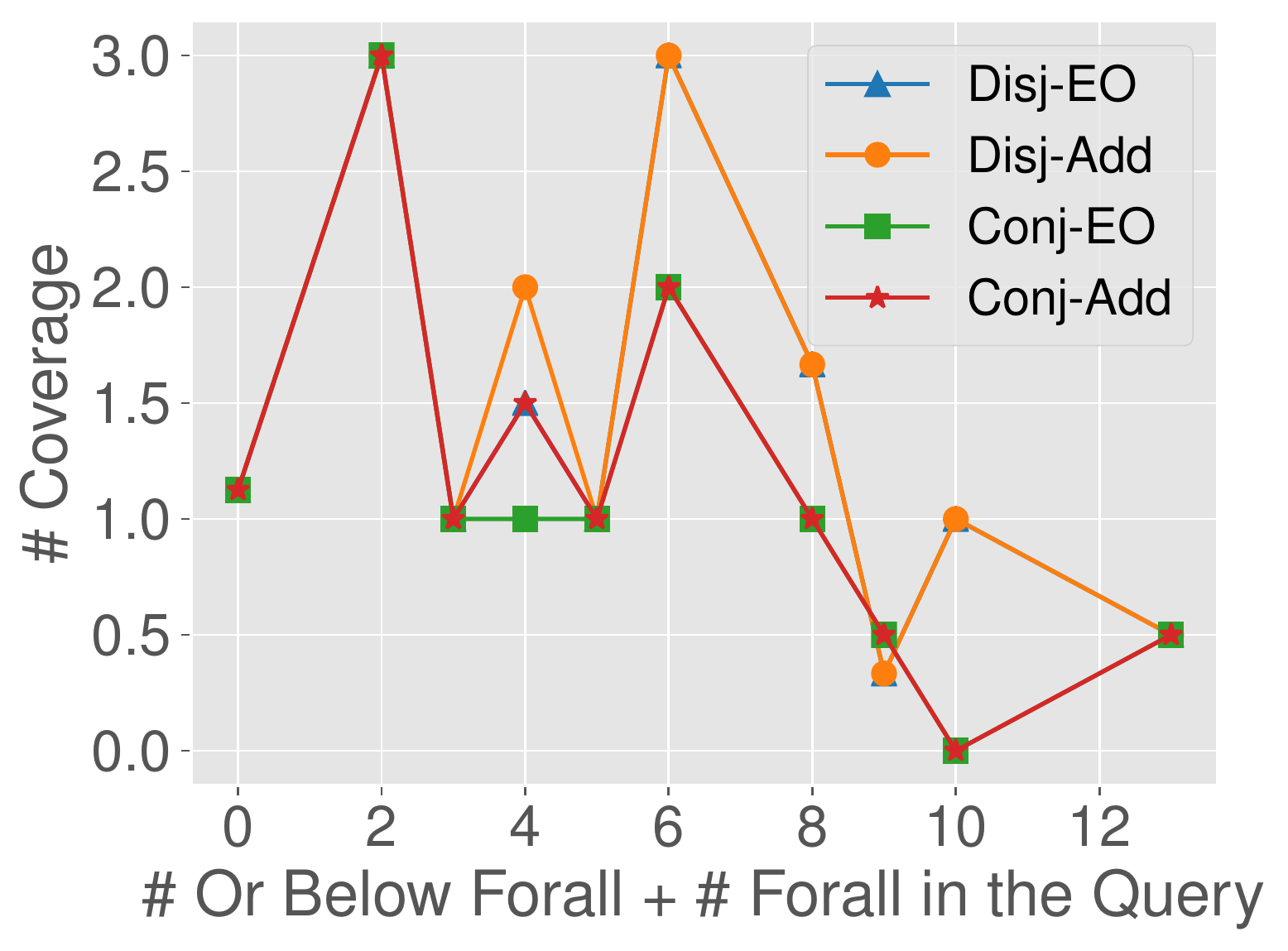}
\end{minipage}
\vspace{-5mm}
\caption{\label{fig:query-complexity-vs-time-and-quality-tpch}\small \common{Running time and result quality by query complexity on TPC-H dataset
\mdseries\ $limit=15, timeout=1200 sec$.}}\label{fig:tpch-runtime-quality}
       \vspace{-4.5mm}
\end{minipage}
\end{figure}


\mypar{Result quality} Our optimized approaches (\conjEO{}, \disjEO{}, \conjAdd{} and \disjAdd{}) run much faster than \disjNaive{} by 
compromise on the completeness of the minimal c-solution to different extents. To evaluate the result quality of these approaches in
terms of both completeness and minimality, we show in Figure~\ref{fig:query-result-quality} the number of distinct coverage from the returned c-solutions and the average size of the c-solutions. Notice that the number of returned minimal c-solution of each variant can be different (either they are unable to find some results, e.g. \disjEO{} returns a subset of \disjAdd{}; or some variants finish before the timeout but the others do not), to guarantee a fair comparison, for each query we only consider c-solutions with a coverage set returned by all of the variants.
For example, for a query $Q$ if \conjNaive{} returns two c-solutions with coverage $C_1$ and $C_2$, and \disjAdd{} returns three c-solutions with coverage $C_1, C_2, C_3$, we will only report the average c-solution size of the two with coverage $C_1$ and $C_2$ for $Q$.

Figure~\ref{fig:query-result-quality} shows that \disjAdd{} returns more distinct coverage sets in most cases, while \conjAdd{}, \conjNaive{},\conjEO{}, and \disjEO{} might fail to return any satisfying instances. There are a few exceptions for some very complex queries where \disjAdd{} did not finish before timeout and thus \conjAdd{} returns more distinct coverage sets. Although \disjNaive{} did not finish running in most cases, the c-solutions it returns can be smaller than other variants when there are more than 10 quantifiers in the query. There are few cases where the \disjAdd{} and \disjEO{} return smaller c-solutions compared to the variants using conjunctive trees.

\begin{figure}
\begin{minipage}{0.46\linewidth}
  \centering
  \includegraphics[width=\linewidth]{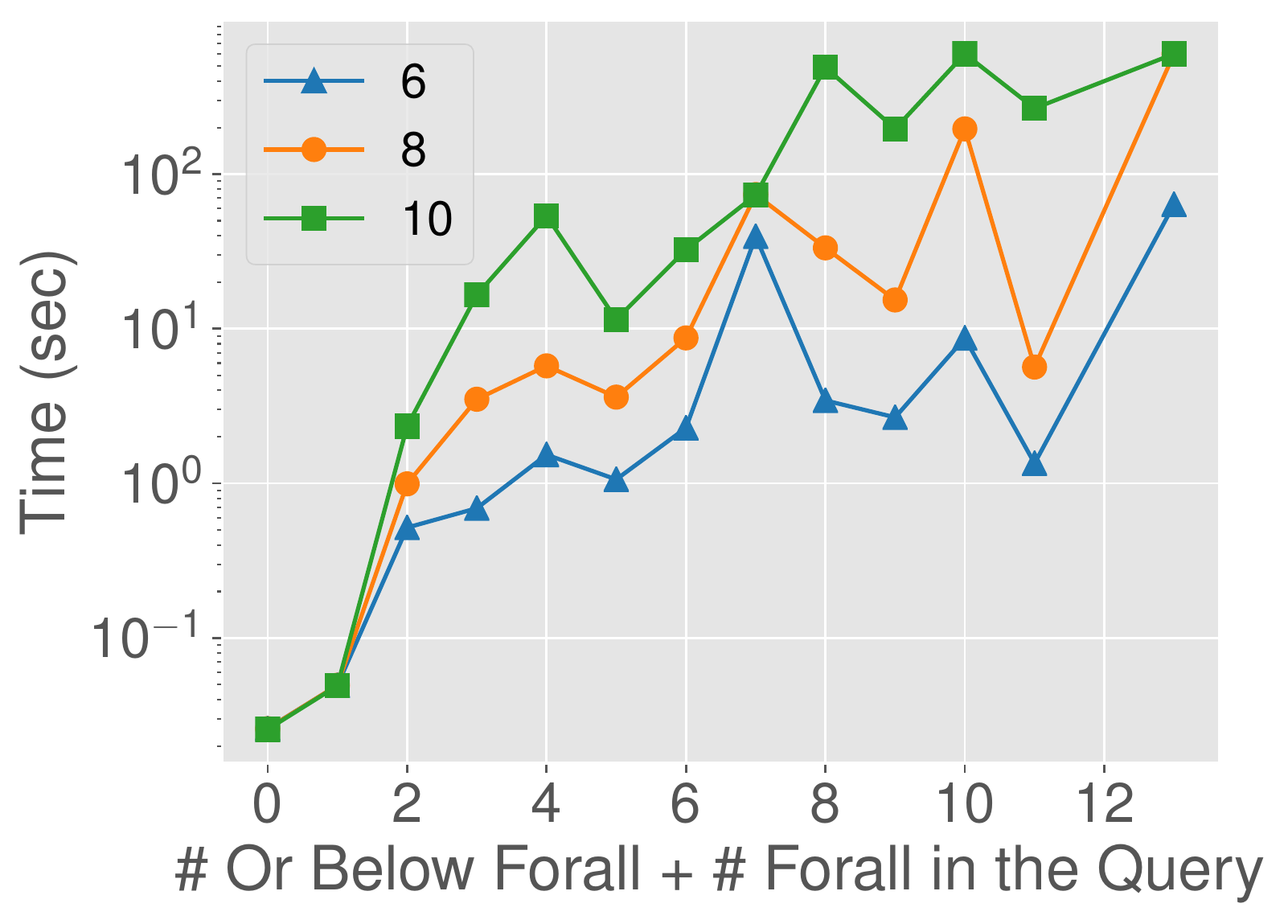}
\end{minipage}
  \hfill
\begin{minipage}{0.46\linewidth}
  \centering
  \includegraphics[width=\linewidth]{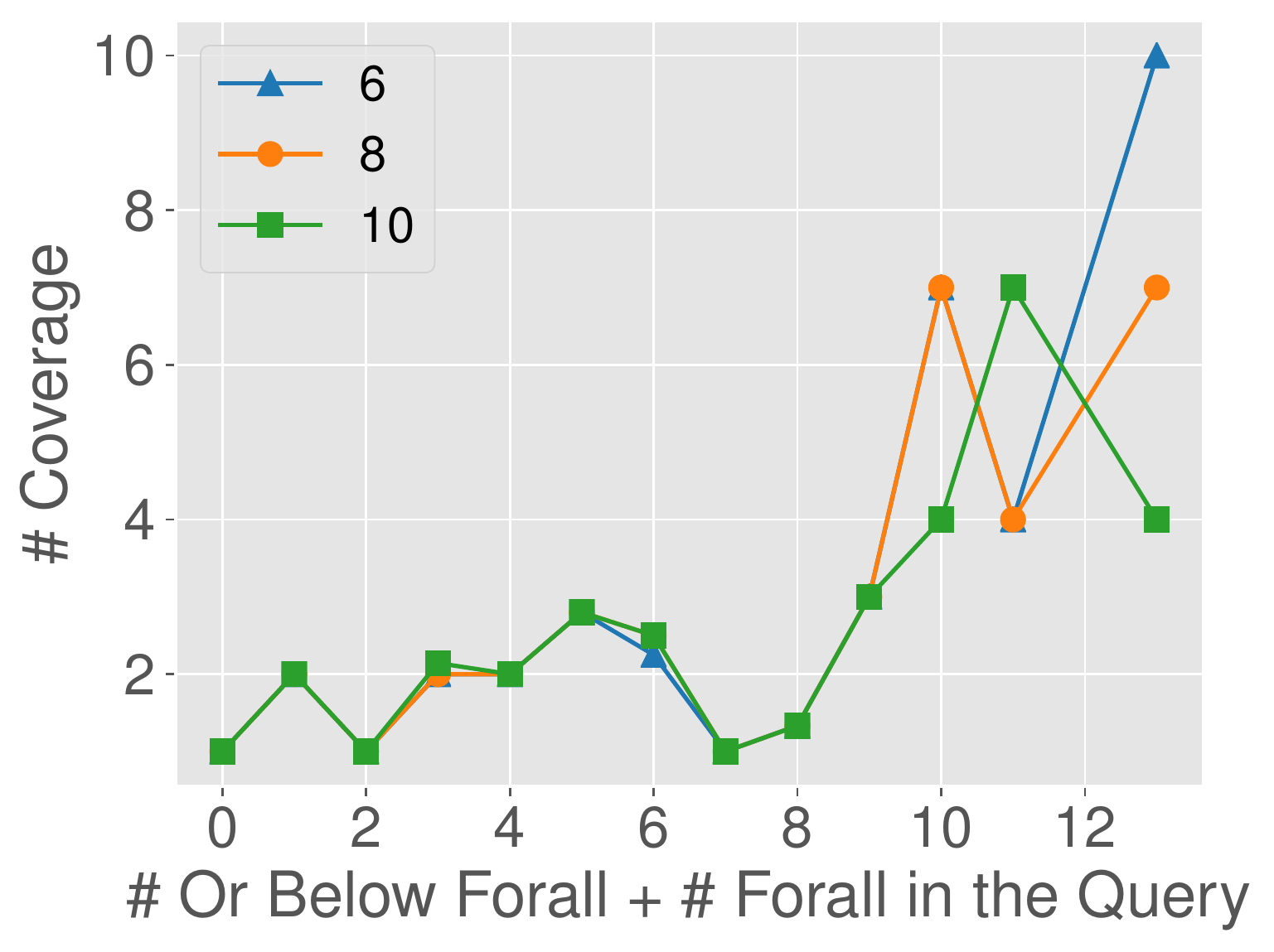}
\end{minipage}
\vspace{-5mm}
\caption{\label{fig:disj-limit-parameter}\small Parameter sensitivity varying limit.\mdseries\ $\disjAdd{}$.}
\vspace{-4.5mm}
\end{figure}

\begin{figure}
  \begin{minipage}{0.46\linewidth}
    \centering
    \includegraphics[width=\linewidth]{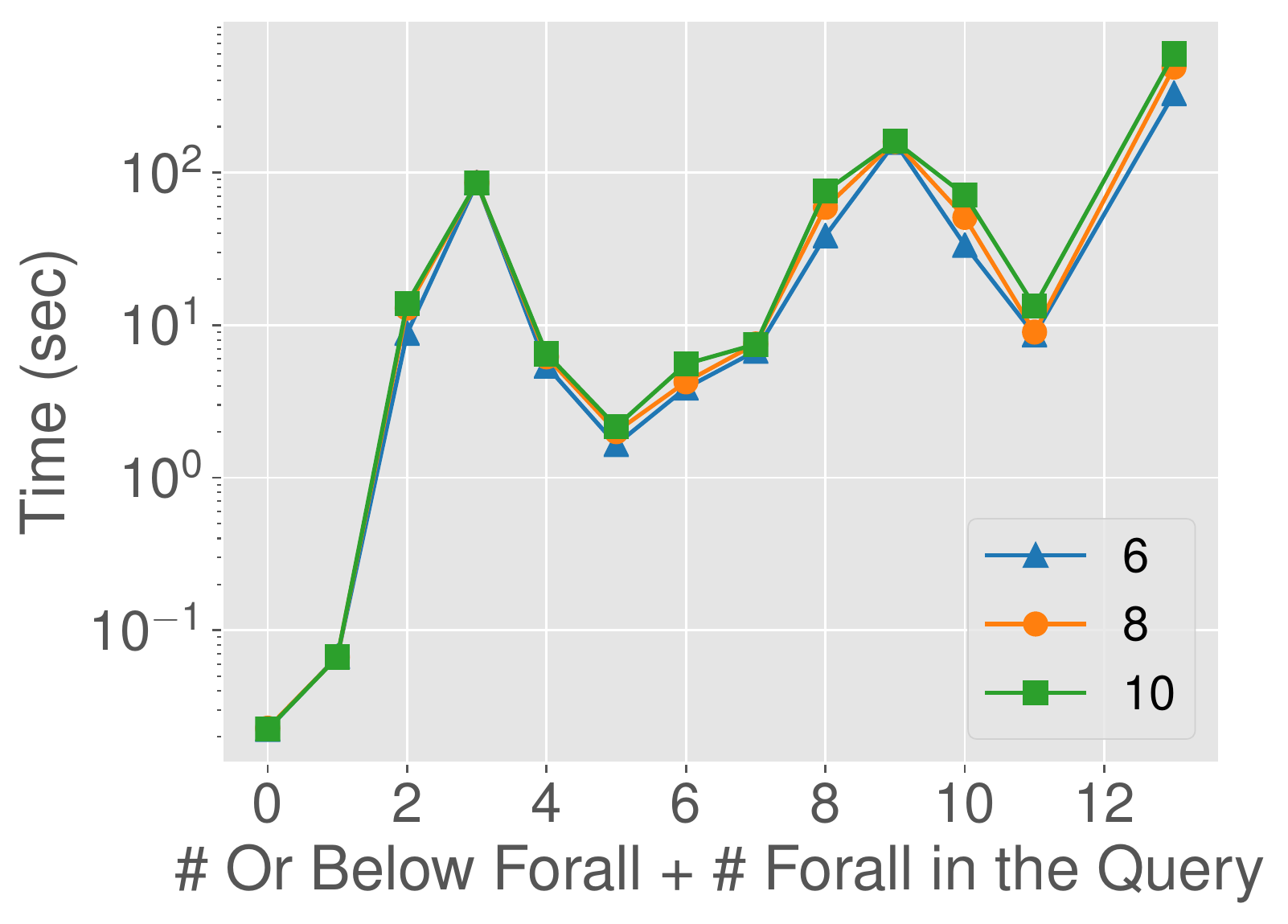}
  \end{minipage}
  \hfill
  \begin{minipage}{0.46\linewidth}
    \centering
    \includegraphics[width=\linewidth]{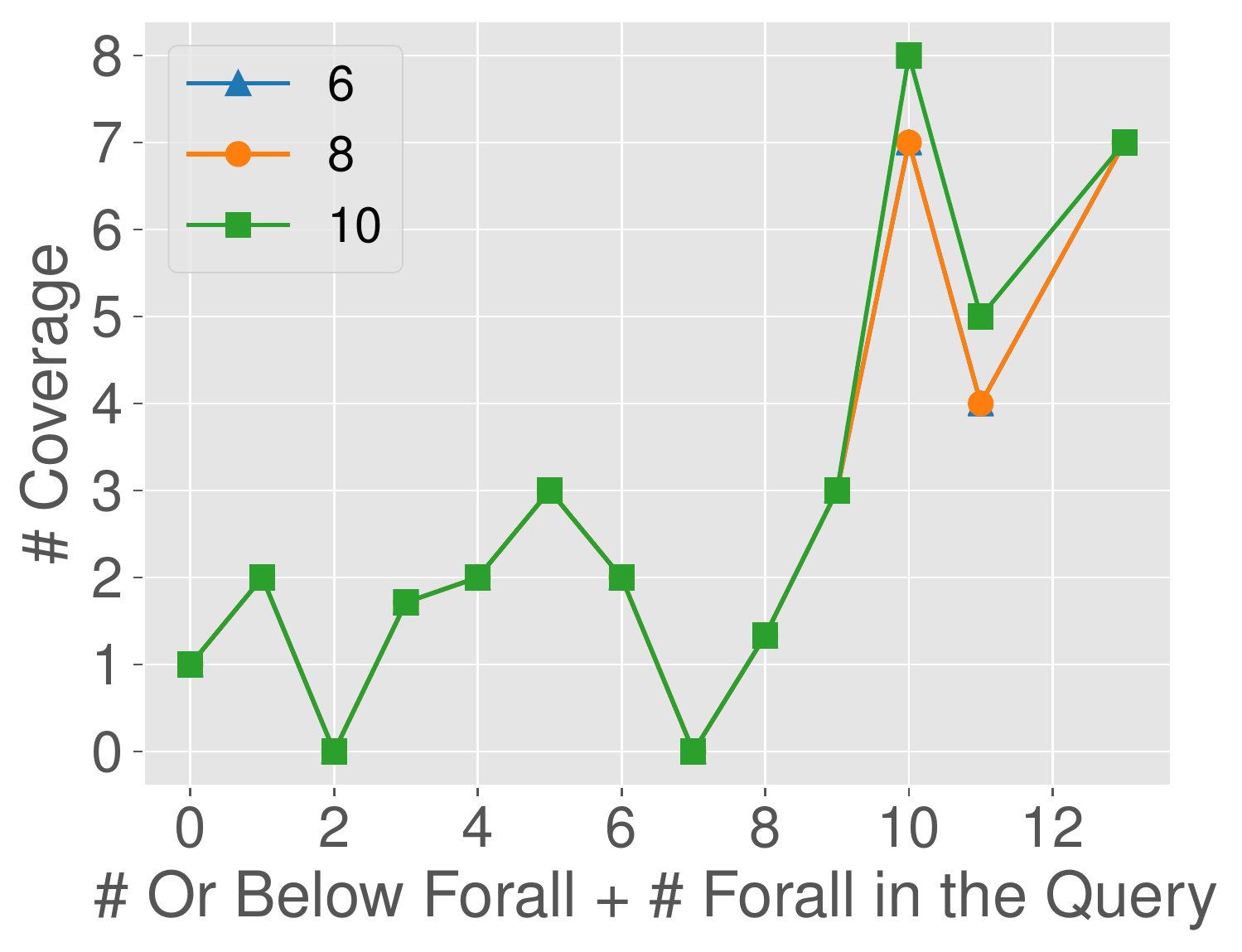}
  \end{minipage}
\vspace{-5mm}
\caption{\label{fig:conj-limit-parameter}\small Parameter sensitivity varying limit.\mdseries\ $\conjAdd$.}
\vspace{-4mm}
\end{figure}

\mypar{Parameter sensitivity} To ensure that the algorithm terminates, we used a \textit{limit} parameter to restrict the size of the c-instances.
Figure~\ref{fig:disj-limit-parameter} and Figure~\ref{fig:conj-limit-parameter} show how this limit affects the running time and completeness for \disjAdd{} and \conjAdd{}. Although the running time grows exponentially with the query complexity, \disjAdd{} runs one order of magnitude faster when limit=6 than limit=10, losing completeness only when the query tree is very complex. For \conjAdd{}, the difference in running time when varying the limit is negligible in most cases, and the number of distinct returned coverage sets only changes for the most complex queries.

\mypar{\common{Interactivity}} \common{To improve interactivity, our algorithms can output the instances one at a time as soon as they are generated, so users can start exploring immediately and have a more interactive experience. The time to produce the first instance for our algorithms on the Beers dataset is only $4.78$ seconds on average (DisjAdd) or $0.77$ seconds (ConjAdd), and the average delay between two consecutive output instances with different coverage is $18.34$ seconds (DisjAdd) or $5.22$ seconds (ConjAdd). While on the TPC-H dataset, the time to produce the first instance is longer but still tolerable: $101.25$ seconds on average (DisjAdd) or $88.02$ seconds (ConjAdd), and the average delay between two consecutive output instances with different coverage is $19.12$ seconds (DisjAdd) or $54.16$ seconds (ConjAdd).
Note that doing so may risk returning non-minimal instances, as minimality is verified in postprocessing (Section \ref{sec:complete-algo}). 
Another option is to start with the optimized version (Section \ref{sec:conjunct}) and if further insights are needed, run the exhaustive search (Section \ref{sec:complete-algo}).
We also note that slightly longer wait times might be acceptable in some scenarios, e.g., providing offline feedback to student solutions, or to help students/instructors when manual debugging would take significantly more effort for complex queries or subtly wrong solutions.
}


\vspace{-1mm}
\subsection{Case Study}\label{sec:case-study}

By providing the ``basis'' to a query $Q$, our work yields
a set of abstract instances that can help users understand and 
debug their query in practice. 
To evaluate the usefulness of the set of abstract instances (the minimal c-solution returned by the algorithms, providing a proxy for the universal solution),
we report one case study on the same real-world dataset as the performance evaluation from an undergrad database course.
We pick two most complex standard solution queries from an assignment each with one wrong query from student submissions.
Table~\ref{tab:case-study-queries} shows the solution queries, wrong queries, and the universal solution for the difference query of the standard and wrong queries. 

The universal solution captures different errors in the wrong query. To compare, we use the ground instances that serve as ``counterexamples'' 
for the wrong queries by a previous system~\cite{MiaoRY19} based on a randomly generated testing database instance.

For $Q_1$ (the same as our running example), the first and the second c-instances pinpoint that if the drinker's first name is not `Eve' but has `Eve' as its prefix. While the first c-instance does not contain the first name condition, it shows that if all three bars serve the same beer at different prices, the query would go wrong. Note that if we add to the last instance the condition $\neg (d_1 \sql{LIKE} \text{`Eve\textvisiblespace\%'})$, it is still a satisfying c-instance, but it is not minimal because its coverage is the same as the second c-instance. 
In comparison, while the ground instance by \cite{MiaoRY19} (as in Figure~\ref{fig:running}) is in the represented world of the first c-instance, it does not highlight that the reason behind the wrong query result is that the prices are ordered in a particular way, but the actual values are unimportant.
 

For $Q_2$, the c-instances in the universal solution indicate that the query would go wrong if there is a drinker frequents to a bar that does not serve any beer, no matter the drinker likes a beer or not (the 1st and 3rd instances). This may pinpoint the error that the Likes table does not interact with the Serves table. Furthermore, the 2nd, 5th, and 6th c-instances imply that if there is a beer served at a bar, to make the query return a wrong result, the drinker should not frequent this bar, which could be interpreted as the correct solution uses the frequents table together with negation in a different way. Actually, the wrong query joins Frequents with Serves, while the correct solution joins Likes with Serves. Therefore, the universal solution provides different perspectives in understanding the query and formulates a hint on how to modify the wrong query.
The ground instance by \cite{MiaoRY19} only consists of four tuples: \Drinker(``Bryan'', ``39934 Main St.''),
\Beer(``Amstel'', ``A. Brewer''), \Bar(``The Edge'', ``802 Morris St.''), and \Frequents(``Bryan'', ``The Edge'', 3), which is in the represented world of the third c-instance from our universal solution in Table~\ref{tab:case-study-queries}. Such a simple counterexample might be less helpful for users to understand why the query goes wrong, as one would benefit from the explicit conditions with negation.

\begingroup\addtolength{\jot}{-1ex}
\begin{table*}[t]
  {    
  \centering
    \scriptsize
    \begin{tabular}{|p{13em}|p{19em}|p{46em}|}
      \hline
      \textbf{Query description}& \textbf{Queries (DRC)} & \textbf{c-instances} \\\hline
      \multirow{3}{13em}{$\query_1$: for each beer liked by any drinker whose first name is Eve, find the bars that serve this beer at the highest price} &
\multirow{3}{*}{
\parbox{19em}{
Correct query: $Q_A$ in Figure~\ref{fig:q_correct}.
\\
Wrong query: $Q_B$ in Figure~\ref{fig:q_incorrect}.}
}
    & 
    \begin{minipage}{0.6\textwidth}
    $I_0$ in Figure~\ref{fig:c-instance}

\end{minipage}
 \\\cline{3-3}
    & &     \begin{minipage}{0.6\textwidth}
    $I_1$ in Figure~\ref{fig:c-instance-1}
\end{minipage}
     \\\cline{3-3}
& & \begin{minipage}{0.6\textwidth}
$ 
\begin{aligned}
&\Drinker(d_1, *), \Drinker(d_2, *), \Beer(b_1, *), \Bar(x_1, *), \Bar(x_2, *), \Likes(d_1, b_1), \Serves(x_1, b_1, p_1),\\& \Serves(x_2, b_1, p_2), d_1 \sql{LIKE} \text{`Eve\%'} \land \neg (d_1 \sql{LIKE} \text{`Eve\textvisiblespace\%'}) \land \neg(\Likes(d_2, b_1))  \land p_1 < p_2
\end{aligned}$
\end{minipage}
     \\\hline

           \multirow{8}{10em}{$\query_2$: Find names of all drinkers who frequent only bars that serve some beer they like} &
           
\multirow{8}{12em}{
\begin{minipage}{0.6\textwidth}
$\begin{aligned}
&\text{Correct Query: }\\
&Q_{2A}=\{(d_1)  \ \mid \ 
         \exists a_1 \big( \Drinker(d_1, a_1) \land\\
         &\forall x_1 \forall t_1 \big( \neg \Frequents(d_1, x_1, t_1) 
         \lor \exists b_1, p_1 \\& (\Serves(x_1, b_1, p_1) \land   \Likes(d_1, b_1) ) \big) 
         \Big)\} \notag \\
         \\
&\text{Wrong Query:} \\
     &Q_{2B}=\{(d_1)  \ \mid \ 
          \exists a_1 \Big( \Drinker(d_1, a_1) \land\\
         &\forall b_1 \big ( \forall t_1, t_1, p_1( \neg \Frequents(d_1, x_1, t_1) \lor \\
         &\neg \Serves(x_1, b_1, p_1) )
         \lor  \Likes(d_1, b_1) \big) \Big)\} \notag \\
         &\text{Showing universal solution for } Q_{2B} - Q_{2A}
\end{aligned} $
     \end{minipage}
\     }
    & 
  \begin{minipage}{0.6\textwidth}
$ \begin{aligned} &\Drinker(d_1, *), \Beer(b_1, *), \Bar(x_1, *), \Likes(d_1, b_1), \Frequents(d_1, x_1, t_1)\\
\end{aligned}$
\end{minipage}
    \\ \cline{3-3}
 &
 &  
\begin{minipage}{0.6\textwidth}
$ 
\begin{aligned} &\Drinker(d_1, *), \Beer(b_1, *), \Bar(x_1, *), \Bar(x_2, *), \Likes(d_1, b_1),
\Serves(x_1, b_1, p_1), \\
&\Frequents(d_1, x_2, t_1),
\neg(\Frequents(d_1, x_1, t_1))
\end{aligned}$
\end{minipage}
 \\\cline{3-3}
    &  &   
\begin{minipage}{0.6\textwidth}
$ 
\begin{aligned} &\Drinker(d_1, *), \Beer(b_1, *), \Bar(x_1, *),  \Frequents(d_1, x_1, *),
\end{aligned}$
\end{minipage}
    \\\cline{3-3}
    &  &
\begin{minipage}{0.6\textwidth}
$ 
\begin{aligned} &\Drinker(d_1, *), \Beer(b_1, *), \Bar(x_1, *), \Bar(x_2, *), \Likes(d_1, b_1),\\ &\Serves(x_1, b_1, p_1), \Frequents(d_1, x_1, t_1), \Frequents(d_1, x_2, t_1)
\end{aligned}$
\end{minipage}
     \\\cline{3-3}
& & 
 \begin{minipage}{0.6\textwidth}
$ 
\begin{aligned} &\Drinker(d_1, *), \Beer(b_1, *), \Bar(x_1, *), \Bar(x_2, *), \Serves(x_1, b_1, p_1), \Frequents(d_1, x_2, t_1),
\\& \neg(\Likes(d_1, b_1)) \land \neg(\Frequents(d_1, x_1, t_1))
\end{aligned}$
\end{minipage}
 \\\cline{3-3}
& & \begin{minipage}{0.6\textwidth}
$ 
\begin{aligned} &\Drinker(d_1, *),\Drinker(d_2, *), \Beer(b_1, *), \Bar(x_1, *), \Bar(x_2, *), \Likes(d_2, b_1),\\ &\Serves(x_1, b_1, p_1), \Frequents(d_2, x_2, t_1), \neg(\Frequents(d_2, x_1, t_1)) \land \neg(\Likes(d_1, b_1))
\end{aligned}$
\end{minipage}

 \\\cline{3-3}
& & \begin{minipage}{0.6\textwidth}
$ 
\begin{aligned} &\Drinker(d_1, *), \Beer(b_1, *),\Beer(b_2, *), \Bar(x_1, *), \Bar(x_2, *), \Likes(d_1, b_1),\\ &\Serves(x_1, b_1, p_1), \Frequents(d_1, x_1, t_1), \Frequents(d_1, x_2, t_1),
\neg(\Likes(d_1, b_2))
\end{aligned}$
\end{minipage}


\\\hline

    \end{tabular}
    \caption{Queries used in the case study; universal solutions generated using \disjAdd{}, $limit=10$. Note that the symbols in queries (representing the query variables) and the symbols in c-instances (representing labeled nulls) are not the same.}
    \label{tab:case-study-queries}
    \vspace{-8mm}
    }
\end{table*}
\endgroup

\vspace{-4mm}

\common{
\subsection{User Study}\label{sec:user-study}
We conducted a user study for the Beers dataset to evaluate: 
(R1) how effective our approach is for explaining and understanding bugs in queries,
and (R2) whether completeness as a quality metric is helpful.
For (R1), we specifically compare our approach (c-instances) with {\bf concrete instances \cite{MiaoRY19}}  having constant values. Note that our approach takes only the queries and the schema as input, whereas \cite{MiaoRY19} also takes a database instance as input and outputs a sub-instance as a concrete counterexample%
\footnote{\common{In a pilot study we also compared these approaches against a baseline of not providing any instances (c-instance or concrete).
We found that study questions involving this baseline significantly increased the length and difficulty of the survey demanding higher participant efforts to the point of discouraging participation, and that from the preliminary results we collected, participants did much better with the help of instances.
Hence, in the final user study we excluded the baseline, but asked the question: \emph{``When you learn SQL queries in the future, would you like to see the example instances shown in this survey to help you understand incorrect queries?''}
All 64 participants answered yes.}}.

\mypar{Participants} We recruited 64 participants, including 22 graduate students from CS departments and 42 undergraduate students from an undergraduate database course. 
Participation was voluntary and anonymous, though the undergraduates were offered small souvenirs as a reward for their participation (we did not get enough responses from the undergraduates in our initial pilot surveys).
The undergraduate students were already familiar with the schema of the Beers dataset from an earlier homework;
while for the graduate students, we explained the schema details and also asked about their familiarity with SQL.
We note that the undergraduate students have also been exposed to the tool using concrete instances developed by~\cite{MiaoRY19} (but only for relational algebra queries); because of this familiarity, concrete instances might hold a slight advantage over c-instances for these students. Half of the graduate students (11 out of 22) graduate students declared high familiarity with SQL queries and the rest reported moderate or low familiarity; our observations for both groups were similar in this study, therefore we report the overall statistics for graduate students. 

\mypar{Tasks}
We asked all participants to spot errors in two SQL queries (each has two major errors, see Table~\ref{tab:user-study-queries}) querying the Beers database, with the help of either our c-instances or concrete instances from \cite{MiaoRY19}. 
We provided each participant with one query followed by c-instances and other query followed by concrete instances as counterexamples, randomly dividing them into two groups:  one group saw $(Q_1$, c-instances$)$  + $(Q_2$, concrete instances$)$ 
and the other saw $(Q_1$, concrete instances$)$  + $(Q_2$, c-instances$)$. The order of showing these two questions for both groups was chosen at random to avoid any familiarity bias against either c-instances or concrete instances. 
Instead of showing completely abstract c-instances, we added an example concrete value to each variable in the c-instance, showing one way that it can be grounded.
This is a trivial extension done to help alleviate novices' potential discomfort with seeing symbols and conditions alone.
This approach somewhat blurs the line between c-instances and concrete instances, but faithfully represents how in practice c-instances would be deployed in an educational setting.
Then, for each query we have the {\em treatment group} (with c-instances as explanation) and the {\em control group} (with concrete-value-only instances as explanation).
To study (R2), following the task involving the first c-instance above, we presented a second c-instance for the same query but with a different coverage (which would illustrate a different error), and asked the participant what errors they found upon seeing both c-instances, and whether they felt the second c-instance provided additional help.
(Note that \cite{MiaoRY19} and other related work, there is no option for generating additional concrete instances that illustrate different errors in the same query.)
At the end of the study, back to (R1), we also asked the participants about their preferences between c-instances and concrete instances for spotting errors in queries.

\mypar{Results and analysis}
Objectively, we evaluate the user performance by the number of errors they spotted.
Figures~\ref{fig:user-study-errors} show the percentage of users who failed to spot any error, spotted one error, and spotted both errors in each group.
For example, consider the last three bars in the left sub-figure in Figure~\ref{fig:user-study-errors}, which show the overall statistics (combining both queries) for all undergraduate participants.
Showing the concrete instance alone is already quite helpful: only 31\% of the users failed to find an error (``total-conc'').
Going from ``total-conc'' to ``total-CI1,'' we see a clear performance improvement among users who were shown c-instances: percentage of the users failing to find an error goes down to 19\%.
Going further from ``total-CI1'' to ``total-CI2,'' we see that as soon as users are show a second c-instance, practically all of them were able to spot at least one error, and the majority (64.3\%) of them indeed spot both errors in the query; in contrast, no users were able to spot both errors with only an concrete instance (``total-conc'').
Similar conclusions can be drawn from per-query statistics (shown by the first two batches of three bars in the left part of Figure~\ref{fig:user-study-errors}) as well as from statistics for graduate students (shown in the right part of Figure~\ref{fig:user-study-errors}).
Overall, these results convincingly show that for (R1), c-instances hold a clear advantage over concrete instances in objectively improving participants' performance in spotting errors in queries; and for (R2), showing multiple c-instances with different coverage dramatically improves participants' ability in spotting remaining errors in queries.

A number of other observations from these results are worth noting but largely confirms intuition.
First, $Q_2$ was easier to debug than $Q_1$.
Second, graduate students overall perform better than undergraduates.
The coupling of these factors explains why we did not see any difference from ``Q2-conc'' to ``Q2-CI1'' in the right part of Figure~\ref{fig:user-study-errors}; apparently most graduate students got enough help from the concrete instance in order to spot at least one error in the simpler $Q_2$.
Nonetheless, they still needed the help with an additional c-instance to uncover the second error.

Figures~\ref{fig:user-study-preference} and~\ref{fig:user-study-second-cins} summarize the subjective responses from participants regarding their preference for c-instances vs.\ concrete instances, and their opinion on the usefulness of additional c-instances.
A clear majority of the participants found the additional c-instance useful per Figure~\ref{fig:user-study-second-cins}.
However, from Figure~\ref{fig:user-study-preference}, it is apparent that many participants prefer viewing concrete instances, despite the fact that they perform objectively better with the help of c-instances.
This preference is stronger among undergraduates---only a third preferred c-instances, compared with more than a half for concrete instances.
A relatively lower fraction of graduate students---but still a half of them---preferred concrete instances.
It would be interesting to conduct additional study to pinpoint their reluctance to embrace c-instances despite their objective advantages, but there are several possible explanations.
First, the abstraction provided by variables and conditions the c-instances may be seen as more intimidating, especially for undergraduates.
This conjecture is corroborated by some of the free-form feedback comments we received.
Second, as mentioned earlier in this section, the undergraduates already had some familiarity working with concrete instances before this user study.
Overall, the fact that still about a third of the participants preferred c-instances shows that there is a sizable and compelling demand for this approach.
We also believe we can mitigate some of the reluctance in this user base with improved interfaces and familiarity. 
}

\begin{table}[t]
  {    
  \centering
    \scriptsize
    \begin{tabular}{|p{18em}|p{18em}|}
      \hline
      \textbf{Query description}& \textbf{Wrong Queries}  \\\hline
      $\query_1$: for each beer liked by any drinker whose first name is ``Eve'', find the bars that serve this beer at the highest price &
$\query_B$ in Figure~\ref{fig:sql-queries} (our running example)\\\hline
$\query_2$: Among the drinkers who frequent ``The Edge'', find the names of those who do not like ``Erdinger''. & 
{\bf SELECT DISTINCT} S.beer~~~~
{\bf FROM} Serves S, Likes L~~~~
{\bf WHERE} S.bar = 'Edge' {\bf AND} S.beer = L.beer
  {\bf AND} L.drinker <> 'Richard';
\cut{
\begin{lstlisting}[
              frame=none,
              language=SQL,
              showspaces=false,
              basicstyle=\ttfamily,
              numbers=none,
              commentstyle=\color{gray}
            ]SELECT DISTINCT S.beer
FROM Serves S, Likes L
WHERE S.bar = 'Edge' AND S.beer = L.beer
  AND L.drinker <> 'Richard'; 
    \end{lstlisting}
}    
 \\\hline
     \end{tabular}
  }
    \caption{\common{Queries used in the user study.}}
    \label{tab:user-study-queries}
    \vspace{-8mm}
\end{table}

\begin{figure}
  \centering
  \includegraphics[width=0.98\linewidth]{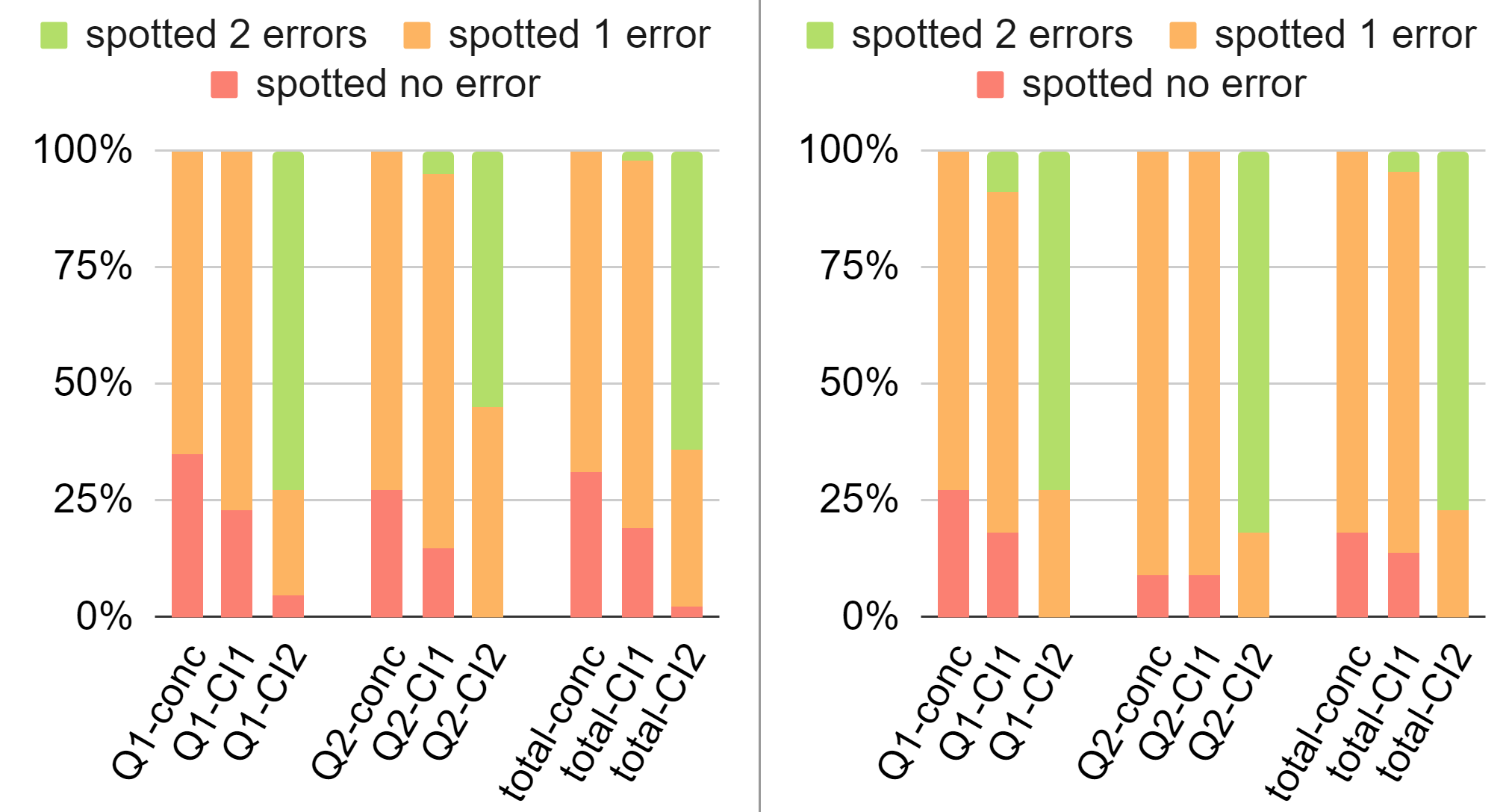}
  \vspace{-5mm}
  \caption{\label{fig:user-study-errors}\small \common{User performance on spotting errors; *-conc: concrete instance only, *-CI1: the first C-Instance, *CI2: the second C-Instance (left: undergrad, right: graduate).}}
\vspace{-3.5mm}
\end{figure}

\begin{figure}
\begin{minipage}{0.42\linewidth}
  \centering
  \includegraphics[width=\linewidth]{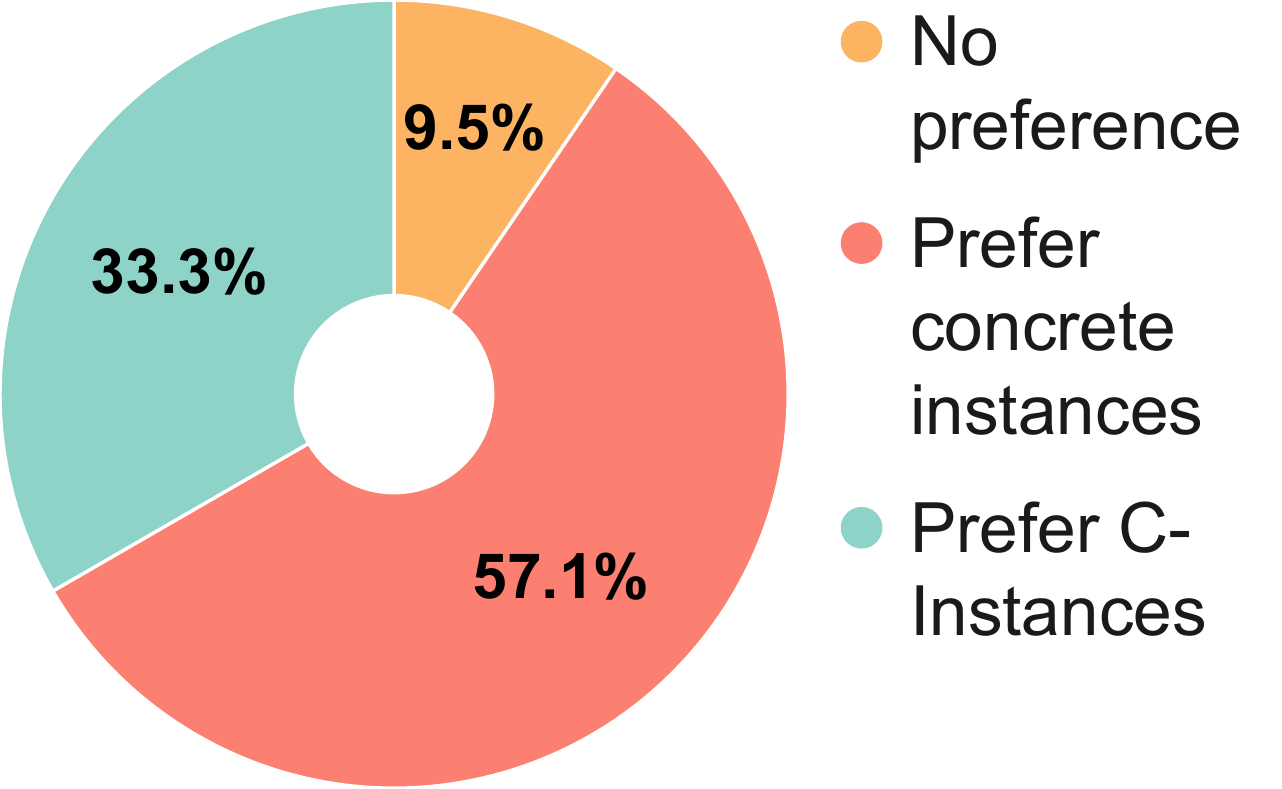}
\end{minipage}
  \hfill
\begin{minipage}{0.42\linewidth}
  \centering
  \includegraphics[width=\linewidth]{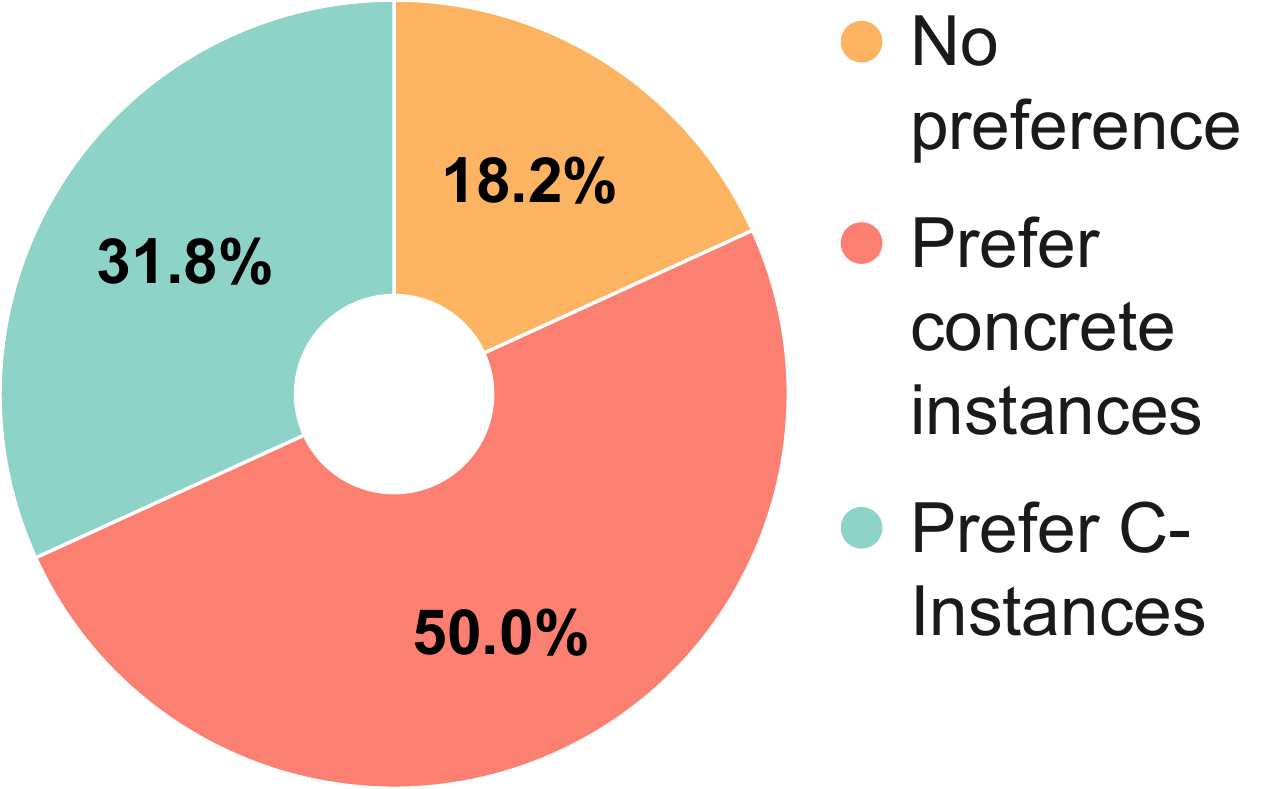}
\end{minipage}
\caption{\label{fig:user-study-preference}\small \common{Preference on explanation types (left: undergraduate, right: graduate)}}
\vspace{-3.5mm}
\end{figure}
\begin{figure}
\begin{minipage}{0.42\linewidth}
  \centering
  \includegraphics[width=\linewidth]{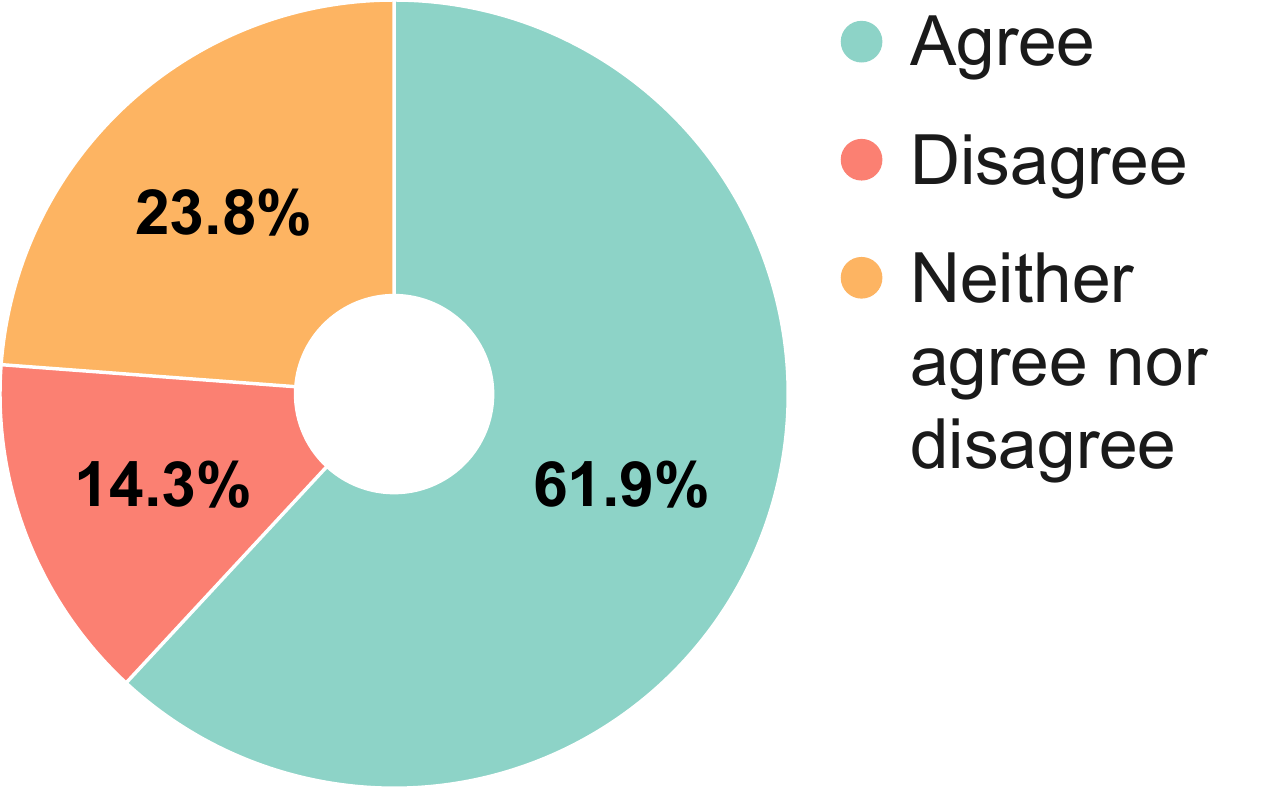}
\end{minipage}
  \hfill
\begin{minipage}{0.42\linewidth}
  \centering
  \includegraphics[width=\linewidth]{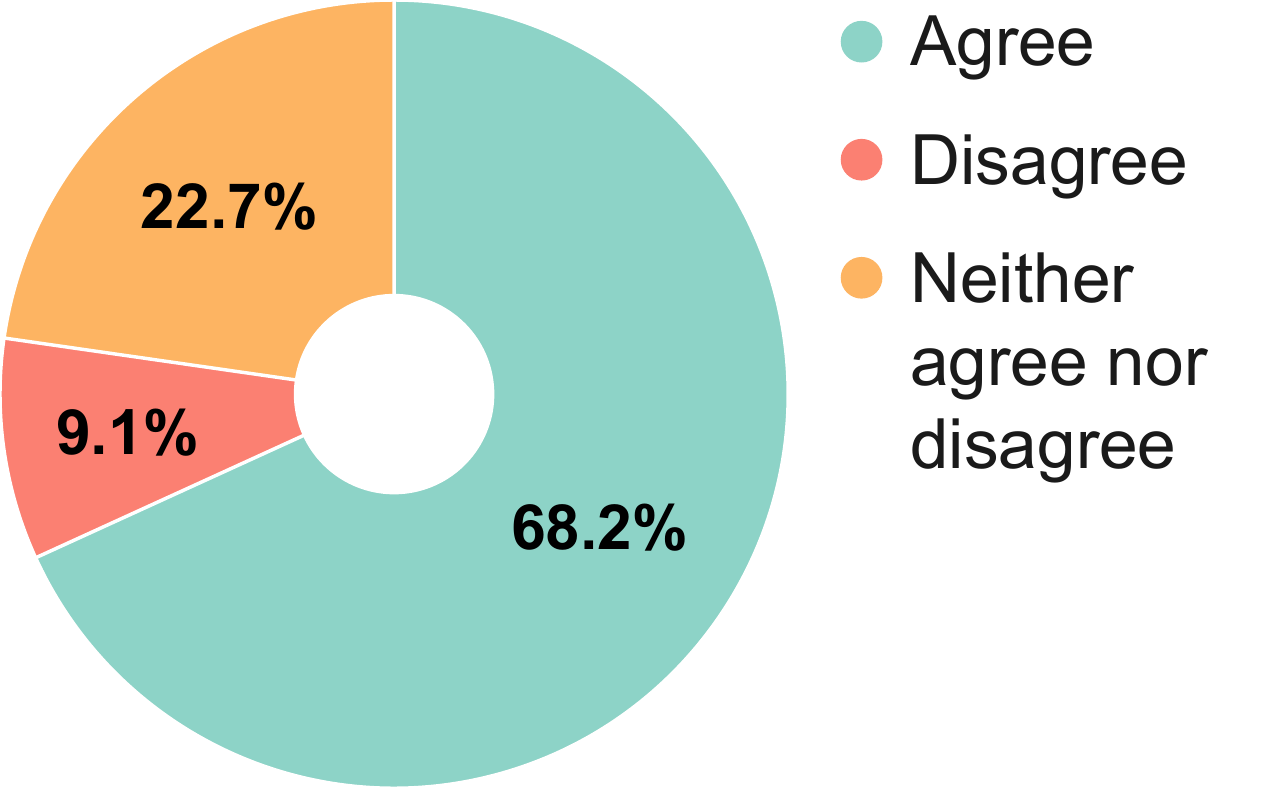}
\end{minipage}
\caption{\label{fig:user-study-second-cins}\small \common{User feedback on ``The second C-Instance provided additional help'' (left: undergraduate, right: graduate)}}
\vspace{-3.5mm}
\end{figure}

\vspace{-1mm}
\section{Conclusions and Future Work}\label{sec:conclusions}
We have defined and studied the 
problem of compact query characterization using the coverage of abstract c-instances. 
We have devised algorithms and optimizations for computing such characterizations building on the concept of chase and utilizing the structure of the syntax tree. 
We experimentally showed that our approach is effective at finding c-instances that characterize the query and examined the effect of query complexity and parameter changes on the scalability of our approach. 
In future work, we plan to study the 
development of further optimizations for finding such solution for more query classes different properties of c-instances.  
In this paper we showed that the problem of finding a universal solution is poly-time for CQ$^\neg$ queries, while the decision version is undecidable for general DRC queries: understanding the computability and complexity for universal solutions for other query classes in between is another interesting research direction. 
\revb{Finally, while our model can support queries with the same final aggregate and different bodies by removing the aggregate, 
extending our model to support arbitrary aggregate queries is another intriguing direction of future work.}

\begin{acks}
This work is supported by the NSF awards IIS-1552538, IIS-1703431, IIS-2008107, IIS-1814493, and by the NIH award R01EB025021.
\end{acks}

\bibliographystyle{ACM-Reference-Format}
\bibliography{main}


\clearpage
\appendix
\section{Appendix}
This appendix is part of the full version of the paper ``Understanding Queries by Conditional Instances''.

\subsection{Omitted Pseudocodes and Proof}
We given the pseudo code of the algorithms that was omitted from Section \ref{sec:complete-algo}. Namely, Algorithms \ref{alg:or-op} and \ref{alg:forall-op} that handle the disjunction connective and universal operator in Algorithm \ref{alg:forall-op} (Lines 11 and 15, respectively). We further give the pseudo code and proof of correctness for the procedure that checks whether a c-instance satisfies a query. This procedure, called \proc{Tree-SAT}, is used in Line 13 of Algorithm \ref{alg:rd-tree-chase-main}.

\begin{algorithm}[h]\caption{Handle-Disjunction}\label{alg:or-op}
  \centering
{\small
\begin{codebox}
  \Procname{$\proc{Handle-Or}(\schemaOf{\rel}, Q, I, f)$}
  
    \li ConjTrees $\gets \proc{Expand-DisjTree}(Q.root)$
    \li \For $T \in ConjTrees$
    \Do
        \li res.extend($\proc{Tree-Chase-BFS}($\schemaOf{\rel}$, T, I, f)$)
    \End
\li \Return $res$
\end{codebox}
}
\end{algorithm}

\begin{algorithm}[h]\caption{Handle-Existential-Operator}\label{alg:exists-op}
  \centering
{\small
\begin{codebox}
  \Procname{$\proc{Handle-Existential}(\schemaOf{\rel}, Q, I, f)$}
  
    \li \For $x \in I.domain(Q.root.variable)$
    \Do
        \li $g \gets f \cup \{Q.root.variable \to x\}$
        \li res.extend($\proc{RD-Tree-Chase-Naive}(\schemaOf{\rel}, Q.root.child, I, g, limit)$)
    \End
    \li Create a fresh labeled null $y$ in the domain of $Q.root.variable$\label{alg:rd-tree-chase:line:exists-add-fresh-start}
    \li $g \gets f \cup \{Q.root.variable \to y\}$
    \li $I.domain(Q.root.variable) \gets I.domain(Q.root.variable) \cup \{y\}$
    \li res.extend($\proc{RD-Tree-Chase-Naive}(\schemaOf{\rel}, Q.root.child, I, g, limit)$)\label{alg:rd-tree-chase:line:exists-add-fresh-end}
\li \Return $res$
\end{codebox}
}
\end{algorithm}

\begin{algorithm}[t]\caption{Handle-Universal}\label{alg:forall-op}
  \centering
{\footnotesize
\begin{codebox}
  \Procname{$\proc{Handle-Universal}(\schemaOf{\rel}, Q, I, h, limit)$}
  
    \li res $\gets []$, $Ilist \gets [I]$
    \li \If $I.domain(Q.root.variable) = \emptyset$
        \Then
            \li $res.append(I)$
    \li \Else
        \li \For $x \in I.domain(Q.root.variable)$
        \Do
            \li $g \gets h \cup \{Q.root.variable \to x\}$
            \li $cur \gets []$
            \li \For $J_1 \in Ilist$
            \Do
                \li $Ilist' \gets \proc{Tree-Chase-BFS}(\schemaOf{\rel}, Q.root.child, J_1, g, limit)$
                \li \For $J_2 \in Ilist'$
                \Do
                   \li \If $\consistent(J_2)$ 
                    \Then
                      \li cur.append($J_2$)
                    \End
                \End
            \End
            \li $Ilist \gets cur$
        \End
        \li $res = res \cup Ilist$ 
    \End
    \li Create a fresh labeled null $y$ in the domain of $Q.root.variable$\label{alg:rd-tree-chase:line:forall-add-fresh-start}
    \li $g \gets h \cup \{Q.root.variable \to y\}$
    \li $I.domain(Q.root.variable) \gets I.domain(Q.root.variable) \cup \{y\}$
    \li $cur \gets []$
    \li \For $J_1 \in Ilist$
    \Do
        \li $Ilist' \gets \proc{Tree-Chase-BFS}(\schemaOf{\rel}, Q.root.child, J_1, g, limit)$
        \li \For $J_2 \in Ilist'$
        \Do
           \li \If $\consistent(J_2)$ 
            \Then
              \li cur.append($J_2$)
            \End
        \End
    \End
    \li $res = res \cup cur$ 
\li \Return $res$
\end{codebox}
}
\end{algorithm}




\mypar{Checking if a c-instance satisfies a query}
The $\proc{Tree-SAT}$ procedure is used in Algorithm \ref{alg:rd-tree-chase-main} (Line 13) to verify that a c-instance satisfies the syntax tree of a query. 
We next describe the pseudo code of the procedure and prove its correctness. The algorithm gets as input the query syntax tree $Q$, the current c-instance $I$, and the current homomorphism $h$. It starts by adding $\exists$ quantifier for every free variable in $Q$ (Lines 1--3) and checks whether $Q$ is a single atom. If it is, the algorithm checks whether this atom is negated or an atomic condition and whether it is contained in the condition of $I$, $\universalcond(I)$, or if the atom is not negated and is mapped to a tuple in $I$. In both cases the algorithm returns True and returns False otherwise (Lines 4--8). If the root of $Q$ is a $\land$ node the algorithm checks recursively whether both of its subtrees are satisfied by $I$ and $h$ (Lines 9--10). 
If the root of $Q$ is a $\lor$ node the algorithm checks recursively whether at least one of its subtrees are satisfied by $I$ and $h$ (Lines 11--12).
If the root of $Q$ is a $\exists$ node the algorithm checks recursively whether there is a mapping for the quantified variable to a labeled null $x$ that satisfies the subtree of the quantifier with $I$ and $h$ extended with this mapping (Lines 13--18).
If the root of $Q$ is a $\forall$ node the algorithm checks recursively whether all mappings to labeled nulls $x$ for the quantified variable satisfy the subtree of the quantifier with $I$ and $h$, extended with these mappings (Lines 19--24).

\begin{algorithm}[t]\caption{Tree-SAT}\label{alg:tree-sat}
{\small
\begin{codebox}
\algorithmicrequire
$Q$: a syntax tree of a DRC query; \\
$I$: current c-instance; \\
$f$: current mapping from $\varset_Q \cup \constset_Q \rightarrow \lnset_I \cup \constset_I$.\\
\algorithmicensure True iff $I$ satisfies $Q$ \\
  \Procname{$\proc{Tree-SAT}(Q, I, f)$}
  \li \If $Q$ has free variables
  \Then
    \li \For $x \in FreeVar(Q)$
      \Do
          \li $Q \gets \exists x Q(x)$
      \End
  \End
    \li \If Q.root is an atom
      \Then 
        \li \If (Q.negated and $\neg f(Q.root.atom) \in \universalcond(I)$) or \\
        (Q = $x \circ y$ and $f(x) \circ f(y) \in \universalcond(I)$) or \\
        (Q.negated = False and $f(Q.root.atom) \in I)$ 
        \Then 
            \li \Return True\label{l:root_in_instance}
            \li \Else
                \li \Return False
            \End
        \End
  \li \If Q.root.operator $\in \{\land\}$
  \Then
    \li \Return \\ $\proc{Tree-SAT}(Q.root.lchild, I, f) \land $ $\proc{Tree-SAT}( Q.root.rchild, I, f)$\label{l:recurse_and}
 \li \ElseIf Q.root.operator $\in \{\lor\}$
  \Then
    \li \Return \\ $\proc{Tree-SAT}(Q.root.lchild, I, f) \lor $ $\proc{Tree-SAT}( Q.root.rchild, I, f)$\label{l:recurse_or}
  \li \ElseIf Q.root.operator $\in \{\exists\}$
  \Then
    \li \For $x \in I.domain(Q.root.variable)$\label{l:iter_demain_exists}
    \Do
        \li $g \gets f \cup \{Q.root.variable \to x\}$
        \li \If $\proc{Tree-SAT}( Q.root.child, I, g)$ 
        \Then
            \li \Return True\label{l:exists_in_instance}
        \End
    \End
    \li \Return False
  \li \ElseIf Q.root.operator $\in \{\forall\}$
  \Then 
    \li \For $x \in I.domain(Q.root.variable)$\label{l:all_mappings_in_instance}
    \Do
        \li $g \gets f \cup \{Q.root.variable \to x\}$
        \li \If $not\ \proc{Tree-SAT}( Q.root.child, I, g)$ 
        \Then
            \li \Return False
        \End
    \End
    \li \Return True
  \End
\end{codebox}
}
\end{algorithm}

\amir{added:}
\begin{proposition}\label{prop:soundness-conj-sat}
Given $I$, and $Q$, Algorithm~\ref{alg:tree-sat} returns True for $Q$ and $I$ iff $I$ satisfies $Q$.
\end{proposition}

\begin{proof}
We prove that $Q$ is satisfied by $I$ with the homomorphism $h$ iff Algorithm~\ref{alg:tree-sat} returns True for some $Q$, $I$, and $h$ using induction over the size of $Q$. 
If $Q$ is of size $1$, then the algorithm returns True in Line \ref{l:root_in_instance} after checking that the only atom in $Q$ has a mapping to the tuple in $I$ if it is not negated, and if it is negated or an atomic condition, it checks whether the global condition of $I$ contains the tuple or atomic condition mapped to the atom in $Q$. 
For the induction hypothesis, assume that for every $Q$ of size $< n$, Algorithm~\ref{alg:tree-sat} returns True for $Q$ and $I$ iff $I$ satisfies $Q$. Now, suppose $Q$ is of size $n$. 

If $Q = Q_1\land Q_2$, 
Algorithm~\ref{alg:tree-sat} returns True in Line \ref{l:recurse_and} iff both $Q_1$ and $Q_2$ are satisfied by $I$ and $h$. 
According to the induction hypothesis, both $Q_1$ and $Q_2$ are satisfied by $I$ and $h$ iff then the algorithm returns True on both. 

If $Q = Q_1\lor Q_2$, 
Algorithm~\ref{alg:tree-sat} returns True in Line \ref{l:recurse_and} iff one of $Q_1$ or $Q_2$ are satisfied by $I$ and $h$. 
According to the induction hypothesis, $Q_1$ or $Q_2$ are satisfied by $I$ and $h$ iff then the algorithm returns True for one of them. 

If $Q = \exists x.~Q'$, 
$Q$ is satisfied by $I$ and $h$ iff there exists $y \in I.domain(Q.root.variable)$, $Q'$ is satisfied by $I$ and $h \cup \{x \gets y\}$. 
In Line \ref{l:iter_demain_exists}, the algorithm iterates over the domain of the quantified variable, adds it to the mapping list and checks whether replacing the quantified variable with its current mapping satisfies the subtree rooted at the child of the $\exists$ quantifier.
Indeed, since the algorithm tries all mappings of $x$ from the domain, it will try $x \gets y$. 
Thus, Algorithm~\ref{alg:tree-sat} will return True iff $Q$ is satisfied by $I$ and $h$. 

If $Q = \forall x.~Q'$, 
$Q$ is satisfied by $I$ and $h$ iff for all\\ 
$y \in I.domain(Q.root.variable)$, $Q'$ is satisfied by $I$ and $h \cup \{x \gets y\}$. 
In Line \ref{l:all_mappings_in_instance}, the algorithm iterates over the domain of the quantified variable, adds it to the mapping list and checks whether replacing the quantified variable with its current mapping of $x$ satisfies $Q'$. If it does not, it returns False. If all mappings of $x$ satisfy $Q'$, the algorithm returns True. 
According to the induction hypothesis, for each such mapping of $x$, 
Indeed, since the algorithm tries all mappings of $x$ from the domain, Algorithm~\ref{alg:tree-sat} will return True iff $Q'$ with this mapping is satisfied by $I$ and $h \cup \{x \gets y\}$. 
Thus, Algorithm~\ref{alg:tree-sat} will return True iff $Q$ is satisfied by $I$ and $h$. 

\end{proof}

\mypar{Checking if a c-instance is consistent}
To guarantee correctness and reduce the search space, Algorithm \ref{alg:rd-tree-chase-main} verifies that a c-instance is consistent every time adding a c-instance to the queue (Line 13, Line 18). 
As defined in Definition~\ref{def:possible-world},
a c-instance $\cinstance$ is consistent if $\Rep(\cinstance) \neq \emptyset$, which is reduced to deciding the satisfiablity of the global condition of $\cinstance$. In our implementation we used the Z3 SMT solver to support complex constraints involving integers, real numbers, and strings.

\subsection{Additional Examples}
We next provide examples for the algorithms described in the paper.

\begin{Example}[Example of Algorithm \ref{alg:rd-tree-chase-main}]
We demonstrate the operation of Algorithm \ref{alg:rd-tree-chase-main} using the query $\qincorrect-\qcorrect$ shown in Figure~\ref{fig:diff-query} with its syntax tree in Figure~\ref{fig:syntax-tree-minus}. Suppose  $limit = 20$ (basically, the $limit$ does not affect the execution), and the algorithm begins with empty $h_0$ and $\cinstance_0$.

At the beginning (Lines 2-5 in Algorithm~\ref{alg:rd-tree-chase-main}), for each free variable in $\qincorrect-\qcorrect$ ($x_1, b_1$), a labeled null is created and added to the domain of $\cinstance_0$ and then to $h_0$. Now we have $\cinstance_0.\dom(\Bar.name)$ $= \{x_1\}$, $\cinstance_0.\dom(\Beer.name)$ $= \{b_1\}$, $h = \{x_1 \to x_1, b_1 \to b_1\}$.

Then we add $\cinstance_0$ to the queue and start the BFS procedure. The current $\cinstance_0$ 
has not been visited (Lines 10-12), and it does not satisfy the query tree and is consistent (Line 13), 
so we call Algorithm~\ref{alg:rd-tree-chase}.

In Algorithm~\ref{alg:rd-tree-chase}, it first goes into the $\exists$ case, and since the domain of drinker names is empty, we can only create new labeled null $d_1$ and add $d_1 \to d_1$ to $g$, and do the same for $p_1$ when we run algorithm~\ref{alg:rd-tree-chase-main} recursively.
The formula under $\exists d_1, p_1 (\cdots)$ has no quantifiers and only conjunction. 
Therefore, we directly obtain a c-instance from its left branch (Lines 2-7 Algorithm~\ref{alg:rd-tree-chase})
with the tuples  $\Likes(d_1, b_1)$ and $\Serves(x_1, b_1, p_1)$, and the condition $d_1 \ \sql{LIKE} \ 'Eve\%'$.

For the next existential quantifier node $\exists x_2$, however, there are two options: we can either map $x_2$ to the existing labeled null $x_1$ created earlier, or we can create a new $x_2$ and add it to the domain of bar names in the instance. It is the same for $\exists p_2$. Hence, we can reach the node $\land$ below $\exists x_2, p_2$ with four different mappings for $x_2$ and $p_2$: $\{x_2 \to x_1, p_2 \to p_1\}$, $\{x_2 \to x_1, p_2 \to p_2\}$, $\{x_2 \to x_2, p_2 \to p_1\}$, $\{x_2 \to x_2, p_2 \to p_2\}$. Again, we obtain c-instances by adding the atoms on the left branch of the $\land$ node. The resulting instances will be inconsistent if we choose to map $p_2$ to $p_1$, because there is an atomic condition $p_1 > p_2$, which will be false if we map $p_1$ and $p_2$ to the same labeled null. 

Continuing, we use the case $\{x_2 \to x_2, p_2 \to p_2\}$ to illustrate. Then, we come to the $\forall d_2$ node. Note that, currently, we have one labeled null $d_1$ in the domain of drinker names, thus Algorithm~\ref{alg:forall-op} will try first mapping $d_2$ to $d_1$ and traverse the next node (Lines 5-7), where $\forall p_3$ would first map $p_3$ to $p_1, p_2$ one by one and goes to the $\lor$ node, where the disjunction algorithm will enumerate all three cases that convert disjunction to conjunction. Assume that for $\{p_3 \to p_1\}$, we negate the right branch of $\lor$ and keep the left branch the same, 
then one possible conjunction under the mapping from the left branch can be $\Likes(d_1, b_1)$ $\land d_2\ \sql{LIKE} \ \text{`Eve\textvisiblespace\%'}$ $\land \Serves(x_1, b_1, p_1)$ (certainly, we cannot have $\neg \Likes(d_1, b_1)$ or $\neg \Serves(x_1, b_1, p_1)$).
We transform the right branch into $\forall x_3 \forall p_4 (\neg \Serves(x_3, b_1, p_4) \lor p_3 \geq p_4)$ to get negation only on the leaves, and we have:
\begin{enumerate}
    \item $\Serves(x_1, b_1, p_1) \land p_1 \geq p_1$ for $\{x_3 \to x_1, p_4 \to p_1\}$
    \item $\neg \Serves(x_2, b_1, p_1) \land p_1 \geq p_1$ for $\{x_3 \to x_2, p_4 \to p_1\}$
    \item $\neg \Serves(x_1, b_1, p_2) \land p_1 \geq p_2$ for $\{x_3 \to x_1, p_4 \to p_2\}$
    \item $\Serves(x_2, b_1, p_2) \land p_1 \geq p_2$ for $\{x_3 \to x_2, p_4 \to p_2\}$
\end{enumerate}

For $\{p_3 \to p_2\}$, we keep both branches the same. One possible conjunction under the mapping from the left branch can be $\Likes(d_1, b_1)$ $\land \neg(d_2\ \sql{LIKE} \ \text{`Eve\textvisiblespace\%'})$ $\land \neg \Serves(x_1, b_1, p_2)$; 
for the right branch, we can map $x_3$ to $x_1$ and $p_4$ to $p_1$ and thus we get $\Serves(x_1, b_1, p_1)$ $\land p_2 < p_1$. By merging the resulting instances in each $\land$ node, we can obtain the c-instance $\cinstance_2$ 
in Figure~\ref{fig:c-instance-2}.

\end{Example}

\begin{figure}[t]
\centering
    \begin{subfigure}[b]{1.0\linewidth}
     \centering
    \begin{tikzpicture}[scale=.7,
    level 1/.style={level distance=0.6cm, sibling distance=3mm},
    level 2/.style={sibling distance=-10mm, level distance=0.6cm}, 
    level 3/.style={level distance=0.8cm,sibling distance=-5mm},
    level 4/.style={level distance=1.0cm,sibling distance=5mm},
    level 5/.style={level distance=1.0cm, sibling distance=-10mm},
    snode/.style = {shape=rectangle, rounded corners, draw, align=center, top color=white, bottom color=blue!20},
    coverednode/.style = {shape=rectangle, rounded corners, draw, align=center, top color=applegreen, bottom color=applegreen!20},
    logic/.style = {shape=rectangle, rounded corners, draw, align=center, top color=white, bottom color=red!70}]
    \Tree
        [.\node[logic,xshift=0mm] {$\boldsymbol{\forall {d_2, p_3}}$}; 
            [.\node[logic,xshift=0mm] {$\boldsymbol{\land}$}; 
                [.\node[logic, xshift=0mm] {$\boldsymbol{\land}$}; 
                    [.\node[logic] {$\boldsymbol{\land}$}; \node[snode]{$\Likes(d_2, b_1)$}; \node[snode]{($d_2$ \sql{LIKE} 'Eve\textvisiblespace\%')};]
                    \node[snode,xshift=0mm] {$\Serves(x_1, b_1, p_3)$};]
                [.\node[logic,xshift=0mm] {$\boldsymbol{\exists {x_3, p_4}}$};
                    [.\node[logic, xshift=0mm] {$\boldsymbol{\land}$};
                        \node[snode,xshift=0mm] {$\Serves(x_3, b_1, p_4)$}; \node[snode,xshift=0mm] {$p_3 < p_4$};
                    ]
                ]
            ]
        ]
    \end{tikzpicture}
\end{subfigure}
\vspace{3mm}

\begin{subfigure}[b]{1\linewidth}
\centering
\begin{tikzpicture}[scale=.7,
level 1/.style={level distance=0.6cm, sibling distance=0mm},
level 2/.style={sibling distance=-10mm, level distance=0.6cm}, 
level 3/.style={level distance=0.8cm,sibling distance=-12mm},
level 4/.style={level distance=1.0cm,sibling distance=7mm},
level 5/.style={level distance=1.0cm, sibling distance=-7mm},
snode/.style = {shape=rectangle, rounded corners, draw, align=center, top color=white, bottom color=blue!20},
coverednode/.style = {shape=rectangle, rounded corners, draw, align=center, top color=applegreen, bottom color=applegreen!20},
logic/.style = {shape=rectangle, rounded corners, draw, align=center, top color=white, bottom color=red!70}]
\Tree
    [.\node[logic,xshift=0mm] {$\boldsymbol{\forall {d_2, p_3}}$}; 
        [.\node[logic,xshift=0mm] {$\boldsymbol{\land}$}; 
            [.\node[logic, xshift=0mm] {$\boldsymbol{\lor}$}; 
                [.\node[logic] {$\boldsymbol{\lor}$}; \node[snode]{$\neg \Likes(d_2, b_1)$}; \node[snode]{$\neg$($d_2$ \sql{LIKE} 'Eve\textvisiblespace\%')};]
                \node[snode,xshift=0mm] {$\neg \Serves(x_1, b_1, p_3)$};]
            [.\node[logic,xshift=0mm] {$\boldsymbol{\forall {x_3, p_4}}$};
                [.\node[logic, xshift=0mm] {$\boldsymbol{\land}$};
                    \node[snode,xshift=0mm] {$\neg \Serves(x_3, b_1, p_4)$}; \node[snode,xshift=0mm] {$p_3 \geq p_4$};
                ]
            ]
        ]
    ]
\end{tikzpicture}
\end{subfigure}
\vspace{-3.5mm}
\caption{Two of the three conjunctive right subtrees of the syntax tree of $\qincorrect-\qcorrect$ depicted in Figure~\ref{fig:syntax-tree-minus}. The trees represent $\neg A \land B$ and $A \land \neg B$, where $A$ ($B$) is the left (right) subtree of the top $\land$ connective.} \label{fig:conjunction-trees}
\end{figure}
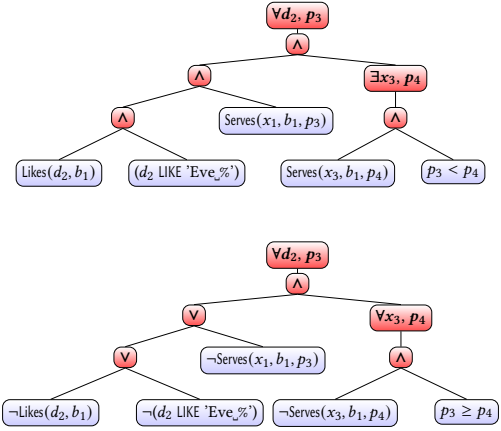

\begin{Example}[Example of the optimized approach in Section \ref{sec:opt-conj-tree-chase}]
A partial result of the algorithm is shown in Figure \ref{fig:conjunction-trees}, showing two out of the three obtained trees for the subtree of the different query $\qincorrect-\qcorrect$ shown in Figure \ref{fig:syntax-tree-minus}. 
The top tree negates the right subtree below the $\lor$ connective, while the bottom negates the left one, representing the two cases where $\neg A \land B$ and $A \land \neg B$. For instance, in the top tree, the formula in the right subtree that was originally $\neg \Serves(x_1, b_1, p_3) \lor (\neg \Likes(d_2, b_1) \land \neg (d_2 \ \sql{LIKE} \ \text{`Eve\textvisiblespace\%'}))$ has been converted into $\Serves(x_1, b_1, p_3) \land (\Likes(d_2, b_1) \land d_2 \ \sql{LIKE} \ \text{`Eve\textvisiblespace\%'})$.
\end{Example}

\vspace{-3mm}
\subsection{Query Details}

\begin{table*}
    \scriptsize
    \begin{tabular}{|c|p{46em}|c|c|c|c|c|}
    \hline
    \textbf{Query description}& \textbf{DRC Query} & Size  & Height &  \# $\forall$+$\exists$ & \# Or & \# Or Below $\forall$ + \# $\forall$    \\\hline
    
    Correct $Q_{1A}$ & $\{(x_{1}, b_{1}) \; \mid \;\exists d_{2} , p_{3} (((Serves(x_{1},b_{1},p_{3}) \land d_{2} \; LIKE \; ``Eve\:\%") \land Likes(d_{2},b_{1})) \land \forall p_{4} , x_{3} (\neg serves(x_{3},b_{1},p_{4}) \lor p_{4} \leq p_{3}))\}$ & 15 & 9 & 10 & 1 & 3 \\\hline
Wrong $Q_{1B}$ & $\{(x_{1}, b_{1}) \; \mid \;\exists d_{1} , p_{1} (((Serves(x_{1},b_{1},p_{1}) \land Likes(d_{1},b_{1})) \land d_{1} \; LIKE \; ``Eve\%") \land \exists x_{2} , p_{2} ((p_{2} < p_{1} \land Serves(x_{2},b_{1},p_{2})) \land x_{1} \; != \; x_{2}))\}$& 17 & 10 & 11 & 0 & 0 \\\hline
$Q_{1A}$-$Q_{1B}$ & $Q_{1A}$-$Q_{1B}$ & 31 & 11 & 20 & 6 & 9 \\\hline
$Q_{1B}$-$Q_{1A}$ & $Q_{1B}$-$Q_{1A}$ & 31 & 11 & 20 & 3 & 3 \\\hline
Correct $Q_{2A}$ & $\{(b_{1}) \; \mid \;\exists tr_1 (Beer(b_{1},tr_1) \land \forall td_1 \neg Likes(td_1,b_{1}))\}$ & 6 & 5 & 4 & 0 & 1 \\\hline
Wrong $Q_{2B}$ & $\{(b_{1})(\exists x_{1} , p_{1} Serves(x_{1},b_{1},p_{1}) \land \neg \exists d_{1} Likes(d_{1},b_{1}))\}$& 7 & 5 & 5 & 0 & 1 \\\hline
$Q_{2A}$-$Q_{2B}$ & $Q_{2A}$-$Q_{2B}$ & 13 & 6 & 9 & 1 & 3 \\\hline
$Q_{2B}$-$Q_{2A}$ & $Q_{2B}$-$Q_{2A}$ & 13 & 6 & 9 & 1 & 3 \\\hline
Correct $Q_{3A}$ & $\{(b_{1}, x_{1}) \; \mid \;\exists tp_{1} (Serves(x_{1},b_{1},tp_{1}) \land \forall tp_{2} , tx_{2} (\neg serves(tx_{2},b_{1},tp_{2}) \lor tp_{2} \leq tp_{1}))\}$ & 10 & 8 & 7 & 1 & 3 \\\hline
Wrong $Q_{3B}$ & $\{(b_{1}, x_{1}) \; \mid \;\exists x_{2} , p_{1} , p_{2} (((Serves(x_{1},b_{1},p_{1}) \land Serves(x_{2},b_{1},p_{2})) \land p_{2} \leq p_{1}) \land x_{1} \; == \; x_{2})\}$& 12 & 9 & 8 & 0 & 0 \\\hline
$Q_{3A}$-$Q_{3B}$ & $Q_{3A}$-$Q_{3B}$ & 21 & 10 & 14 & 4 & 7 \\\hline
$Q_{3B}$-$Q_{3A}$ & $Q_{3B}$-$Q_{3A}$ & 21 & 10 & 14 & 1 & 2 \\\hline
Wrong $Q_{3C}$ & $\{(b_{1}, x_{1}) \; \mid \;\exists r_{1} , p_{1} (Beer(b_{1},r_{1}) \land (Serves(x_{1},b_{1},p_{1}) \land \neg \exists x_{2} , p_{2} (Serves(x_{2},b_{1},p_{2}) \land p_{1} \; < \; p_{2})))\}$& 13 & 10 & 9 & 1 & 3 \\\hline
$Q_{3A}$-$Q_{3C}$ & $Q_{3A}$-$Q_{3C}$ & 22 & 11 & 15 & 3 & 6 \\\hline
$Q_{3C}$-$Q_{3A}$ & $Q_{3C}$-$Q_{3A}$ & 22 & 11 & 15 & 2 & 5 \\\hline
Correct $Q_{4A}$ & $\{(d_{1}) \; \mid \;\exists ta_{1} (Drinker(d_{1},ta_{1}) \land \neg \exists tx_{1} , tt_{1} (Frequents(d_{1},tx_{1},tt_{1}) \land \neg \exists tb_{1} , tp_{1} (Likes(d_{1},tb_{1}) \land Serves(tx_{1},tb_{1},tp_{1}))))\}$ & 13 & 10 & 9 & 1 & 3 \\\hline
Wrong $Q_{4B}$ & $\{(d_{1}) \; \mid \;\exists x_{1} , b_{1} (\exists p_{1} , t_{1} (Frequents(d_{1},x_{1},t_{1}) \land Serves(x_{1},b_{1},p_{1})) \land Likes(d_{1},b_{1}))\}$& 10 & 8 & 7 & 0 & 0 \\\hline
$Q_{4A}$-$Q_{4B}$ & $Q_{4A}$-$Q_{4B}$ & 23 & 11 & 16 & 3 & 9 \\\hline
$Q_{4B}$-$Q_{4A}$ & $Q_{4B}$-$Q_{4A}$ & 23 & 11 & 16 & 2 & 5 \\\hline
Wrong $Q_{4C}$ & $\{(d_{1}) \; \mid \;\exists x_{1} (\exists t_{1} Frequents(d_{1},x_{1},t_{1}) \land \neg (\exists t_{2} Frequents(d_{1},x_{1},t_{2}) \land \neg \exists b_{1} , p_{1} (Likes(d_{1},b_{1}) \land Serves(x_{1},b_{1},p_{1}))))\}$& 13 & 8 & 9 & 1 & 1 \\\hline
$Q_{4A}$-$Q_{4C}$ & $Q_{4A}$-$Q_{4C}$ & 26 & 11 & 18 & 3 & 9 \\\hline
$Q_{4C}$-$Q_{4A}$ & $Q_{4C}$-$Q_{4A}$ & 26 & 11 & 18 & 3 & 6 \\\hline
Wrong $Q_{4D}$ & $\{(d_{1})(\exists a_{1} Drinker(d_{1},a_{1}) \land \neg \exists b_{1} (\exists x_{1} , t_{1} , p_{1} (Frequents(d_{1},x_{1},t_{1}) \land Serves(x_{1},b_{1},p_{1})) \land \neg Likes(d_{1},b_{1})))\}$& 13 & 9 & 9 & 2 & 6 \\\hline
$Q_{4A}$-$Q_{4D}$ & $Q_{4A}$-$Q_{4D}$ & 26 & 11 & 18 & 2 & 4 \\\hline
$Q_{4D}$-$Q_{4A}$ & $Q_{4D}$-$Q_{4A}$ & 26 & 11 & 18 & 4 & 11 \\\hline
Correct $Q_{5A}$ & $\{(d_{1}) \; \mid \;\exists ta_{1} (Drinker(d_{1},ta_{1}) \land \neg \exists tx_{1} (\exists tb_{1} , tp_{1} (Likes(d_{1},tb_{1}) \land Serves(tx_{1},tb_{1},tp_{1})) \land \neg \exists tt_{1} Frequents(d_{1},tx_{1},tt_{1})))\}$ & 13 & 9 & 9 & 2 & 5 \\\hline
Wrong $Q_{5B}$ & $\{(d_{1}) \; \mid \;\exists x_{1} , t_{1} (Frequents(d_{1},x_{1},t_{1}) \land \neg \exists x_{2} (\exists b_{1} , p_{1} (Likes(d_{1},b_{1}) \land Serves(x_{2},b_{1},p_{1})) \land \exists t_{2} \neg Frequents(d_{1},x_{2},t_{2})))\}$& 14 & 10 & 10 & 2 & 6 \\\hline
$Q_{5A}$-$Q_{5B}$ & $Q_{5A}$-$Q_{5B}$ & 27 & 11 & 19 & 3 & 8 \\\hline
$Q_{5B}$-$Q_{5A}$ & $Q_{5B}$-$Q_{5A}$ & 27 & 11 & 19 & 3 & 9 \\\hline
Wrong $Q_{5C}$ & $\{(d_{1})(\exists b_{1} , x_{1} , t_{1} , p_{1} ((Frequents(d_{1},x_{1},t_{1}) \land Serves(x_{1},b_{1},p_{1})) \land Likes(d_{1},b_{1})) \land \neg \exists x_{2} , b_{2} (\exists p_{2} (Likes(d_{1},b_{2}) \land Serves(x_{2},b_{2},p_{2})) \land \neg \exists p_{3} , t_{2} ((Frequents(d_{1},x_{2},t_{2}) \land Serves(x_{2},b_{2},p_{3})) \land Likes(d_{1},b_{2}))))\}$& 25 & 10 & 17 & 2 & 5 \\\hline
$Q_{5A}$-$Q_{5C}$ & $Q_{5A}$-$Q_{5C}$ & 38 & 11 & 26 & 7 & 13 \\\hline
$Q_{5C}$-$Q_{5A}$ & $Q_{5C}$-$Q_{5A}$ & 38 & 11 & 26 & 3 & 8 \\\hline
Wrong $Q_{5D}$ & $\{(d_{1})(\exists b_{1} , x_{1} , p_{1} (Likes(d_{1},b_{1}) \land Serves(x_{1},b_{1},p_{1})) \land \neg \exists x_{2} (\exists b_{2} , p_{2} (Likes(d_{1},b_{2}) \land Serves(x_{2},b_{2},p_{2})) \land \neg \exists t_{1} Frequents(d_{1},x_{2},t_{1})))\}$& 17 & 8 & 12 & 2 & 5 \\\hline
$Q_{5A}$-$Q_{5D}$ & $Q_{5A}$-$Q_{5D}$ & 30 & 10 & 21 & 4 & 10 \\\hline
$Q_{5D}$-$Q_{5A}$ & $Q_{5D}$-$Q_{5A}$ & 30 & 10 & 21 & 3 & 8 \\\hline

    \end{tabular}
    \caption{Beers queries used in the experiments and their complexity measures.}
    \label{experiment-queries-beers}
\end{table*}

Tables~\ref{experiment-queries-beers} and \ref{experiment-queries-tpch} show the queries that we used in our experiments and their complexity. Note that in order to run our algorithms on TPC-H queries, we modified the original queries and rewrote them in DRC.

\begin{table*}[t]
  { 
  \centering
    \scriptsize
    \begin{tabular}{|c|p{46em}|c|c|c|c|c|}
    \hline
    \textbf{Query description}& \textbf{DRC Query} & Size  & Height &  \# $\forall$+$\exists$ & \# Or & \# Or Below $\forall$ + \# $\forall$    \\\hline
 Correct $Q_{4A}$ & $\{(o_{1}, o_{2})(\exists o_{3} , o_{6} (orders(o_{1},o_{3},*,*,o_{6},o_{2},*,*,*) \land (19930701 \leq o_{6} \land o_{6} \; < \; 19931001)) \land \exists l_{2} , l_{3} , l_{12} , l_{13} (lineitem(o_{1},l_{2},l_{3},*,*,*,*,*,*,*,*,l_{12},l_{13},*,*,*) \land l_{12} \; < \; l_{13}))\}$ & 17 & 9 & 12 & 0 & 0 \\\hline
Wrong $Q_{4B}$ & $\{(o_{1}, o_{2})(\exists o_{3} , o_{6} (orders(o_{1},o_{3},*,*,o_{6},o_{2},*,*,*) \land (19930701 \leq o_{6} \land o_{6} \; < \; 19931001)) \land \exists l_{2} , l_{3} , l_{12} , l_{13} (lineitem(o_{1},l_{2},l_{3},*,*,*,*,*,*,*,*,l_{12},l_{13},*,*,*) \land l_{13} < l_{12}))\}$& 17 & 9 & 12 & 0 & 0 \\\hline
$Q_{4A}$-$Q_{4B}$ & $Q_{4A}$-$Q_{4B}$ & 33 & 10 & 23 & 4 & 8 \\\hline
$Q_{4B}$-$Q_{4A}$ & $Q_{4B}$-$Q_{4A}$ & 33 & 10 & 23 & 4 & 8 \\\hline
Wrong $Q_{4C}$ & $\{(o_{1}, o_{2})(\exists o_{3} , o_{6} (orders(o_{1},o_{3},*,*,o_{6},o_{2},*,*,*) \land (19930701 \leq o_{6} \land o_{6} \; < \; 19931001)) \land \neg \exists l_{2} , l_{3} , l_{12} , l_{13} (lineitem(o_{1},l_{2},l_{3},*,*,*,*,*,*,*,*,l_{12},l_{13},*,*,*) \land l_{12} \; < \; l_{13}))\}$& 17 & 9 & 12 & 1 & 5 \\\hline
$Q_{4A}$-$Q_{4C}$ & $Q_{4A}$-$Q_{4C}$ & 33 & 10 & 23 & 3 & 3 \\\hline
$Q_{4C}$-$Q_{4A}$ & $Q_{4C}$-$Q_{4A}$ & 33 & 10 & 23 & 5 & 13 \\\hline
Correct $Q_{16A}$ & $\{(p_{4}, p_{5}, p_{6}, ps_{2}) \; \mid \;\exists p_{1} (\exists p_{2} ((part(p_{1},p_{2},*,p_{4},p_{5},p_{6},*,*,*) \land (49 \; == \; p_{6} \lor 14 \; == \; p_{6})) \land (``Brand\#45" \; != \; p_{4} \land p_{5} \; LIKE \; ``MEDIUM\:POLISHED\%")) \land (partsupp(p_{1},ps_{2},*,*,*) \land \neg \exists s_{7} (supplier(ps_{2},*,*,*,*,*,s_{7}) \land s_{7} \; LIKE \; ``\%complain,")))\}$ & 22 & 11 & 14 & 2 & 2 \\\hline
Wrong $Q_{16B}$ & $\{(p_{4}, p_{5}, p_{6}, ps_{2}) \; \mid \;\exists p_{1} (\exists p_{2} ((part(p_{1},p_{2},*,p_{4},p_{5},p_{6},*,*,*) \land (49 \; == \; p_{6} \lor 14 \; == \; p_{6})) \land (``Brand\#45" \; != \; p_{4} \land p_{5} \; LIKE \; ``MEDIUM\:POLISHED\%")) \land (partsupp(p_{1},ps_{2},*,*,*) \land \neg \exists s_{7} (supplier(ps_{2},*,*,*,*,*,s_{7}) \land s_{7} \; LIKE \; ``\%complain")))\}$& 22 & 11 & 14 & 2 & 2 \\\hline
$Q_{16A}$-$Q_{16B}$ & $Q_{16A}$-$Q_{16B}$ & 41 & 12 & 25 & 7 & 6 \\\hline
$Q_{16B}$-$Q_{16A}$ & $Q_{16B}$-$Q_{16A}$ & 41 & 12 & 25 & 7 & 6 \\\hline
Wrong $Q_{16C}$ & $\{(p_{4}, p_{5}, p_{6}, ps_{2}) \; \mid \;\exists p_{1} (\exists p_{2} ((part(p_{1},p_{2},*,p_{4},p_{5},p_{6},*,*,*) \land (49 \; == \; p_{6} \lor 14 \; == \; p_{6})) \land (``Brand\#45" \; != \; p_{4} \land p_{5} \; LIKE \; ``MEDIUM\:POLISHED\%")) \land (partsupp(p_{1},ps_{2},*,*,*) \land \exists s_{7} (supplier(ps_{2},*,*,*,*,*,s_{7}) \land \neg s_{7} \; LIKE \; ``\%complain,")))\}$& 22 & 11 & 14 & 1 & 0 \\\hline
$Q_{16A}$-$Q_{16C}$ & $Q_{16A}$-$Q_{16C}$ & 41 & 12 & 25 & 8 & 8 \\\hline
$Q_{16C}$-$Q_{16A}$ & $Q_{16C}$-$Q_{16A}$ & 41 & 12 & 25 & 6 & 4 \\\hline
Correct $Q_{19A}$ & 
$\begin{aligned}
&\{(l_{6}, l_{7}) \; \mid \;\exists l_{1} , l_{2} , l_{4} , l_{5}, l_{15} , p_{4} , p_{6} , p_{7}\\ &((lineitem(l_{1},l_{2},*,l_{4},l_{5},l_{6},l_{7},*,*,*,*,*,*,``DELIVER\:IN\:PERSON",``AIR",*) \land\\ &part(l_{2},*,*,p_{4},*,p_{6},p_{7},*,*)) \land (((``Brand\#12" \; == \; p_{4} \land p_{7} \; LIKE \; ``SM\%") \land 
(l_{5} \; <= \; 11 \land p_{6} \; <= \; 5)) \\
&\lor ((``Brand\#23" \; == \; p_{4} \land p_{7} \; LIKE \; ``MED\%") \land ((10 \leq l_{5} \land l_{5} \; <= \; 20) \land p_{6} \; <= \; 10))))
\}
\end{aligned}
$ & 31 & 16 & 20 & 1 & 0 \\\hline
Wrong $Q_{19B}$ & $
\begin{aligned}
&\{(l_{6}, l_{7}) \; \mid \;\exists l_{1} , l_{2} , l_{4} , l_{5} , l_{15} , p_{4} , p_{6} , p_{7}\\ &((lineitem(l_{1},l_{2},*,l_{4},l_{5},l_{6},l_{7},*,*,*,*,*,*,``DELIVER\:IN\:PERSON",``AIR",*) \land\\ &part(l_{2},*,*,p_{4},*,p_{6},p_{7},*,*)) \land (((``Brand\#12" \; == \; p_{4} \land p_{7} \; LIKE \; ``SM\%") \land (l_{5} \; <= \; 10 \land p_{6} \; <= \; 5))\\
&\lor ((``Brand\#234" \; == \; p_{4} \land p_{7} \; LIKE \; ``MED\%") \land (l_{5} \; <= \; 20 \land p_{6} \; <= \; 10))))\}
\end{aligned}$& 29 & 15 & 19 & 1 & 0 \\\hline
$Q_{19A}$-$Q_{19B}$ & $Q_{19A}$-$Q_{19B}$ & 59 & 17 & 38 & 9 & 9 \\\hline
$Q_{19B}$-$Q_{19A}$ & $Q_{19B}$-$Q_{19A}$ & 59 & 17 & 38 & 10 & 9 \\\hline
Wrong $Q_{19C}$ & $
\begin{aligned}
&\{(l_{6}, l_{7}) \; \mid \;\exists l_{1} , l_{2} , l_{4} , l_{5} , l_{15} , p_{4} , p_{6} , p_{7}\\
&((lineitem(l_{1},l_{2},*,l_{4},l_{5},l_{6},l_{7},*,*,*,*,*,*,``DELIVER\:IN\:PERSON",``AIR",*) \land\\
&part(l_{2},*,*,p_{4},*,p_{6},p_{7},*,*)) \land ((``Brand\#12" \; == \; p_{4} \land p_{7} \; LIKE \; ``SM\%") \land\\
&(l_{5} \; <= \; 11 \land p_{6} \; <= \; 5)))\}
\end{aligned}$& 21 & 14 & 15 & 0 & 0 \\\hline
$Q_{19A}$-$Q_{19C}$ & $Q_{19A}$-$Q_{19C}$ & 51 & 17 & 34 & 6 & 9 \\\hline
$Q_{19C}$-$Q_{19A}$ & $Q_{19C}$-$Q_{19A}$ & 51 & 17 & 34 & 9 & 9 \\\hline
Correct $Q_{21A}$ & $
\begin{aligned}
&\{(s_{1}, s_{2}, o_{1})((\exists l_{12} , l_{13} (lineitem(o_{1},*,s_{1},*,*,*,*,*,*,*,*,l_{12},l_{13},*,*,*) \land l_{12} \; < \; l_{13}) \land \\
&\exists ll_{3} , ll_{12} , ll_{13} (lineitem(o_{1},*,ll_{3},*,*,*,*,*,*,*,*,ll_{12},ll_{13},*,*,*) \land ll_{3} \; != \; s_{1})) \land\\ &((orders(o_{1},*,``F",*,*,*,*,*,*) \land \exists s_{4} (supplier(s_{1},s_{2},*,s_{4},*,*,*) \land\\ 
&nation(s_{4},``SAUDI\:ARABIA",*,*))) \land \neg \exists lll_{3} , lll_{12} , lll_{13}\\ &(lineitem(o_{1},*,lll_{3},*,*,*,*,*,*,*,*,lll_{12},lll_{13},*,*,*) \land (lll_{12} \; < \; lll_{13} \land lll_{3} \; != \; s_{1}))))\}
\end{aligned}$ & 31 & 11 & 21 & 2 & 4 \\\hline
Wrong $Q_{21B}$ & $
\begin{aligned}
&\{(s_{1}, s_{2}, o_{1})((\exists l_{12} , l_{13} (lineitem(o_{1},*,s_{1},*,*,*,*,*,*,*,*,l_{12},l_{13},*,*,*) \land l_{12} \; < \; l_{13}) \land\\ 
&(orders(o_{1},*,``F",*,*,*,*,*,*) \land \exists s_{4} (supplier(s_{1},s_{2},*,s_{4},*,*,*) \land\\ &nation(s_{4},``SAUDI\:ARABIA",*,*)))) \land \exists lll_{3} , lll_{12} , lll_{13}\\
&(lineitem(o_{1},*,lll_{3},*,*,*,*,*,*,*,*,lll_{12},lll_{13},*,*,*) \land (lll_{13} \leq lll_{12} \land lll_{3} \; != \; s_{1})))\}
\end{aligned}$& 24 & 10 & 16 & 0 & 0 \\\hline
$Q_{21A}$-$Q_{21B}$ & $Q_{21A}$-$Q_{21B}$ & 53 & 12 & 35 & 9 & 13 \\\hline
$Q_{21B}$-$Q_{21A}$ & $Q_{21B}$-$Q_{21A}$ & 53 & 12 & 35 & 7 & 9 \\\hline
Wrong $Q_{21C}$ & $\{(s_{1}, s_{2}, o_{1})(\exists l_{12} , l_{13} (lineitem(o_{1},*,s_{1},*,*,*,*,*,*,*,*,l_{12},l_{13},*,*,*) \land l_{12} \; < \; l_{13}) \land (\exists o_{3} orders(o_{1},*,o_{3},*,*,*,*,*,*) \land \exists s_{4} (supplier(s_{1},s_{2},*,s_{4},*,*,*) \land nation(s_{4},``SAUDI\:ARABIA",*,*))))\}$& 16 & 8 & 11 & 0 & 0 \\\hline
$Q_{21A}$-$Q_{21C}$ & $Q_{21A}$-$Q_{21C}$ & 45 & 12 & 30 & 6 & 10 \\\hline
$Q_{21C}$-$Q_{21A}$ & $Q_{21C}$-$Q_{21A}$ & 45 & 12 & 30 & 7 & 9 \\\hline
\end{tabular}
\caption{TPC-H queries used in the experiments and their complexity measures. * denotes ``don't care''.}
\label{experiment-queries-tpch}
    }
\end{table*}


\end{document}